\documentclass[10pt]{article}
\usepackage{amsthm,amsfonts,amsmath,amssymb}
\usepackage{euscript}
\usepackage{hyperref}
\usepackage{graphics}
\usepackage{epstopdf} 
\usepackage{xypic}
\usepackage[all,arrow,matrix]{xy}

\usepackage{graphicx}

\newcommand\void[1]       {}

\newtheorem{thm}{Theorem}
\newtheorem{defn}[thm]{Definition}
\newtheorem{prop}[thm]{Proposition}
\newtheorem{cor}[thm]{Corollary}
\newtheorem{rema}[thm]{Remark}
\newtheorem{lemma}[thm]{Lemma}

\newtheorem{conj}[thm]{Conjecture}
\newtheorem{expl}[thm]{Example}

\numberwithin{equation}{section}
\numberwithin{thm}{section}

\newcommand\nn             {\nonumber \\}

\newcommand\be            {\begin{equation}}
\newcommand\ee            {\end{equation}}
\newcommand\bea           {\begin{eqnarray}}
\newcommand\eea         {\end{eqnarray}}
\newcommand\bnu          {\begin{enumerate}}
\newcommand\enu          {\end{enumerate}}

\newcommand\BF           {\mathrm{BF}}

\newcommand\TO           {\EuScript{TO}}
\newcommand\CHom      {\mathcal{H}\mathrm{om}}
\newcommand\Hom           {\mathrm{Hom}}

\newcommand\id            {\mathrm{id}}
\newcommand\op          {\mathrm{op}}

\newcommand\ev          {\mathrm{ev}}

\newcommand\hilb   {\EuScript{H}\mathrm{ilb}}
\newcommand\vect    {\EuScript{V}\mathrm{ect}}
\newcommand\rep     {\EuScript{R}\mathrm{ep}}
\newcommand\fun     {\EuScript{F}\mathrm{un}}
\newcommand\alg     {\mathrm{Alg}}

\newcommand\prebf  {\mathrm{BF}^{pre}}

\newcommand\one    {\mathbf{1}}
\newcommand\tf      {\vec{f}}
\newcommand{\rev} {\mathrm{rev}}

\newcommand{\bulk}   {\underline{\sl bulk} }

\newcommand{\pf}{\begin{proof}}
\newcommand{\epf}{\end{proof}}

\newcommand\Cb            {\mathbb{C}}

\newcommand\Zb            {\mathbb{Z}}

\newcommand\EA           {\EuScript{A}}
\newcommand\EB           {\EuScript{B}}
\newcommand\EC           {\EuScript{C}}
\newcommand\ED           {\EuScript{D}}
\newcommand\EE          {\EuScript{E}}
\newcommand\EF          {\EuScript{F}}
\newcommand\EG         {\EuScript{G}}

\newcommand\EM          {\EuScript{M}}
\newcommand\EN         {\EuScript{N}}

\newcommand\EX         {\EuScript{X}}
\newcommand\EY         {\EuScript{Y}}
\newcommand\EZ         {\EuScript{Z}}

\newcommand\Z           {\mathcal{Z}}

\begin{document}

\begin{center} \LARGE
Boundary-bulk relation for topological orders \\ 
as the functor mapping higher categories to their centers
\\
\end{center}

\vskip 2em
\begin{center}
{\large
Liang Kong$^{a,b,c}$,\,  
Xiao-Gang Wen$^{d,e}$,
Hao Zheng$^{f}$\,
~\footnote{Emails:
{\tt  kong.fan.liang@gmail.com, wen@dao.mit.edu, hzheng@math.pku.edu.cn}}}
\\[2em]
$^a$ Department of Mathematics and Statistics\\
University of New Hampshire, Durham, NH 03824, USA
\\[1em]
$^b$ Center of Mathematical Sciences and Applications\\ 
Harvard University, Cambridge, MA 02138, USA
\\[1em] 
$^c$ Institute for Advanced Study (Science Hall) \\
Tsinghua University, Beijing 100084, China
\\[1em]
$^d$ Department of Physics, Massachusetts Institute of Technology, \\
Cambridge, Massachusetts 02139, USA
\\[1em]
$^e$ Perimeter Institute for Theoretical Physics, \\
Waterloo, Ontario, N2L 2Y5 Canada
\\[1em]
$^f$ Department of Mathematics, Peking University,\\
Beijing, 100871, China
\end{center}

\vskip 3em

\begin{abstract}
In this paper, we study the relation between topological orders and their
gapped boundaries.  We propose that the bulk for a given gapped boundary theory
is unique.  It is actually a consequence of a microscopic definition of a local topological order, which is a 
(potentially anomalous) topological order defined on an open disk. 
Using this uniqueness, we show that the notion of ``bulk'' is equivalent to the notion of center in
mathematics. We achieve this by first introducing the notion of a morphism
between two local topological orders of the same dimension, then proving that
the bulk satisfying the same universal property as that of the center in
mathematics. We propose a classification (formulated as a macroscopic definition) of $n+$1D local topological orders by unitary multi-fusion $n$-categories, 
and explain that the notion of a morphism between two local topological orders 
is compatible with that of a unitary monoidal $n$-functor in a few low dimensional cases. 
We also explain in some low dimensional cases that this classification is compatible with the result of ``bulk = center". In the end, we explain that above boundary-bulk relation is only the first layer of a
hierarchical structure which can be summarized by the functoriality of the bulk
(or center). This functoriality also provides the physical meanings of some well-known mathematical results on fusion 1-categories. This work can also be
viewed as the first step towards a systematic study of the category of local
topological orders, and the boundary-bulk relation actually provides a useful
tool for this study.



\end{abstract}

\newpage

\tableofcontents

\section{Introduction}


In this work, we study the boundary-bulk relation of local topological orders in any
dimensions via higher categories. We avoid a lengthy introduction to
topological orders \cite{wen89}. Instead, we direct readers to many online
resources for this vast and fast-growing topics (see for example a long list of
physics references in \cite{kong-wen}). A local topological order is a potentially anomalous  
topological order on the open $n$-disk $D^n$ as the space manifold\footnote{It can be viewed as a TQFT restricted to an open disk in the sense of Morrison-Walker \cite{mw,walker}.} (see Def.\,\ref{def:TO} and Remark\,\ref{rema:global-TO}, \ref{rema:spins}). Since we only study local topological orders in this work, unless specified otherwise, we abbreviate ``a local topological order" to ``a topological order" throughout this work without further announcement. 


\medskip
There are at least two types of problems to study in the field of topological
orders. The first type is to construct and classify all topological orders.
The second type is to study relations
among all topological orders, such as phase transitions. To study all these
relations among all topological orders as a whole amounts to study the category
$\TO$ of (simple) topological orders.  The topological orders of the same space-time
dimension $n$, or $n$D topological orders, form an obvious subcategory of
$\TO$, denoted by $\TO_n$.  
There are many natural but in-equivalent definitions of $\TO_n$. 
In this work, three different and in-equivalent definitions $\TO_n^{closed-wall}$, $\TO_n^{fun}$ and $\TO_n^{wall}$ appear. 

\medskip 
One type of relation that is important in physics is the relation between the boundary physics and the bulk
physics. Mathematically, the boundary-bulk relation for topological orders with gapped boundaries is a relation between $\TO_n$ and $\TO_{n+1}$ (in some sense a functor $\TO_n \to \TO_{n+1}$ see
Sec.\,\ref{sec:Z-functor}). All three definitions of $\TO_n$ play interesting roles in this relation. This work can also be viewed as the first step towards a systematic study of the category $\TO_n$. The boundary-bulk relation actually provides a useful tool for this study.

The topological excitations on gapped boundaries were first studied in the 2+1D toric code model (\cite{kitaev1}) by Bravyi and Kitaev in \cite{bravyi-kitaev}. It was latter generalized to Levin-Wen models (\cite{lw-mod}) with gapped boundaries in \cite{kitaev-kong}, where the topological excitations on a boundary of such a lattice
model were shown to form a unitary fusion 1-category $\EC$, and those in the bulk form a
unitary braided fusion 1-category which is given by the Drinfeld center
$Z_2(\EC)$ of $\EC$ (see also \cite{lan-wen}). But this work does not address the
question if the bulk theory for a given boundary theory is unique. In
\cite{fsv}, it was shown model-independently that among all possible bulks associated to the same boundary theory $\EC$, the bulk given by $Z_2(\EC)$ is the universal one (a terminal object).
One way to complete the proof of the uniqueness of the bulk is to view the
gapped boundary as a consequence of anyon condensation \cite{bs} of a given bulk theory $\ED$ to the trivial phase. This idea leads to a classification of gapped boundaries for abelian topological theories \cite{kapustin-saulina,levin,bjq}. This result, together with results in \cite{fsv}, implies the uniqueness of the bulk for abelian topological theories. The proof for general 2+1D topological orders appeared in the mathematical theory of anyon condensation developed in \cite{kong-anyon}, in
which it was shown that such a condensation is determined by a Lagrangian
algebra $A$ in $\ED$, and $\EC$ is monoidally equivalent to the category
$\ED_A$ of $A$-modules in $\ED$.\footnote{In \cite{kong-anyon}, the bootstrap
analysis shows that $\EC$ must be a sub-fusion-category of $\ED_A$, and the
boundary-bulk relation was used as a supporting evidence for $\EC=\ED_A$. But it
is just physically natural that all possible quasi-particles that can be
confined on the boundary, i.e. objects in $\ED_A$, should all survive on the
boundary after the condensation. In other words, $\EC=\ED_A$ is a natural
physical requirement.} Moreover, we have $Z_2(\ED_A)\simeq \ED$. This completes
the proof of a part of the {\it bulk-boundary relation} in 3D, which says that the 3D
bulk theory $\ED$ of a given 2D boundary $\EC$ is unique and given by the
Drinfeld center of $\EC$.

\smallskip Does this bulk-boundary relation hold in higher dimensions?  The
gapped boundary of a non-trivial topological order should be viewed as an
anomalous topological order. Following \cite{chen-gu-wen}, a microscopic
definition of anomalous topological order was proposed in \cite{kong-wen} (see
a refined version in Def.\,\ref{def:TO}). It follows immediately from this
definition that the bulk of a given gapped boundary must be unique
\cite[Sec.VII.C]{kong-wen}. But this result is highly non-trivial from a
macroscopic point of view. Macroscopically, 
up to invertible topological orders \cite{kong-wen}, the local order parameters of a
topological state of matter are given by the fusion-braiding (and spins) properties of its topological excitations in the open disk. 
A complete set of these local order parameters uniquely determines the universal class of the topological state up to invertible topological orders. Therefore, all such sets should give a classification of all topological orders up to invertible topological orders. Since higher categories can encode the fusion-braiding properties of excitations in an efficient way \cite{kong-wen}, this idea leads us to propose a classification of $n+$1D topological orders by unitary multi-fusion $n$-categories (see Def.\,\ref{def:unitary-fusion-n-cat}). For convenience, we formulate this classification proposal as a macroscopic definition of an $n+$1D topological order (see Def.\,\ref{def:cat-TO}). Whether the microscopic definition is equivalent
to the macroscopic one is a highly non-trivial and important open
problem. For this reason, we would like to refer to the uniqueness of the bulk
for a given gapped boundary as the {\it unique-bulk hypothesis}, and the unique
bulk is denoted by \bulk. One of the main goals of this work is to prove, assuming the
unique-bulk hypothesis, that the \bulk of a given gapped boundary is the center of the boundary theory in a mathematical sense. 
This result does not depend on how we define a topological order precisely. 

The main idea of our proof is to introduce the notion of a morphism
between two topological orders (see Def.\,\ref{def:morphism-2}). 
This notion is defined physically and independent
of any microscopic/macroscopic definition of a topological order, 
but can be viewed as a special physical realization of a unitary $(n+$$1)$-functor (see Def.\,\ref{def:unitary-n-functor}). 
All simple $n+$1D topological orders together with such morphisms defines the category
$\TO_{n+1}^{fun}$. Using such morphisms, we are able to show in Sec.\,\ref{sec:universal} that the \bulk
satisfies the same universal property as that of the center (of an algebra) in
mathematics. 
By assuming that the new morphism
coincides with usual notions of morphisms in mathematics, we obtain that the \bulk coincides with
the usual notion of center (see Sec.\,\ref{sec:center}). Conversely, by assuming \bulk= center,
we also show that our new notion of a morphism is
compatible with usual notions of morphisms in mathematics in a few low dimensional cases (see
Sec.\,\ref{sec:morphism=morphism}). 


\medskip 
Actually, \bulk = center is the main tool that allows us to pin down the categorical formulation of an $n$D topological order (see Def.\,\ref{def:cat-TO}). Moreover, \bulk = center is only the first layer of the hierarchical structure of a rather complete boundary-bulk relation. In Sec.\,\ref{sec:Z-functor}, we propose a stronger version of the {\it unique-bulk hypothesis} (see Fig.\,\ref{fig:unique-bulk-hypothesis}), which allows us to define the Morita/Witt equivalence of topological orders and closed/anomalous domain walls. It also provides a natural explanation of the so-called duality-defect correspondence (in 3D) as a part of the second layer of the complete boundary-bulk relation. We explain that this hierarchical structure of the boundary-bulk relation is nothing but the functoriality of the \bulk (or center). 

\begin{rema} {\rm 
A similar boundary-bulk relation was first discovered in 2D rational conformal field theories, where it was also called open-closed duality. In particular, the uniqueness of the bulk was first proved in \cite{fjfrs} (in the FRS framework) and in \cite{cardy-alg} (in a modified Segal's framework), duality-defect correspondence in \cite{ffrs-defect,dkr1} and the functoriality of the bulk or center in \cite{dkr2,dkr3}. For topological orders, the functoriality of the bulk or center was proposed earlier in 3D Levin-Wen models \cite{kong-icmp}.
}
\end{rema}

We give some remarks in Sec.\,\ref{sec:conclusion} on what our results suggest to a possible condensation theory and a theory of topological orders enriched by invertible 1-codimensional walls in higher dimensions. The main results of this work are physical and not mathematically rigorous. In Sec.\,\ref{sec:TO-fun} and Sec.\,\ref{sec:universal}, however, we would like to borrow the mathematical terminologies of lemma, proposition and theorem to highlight our physics results, which are often supported and illustrated by mathematically rigorous results in Examples. In Sec.\,\ref{sec:def-unitary-n-cat}, we give a mathematical definition of a unitary $n$-category; in Sec.\,\ref{sec:universal-higher-morphism}, we discuss the universal property of the \bulk with higher morphisms; in Sec.\,\ref{sec:weak-n-morphism}, we introduce the notion of a weak morphism between topological orders.

\begin{rema} \label{rema:f-alg}  {\rm
The fusion-braiding properties encoded in an higher category can only describe topological excitations living in an open disk, thus define only a local topological invariant. To obtain global topological invariants on closed manifolds, we need glue local topological orders via factorization homology \cite{lurie3,aft,afr}. We will do that in \cite{ai} (see also Remark\,\ref{rema:global-TO}).
}
\end{rema}


\noindent {\bf Acknowledgement}: We thank Yinghua Ai for teaching us some basics of factorization algebras, Dan Freed for pointing out issues related to framing anomalies, and Jacob Lurie for explaining to us some properties of the notion of center for $E_k$-algebras in symmetric monoidal $\infty$-categories. In particular, Remark\,\ref{rema:lurie} is due to Jacob Lurie. We also thank J\"{u}rgen Fuchs and Kevin Walker for helpful discussion and helps on references. LK is supported by the Basic Research Young Scholars Program and the Initiative Scientific Research Program of Tsinghua University, and by NSFC under Grant No. 11071134, and by the Center of Mathematical Sciences and Applications at Harvard University. X-G.W is supported by NSF Grant No.  DMR-1005541 and NSFC 11274192 and by the BMO Financial Group and the John Templeton Foundation. Research at Perimeter Institute is supported by the Government of Canada through Industry Canada and by the Province of Ontario through the Ministry of Research. HZ is supported by NSFC under Grant No. 11131008.

\section{Definitions of topological orders} \label{sec:def-TO}

In this section, we propose a microscopic definition of a topological order and a categorical classification, which is formulated as a macroscopic definition of a topological order. 

\subsection{A microscopic definition of a topological order}  \label{sec:def-TO-mic}

\begin{figure}[tb]
\begin{center}
\raisebox{-1.5em}{
\includegraphics[scale=0.6]{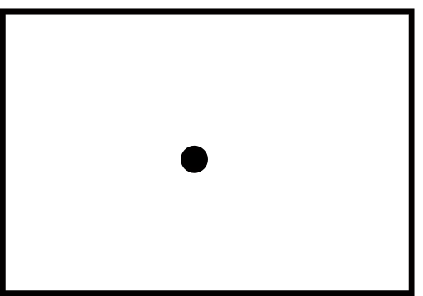}}
\quad\quad\quad $ \Longrightarrow$  \quad\quad\quad
\raisebox{-1.3em}{
\includegraphics[scale=0.5]{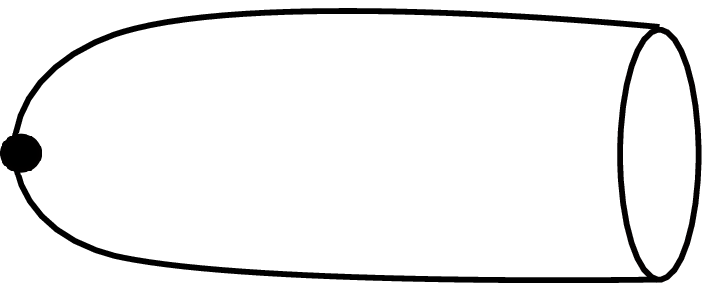}}
\end{center}
\caption{Dimensional reduction:
A defect of codimension 2 (or higher) in a local bosonic Hamiltonian lattice model, depicted as the dark dot in the picture, can be viewed as a boundary of a new lattice model obtained by wrapping the space around the defect such that the whole system looks like a cigar, which viewed from far away is equivalent to a system with a boundary of codimension 1.
}
\label{fig:dim-reduction}
\end{figure}

We only briefly discuss a physical definition of an $n+$1D topological order on an open $n$-disk from a microscopic point of view. The focus of this work is a macroscopic one. 

\medskip
A (potentially anomalous) topologically order can always be realized as a
gapped defect in a higher dimensional bosonic Hamiltonian lattice model defined
on an open disk in space with only short range interactions. Certain gapped properties
need to be satisfied when we take the large-size limit \cite{zeng-wen}. Using
the dimensional reduction depicted in Fig.\,\ref{fig:dim-reduction}, we can see
that any topological order can always be realized as a gapped
boundary of a one-dimensional higher bosonic Hamiltonian lattice model with only short
range interactions. A topological order is an equivalence class of such lattice
realizations. 

Let $L$ be an open $n$-disk on the boundary (a sphere) of a closed $n+$1-disk
$\overline{D}^{n+1}$ (see Fig.\,\ref{fig:def-TO}). In the following, we define the
notion of an $n+$1D (potentially anomalous) topological order on the space
manifold $L$ (or equivalently, on the space-time manifold
$L\times \mathbb{R}$) as an equivalence class of Hamiltonian lattice models defined on
the space manifold $\overline{D}^{n+1}$. 

\begin{defn} \label{def:TO} {\rm
Consider two bosonic Hamiltonian lattice models $H$ and $H'$ with short ranged
interactions defined on a closed $n+$1-disk $\overline{D}^{n+1}$ with an $n$-sphere
boundary depicted in Fig.\,\ref{fig:def-TO} such that both models have liquid
gapped ground states (\cite{zeng-wen}) and the boundary is gapped. We say that
these two lattice models realize the same $n+$1D topological order on the open
$n$-disk $L$ if for any neighborhood $U$ of the boundary $L$ with large enough but finite thickness (see
Fig.\,\ref{fig:def-TO}) such that 
the restriction of $H$ in $U$, denoted by
$H|_U$, can be deformed smoothly to $H'|_U$
without closing the gap. In other words, there is a smooth family $H_t$ for
$t\in [0,1]$ without closing the gap such that $H_0=H$, and $H_r$ and $H_s$
differ only in $U$ for $s,t\in [0,1]$ and $H_1|_U=H'|_U$.  }
\end{defn}

\begin{figure}[tb]
$$
\begin{picture}(100, 125)
   \put(0,0){\scalebox{1.3}{\includegraphics{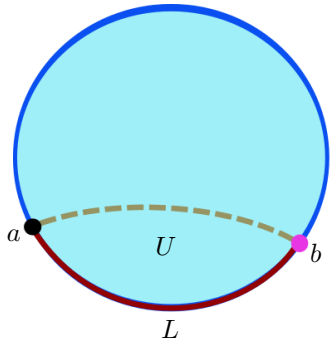}}}
   \put(0,0){
     \setlength{\unitlength}{.75pt}\put(-18,-40){
     \put(10, 77)     { $a$}
     \put(163,66)    { $b$ }
     \put(85, 70)     { $U$}
     \put(88,28)        { $L$}
     }\setlength{\unitlength}{1pt}}
  \end{picture}
$$
\caption{Consider a bosonic Hamiltonian lattice model defined on the closed $n$-disk $\overline{D}^{n+1}$ and an open $(n-1)$-disk $L$ in $\partial \overline{D}^{n+1}=S^n$. The disk $L$ is depicted as an open line segment between $a$ and $b$ on $\partial \overline{D}^{n+1}$, where $a \cup b = \partial L = S^{n-1}$. The lattice model in a neighborhood of $L$ determines a topological order on the open $n$-disk $L$ (see Def.\,\ref{def:TO}). 
}
\label{fig:def-TO}
\end{figure}


We denote $n+$1D topological orders by $\EA_{n+1}, \EB_{n+1}, \EC_{n+1}, \ED_{n+1}$, etc.. If the space-time dimension is clear from the context, we abbreviate $\EC_{n+1}$ as $\EC$.

\begin{defn} \label{def:closed-TO-1}{\rm
If an $n+$1D topological order (on an open $n$-disk $L$) $\EC_{n+1}$ can be realized by a bosonic Hamiltonian lattice model defined on $S^n$ (Fig.\,\ref{fig:def-TO} with an empty bulk), it is called {\it closed} (or {\it anomaly-free}). If otherwise, it is called {\it anomalous}.  } 
\end{defn}

The topological orders defined in \cite{chen-gu-wen,zeng-wen} are closed in our sense.  
In physics, topological orders are often defined on closed manifolds. One can easily generalize Def.\,\ref{def:TO} to topological orders on closed manifolds. For example, if we replace $L$ by the entire boundary $S^n=\partial \overline{D}^{n+1}$ in Def.\,\ref{def:TO} (see Fig.\,\ref{fig:def-TO}), we obtain a microscopic definition of a topological order on  $S^n$. We don't need them in this work. To study the category of topological orders, one also need generalize Def.\,\ref{def:TO} to a microscopic definition of multiple topological orders connected by gapped domain walls and walls between walls, etc. This goes beyond the scope of this paper. We hope to come back to this issue in the future. 

\begin{rema} {\rm
The notion of a closed local topological order seems relate to a different notion, which is also called a ``local topological order" and was introduced  in \cite{MZ}. It is interesting to know their relation. 
}
\end{rema}



\subsection{Unitary multi-fusion $n$-categories}

The definitions of a unitary $n$-category and a unitary $n$-functor for $n\geq 0$ are given in Sec.\,\ref{sec:def-unitary-n-cat} (see Def.\,\ref{def:unitary-n-cat} and Def.\,\ref{def:unitary-n-functor}).

\begin{defn} \label{def:unitary-fusion-n-cat} {\rm
A {\it unitary fusion $n$-category} for $n\geq 0$ is a unitary $(n+$$1)$-category with a unique simple object $\ast$. We also identify it with the unitary $n$-category $\hom(\ast,\ast)$.  We define a {\it unitary multi-fusion $n$-category} to be the $n$-category $\hom(x,x)$ for an (not necessarily simple) object $x$ in a unitary $(n+$$1)$-category. 
}
\end{defn}

Note that a unitary multi-fusion $n$-category $\EC$ can be viewed either as an $(n+$$1)$-category or an $n$-category. But we always denote it by $\EC_{n+1}$ because $n+$1 is the spacetime dimension of the associated topological order (see Def.\,\ref{def:cat-TO}), or simply by $\EC$ if the subscript is clear from the context. For a unitary multi-fusion $n$-category $\EC_{n+1}$, $\id_x$ (or $\id_\ast$ in the fusion case) is also called the {\it tensor unit}, denoted by $1_\EC$. A unitary multi-fusion $n$-category $\EC_{n+1}$ is fusion iff $1_\EC$ is simple.

\begin{rema} \label{rema:def-TO} {\rm 
Although a unitary fusion $n$-category can be viewed as an $n$-category or as an $(n+$$1)$-category (with a unique simple object), there are subtle differences between these two views when we discuss the notion of center (see Def.\,\ref{def:Z-center} and Remark\,\ref{rema:lurie}). 
} \end{rema}

\begin{defn} {\rm
For $n\geq 0$, a unitary multi-fusion $n$-category is called {\it indecomposable} if it is not a direct sum of two such categories. 
}\end{defn}

Note that a unitary fusion $n$-category is automatically an indecomposable unitary multi-fusion $n$-category. 

\begin{expl} {\rm
We give a few examples of unitary multi-fusion $n$-categories. 
\bnu
\item A unitary multi-fusion $0$-category $\EC_1$ is the $\hom(x,x)$ for an object $x$ in a unitary 1-category. It is nothing but a finite dimensional semisimple $C^\ast$-algebra, i.e. a finite direct sum of matrix algebras. It is indecomposable if $\EC_1$ is a simple algebra, i.e. a matrix algebra. It is fusion if $\EC_1=\Cb$. 
\item A unitary multi-fusion $1$-category $\EC_2$ is a unitary abelian semisimple rigid monoidal 1-category with finitely many simple objects \cite{eno2002}.  It is a unitary fusion 1-category if the tensor unit $1_\EC$ is simple, i.e. $\hom_\EC(1_\EC,1_\EC)\simeq \Cb$. 
\item The trivial unitary $(n+$$1)$-category $\one_{n+1}$ consists of a unique simple object $\ast$, a unique simple $k$-morphism $\id_\ast^k$ for $1\leq k\leq n$, where $\id_\ast^k$ is defined inductively by $\id_\ast^1=\id_\ast$ and $\id_\ast^k:=\id_{\id_\ast^{k-1}}$. The category $\one_{n+1}$ is the simplest unitary fusion $n$-category. 

\enu
}
\end{expl}

\begin{defn} {\rm
A {\it unitary braided fusion $n$-category} for $n\geq 0$ is a unitary $(n+$$2)$-category with a unique simple object $\ast$ and unique simple 1-morphism $\id_\ast$. 
We also identify it with the unitary $n$-category $\hom(\id_\ast, \id_\ast)$.
}
\end{defn}

Note that a unitary braided fusion $n$-category $\EC$ can be viewed either as an $(n+$$2)$-category or an $(n+$$1)$-category or an $n$-category. But we always denoted 
it as $\EC_{n+2}$ because $n+$2 is the spacetime dimension of the associated topological order (see Def.\,\ref{def:cat-TO}). In this case, the tensor unit of $\EC_{n+2}$ is $\id_\ast^2$, and is also denoted by $1_\EC$. 

\begin{expl} {\rm
We give a few examples of unitary braided fusion $n$-categories. 
\bnu
\item A unitary braided fusion $0$-category $\EC_2$ can be viewed as a unitary 2-category with a unique simple object $\ast$ and a unique simple $1$-morphism $\id_\ast$, or equivalently,  the 1-dimensional $C^\ast$-algebra $\hom_{\EC_2}(\id_\ast,\id_\ast)=\Cb$. 
\item A unitary braided fusion $1$-category is unitary abelian semisimple rigid braided monoidal 1-category with finitely many simple objects (see for example \cite{dgno}). 
\item $\one_{n+2}$ is the simplest example of a unitary braided fusion $n$-category. 
\enu
}
\end{expl}

\begin{defn} \label{def:Z-center} {\rm
The $Z_{n+1}$-center of a unitary multi-fusion $n$-category $\EC_{n+1}$, denoted by $Z_{n+1}(\EC_{n+1})$, is defined by the category $\fun_{\EC|\EC}(\EC,\EC)$ of unitary $\EC$-$\EC$-bimodule $n$-functors, where $\EC_{n+1}$ is viewed as a monoidal $n$-category instead of an $(n+$$1)$-category. 
}
\end{defn}

\begin{conj} \label{conj:f-Z-bf}
For an indecomposable unitary multi-fusion $n$-category $\EC_{n+1}$ ($n\geq 0$), its $Z_{n+1}$-center $Z_{n+1}(\EC_{n+1})$ is a unitary braided fusion $n$-category, and $Z_{n+1}(\one_{n+1})\simeq\one_{n+2}$.
\end{conj}

This conjecture is known to be true for $n=0,1$. When $n=0$, the $Z_1$-center is the usual notion of the center of an algebra over $\Cb$; when $n=1$, the $Z_2$-center is the Drinfeld center. 

\begin{defn} {\rm
For $n\geq 0$, a unitary multi-fusion $n$-category $\EC_{n+1}$ is called {\it closed} if its $Z_{n+1}$-center is trivial, i.e. $Z_{n+1}(\EC_{n+1})\simeq \one_{n+2}$.
}
\end{defn}

\begin{expl}  {\rm
A closed unitary multi-fusion $n$-category $\EC_{n+1}$ is indecomposable. 
\bnu
\item When $n=0$, $\EC_1$ is closed iff it is indecomposable. 
\item When $n=1$, $\EC_2$ is closed iff it is monoidally equivalent to the category $\fun(\EM, \EM)$ of unitary 1-functors for a unitary $1$-category $\EM$; 
\enu
}
\end{expl}

\begin{defn} \label{def:centralizer} {\rm
For $n\geq 0$, the {\it centralizer} of a unitary braided fusion $n$-category $\EC_{n+2}$, denoted by $\EC_{n+2}'$, is defined to be the subcategory of $\EC_{n+2}$ containing morphisms that have trivial double braidings with all morphisms.  
} 
\end{defn}

\begin{rema} {\rm
We pretend that the meaning of a double braiding is clear in Def.\,\ref{def:centralizer}. Actually, the centralizer $\EC_{n+2}'$ of $\EC_{n+2}$ can be equivalently defined by the universal property (see Example\,\ref{expl:E2-center}). By Lurie's results  (see Remark\,\ref{rema:lurie}), it can also be equivalently defined by the $n$-category $\hom_{Z_{n+2}(\EC_{n+2})}(1_{Z_{n+2}(\EC_{n+2})}, 1_{Z_{n+2}(\EC_{n+2})})$, where $1_{Z_{n+2}(\EC_{n+2})}$ is the identity bimodule functor $\id_\EC$ in $\fun_{\EC|\EC}(\EC,\EC)$. It always contains $\one_{n+2}$ as a unitary subcategory. 
}
\end{rema}

\begin{rema} \label{rema:E_12-alg} {\rm 
A theory of higher algebras was developed by Jacob Lurie \cite{lurie2}. A unitary multi-fusion $n$-category can be viewed as an $E_1$-algebra object (satisfying additional unitary conditions) in the symmetric monoidal $\infty$-category $\text{Cat}_{(\infty,n)}$ of $(\infty,n)$-categories with some additional structures. The $Z_n$-center defined in Def.\,\ref{def:Z-center} is an $E_1$-center of an $E_1$-algebra, and is automatically an $E_2$-algebra \cite{lurie2}. A unitary braided fusion $n$-category can be viewed as an $E_2$-algebra object in $\text{Cat}_{(\infty,n)}$. The centralizer of a unitary braided fusion $n$-category is an $E_2$-center of an $E_2$-algebra, and is automatically an $E_3$-algebra. 
}
\end{rema}

\begin{defn} {\rm
For $n\geq 0$, a unitary braided fusion $n$-category $\EC_{n+2}$ is called {\it non-degenerate} if its centralizer is trivial, i.e. $\EC_{n+2}'\simeq\one_{n+2}$. 
}
\end{defn}

\begin{conj} \label{conj:closed-f=bf}
For $n\geq 0$, a unitary fusion $(n+$$1)$-category $\EC_{n+2}$ is closed if and only if it is a non-degenerate unitary braided fusion $n$-category. 
\end{conj}

The proof of the ``only if" part follows from Remark\,\ref{rema:lurie}. The complete mathematical proofs of Conjecture\,\ref{conj:closed-f=bf} are not known to us except the trivial case $n=0$. But we provide a physical explanation of this conjecture in Remark\,\ref{rema:hCat}. 

\begin{defn} {\rm
For $n\geq 0$, a unitary braided fusion $n$-category $\EC_{n+2}$ is called {\it exact} if there is an indecomposable multi-fusion $n$-category $\ED_{n+1}$ such that $\EC_{n+2}\simeq Z_{n+1}(\ED_{n+1})$. 
}
\end{defn}

For $n=0$, the only unitary braided fusion $0$-category $\EC_2$ is $\Cb$, and is exact. For $n=1$, the exact unitary braided fusion 1-categories are called ``non-chiral" topological orders, and the quotient $\{ closed \}/\{ exact \}$ 
is nothing but the Witt group \cite{dmno}.

\subsection{A macroscopic definition of a topological order} \label{sec:def-TO-mac}

Macroscopically, up to invertible topological orders \cite{kong-wen}, the local order parameters of a topological state of matter are given by the fusion-braiding (and spins) structures of its topological excitations. Note that these fusion-braiding properties only make sense locally (on an open disk), thus give topological invariants on an open disk. 
For example, the ``double braiding" of two particles on a 2-sphere is physically ambiguous. It was well-known that these local order parameters (without the spin structures see Remark\,\ref{rema:spins}) can be encoded efficiently in the data and axioms of an higher category \cite{baez} (see also \cite{kong-wen}). In this subsection, we propose a classification of $n+$1D topological orders\footnote{A topological order here was called a $\BF$-category in \cite{kong-wen}. It determines a topological order up to invertible topological orders \cite{kong-wen,freed}, which are topological orders without topological excitations.} by unitary (multi-)fusion $n$-categories. For convenience, we formulate this classification as the macroscopic definition of a topological order. 

\medskip
The mathematical definition of a unitary fusion $n$-category is based on that of a unitary $(n+$$1)$-category, which is defined in Appendix (see Def.\,\ref{def:unitary-n-cat}). Here, we only remind readers of some basic ingredients (and their physical meanings) of a unitary $(n+$$1)$-category without being very precise. In a unitary $(n+$$1)$-category, the physical meaning of a 1-morphism is a disk-like (see Remark\,\ref{rema:disk-stratification}) 1-codimensional excitation (also called a domain wall); an $l$-morphism, denoted by $g^{[l]}$ (or $g$ if the superscript ${}^{[l]}$ is clear from the context), is an $l$-codimensional excitation; an $(n+$$1)$-morphism is an instanton localized on a point on the time axis. The spaces of instantons are all finite dimensional. An $(n-1)$-morphism $h$ is simple if $\hom(h,h)\simeq\Cb$ and an $l$-morphism $g$ is simple if $\id_g$ is simple (defined inductively). For $l\geq 0$, each $l$-morphism is a direct sum of simple $l$-morphisms, and there are only finitely many simple $l$-morphisms. For $k<l$, an $l$-morphism $g^{[l]}$ in $\hom(h^{[k]}, h^{[k]})$ is an excitation nested on the $k$-codimensional excitation $h^{[k]}$. The composition of two $l$-morphisms corresponds to the fusion of two excitations. 
Each $l$-morphism $g$ has a two-side dual $\bar{g}$, which should be viewed as the anti-excitation of $g$. For any two non-isomorphic simple $k$-morphisms $f$ and $g$, $\hom(f,g)=0$. The physical meaning of this condition was explained in \cite{kong-wen}. The category $\one_{n+1}$ is the smallest unitary $(n+$$1)$-category consisting of a unique simple object $\ast$ and a unique simple $k$-morphism $\id_\ast^k$ for $k=1,2,\dots, n$.

\begin{rema} \label{rema:disk-stratification}  {\rm
Microscopically, a $p$-space-dimensional excitation can be defined in a lattice model with additional Hamiltonian term $\Delta H$, which is non-zero only in a neighborhood of a submanifold $M^p$ \cite{kong-wen}. But in an $n+$1D topological order, a $p$-space-dimensional excitation only makes sense on an open $p$-disk in general \cite{kong-wen} because a submanifold is not a local concept.  Therefore, by a $p$-dimensional excitation in a topological order, we always mean a $p$-disk-like excitation. If a $p$-dimensional excitation defined on a closed submanifold $M^p$ happens to be small, i.e. lying in an open $n$-disk, we should choose a cellular decomposition of $M^p$ and view it as a disjoint union of $k$-disk-like excitations for $0\leq k\leq p$. For example, a string-like excitation on $M^1=S^1=\{ e^{i\theta} | \theta \in [0,2\pi)\}$ can be viewed as the disjoint union of two 1-disk-like excitations on $(0,\pi)$ and $(\pi,2\pi)$ and two 0-disk-like excitations at $\theta=0$ and $\theta=\pi$. 
An $n+$1D topological order should produce topological invariants for disk-stratified open $n$-disks \cite{afr}. 
}
\end{rema}

To determine the categorical description of a potentially-anomalous $n+$1D topological order $\EC_{n+1}$, we first determine that of a closed $n+$2D topological order $\EB_{n+2}$, then that of $\EC_{n+1}$ as a boundary of $\EB_{n+1}$. If $\EC_{n+1}$ is also realized as a $k$-codimensional defect of an $n+k$D closed topological order $\EA_{n+k}$ for $k>1$, it automatically inherits a categorical description from that of $\EA_{n+k}$ (see Remark\,\ref{rema:codim}). These two descriptions of $\EC_n$ can be different. Therefore, we expect that different categorical descriptions of a topological order according to different codimensions (see Remark\,\ref{rema:codim}). The description we give in Def.\,\ref{def:cat-TO} is the 1-codimensional description. 

We also distinguish two types of topological orders: simple and
composite. Most of the topological orders studied in physics
are simple topological orders. Composite topological orders are direct sums of
simple topological orders. They naturally occur as the result of a dimensional
reduction or a fine tuning of a system near multi-phase transition points in
the phase diagram \cite{kong-wen,zeng-wen}. A direct sum of two simple
topological orders $\EA_n$ and $\EB_n$ means that the system has accidental
degenerate ground states, and is in either the $\EA_n$-state or the $\EB_n$-state. 
Additional perturbations such as applying external fields often push the system
to select one ground state from the two.

In a closed (or anomaly-free) $n+$2D topological order, excitations should be able to detect each other via braidings \cite{kong-wen}. Since 1-codimensional excitations can not be braided with any excitations, it implies that the only simple 1-codimensional excitation is the trivial one. Therefore, a closed $n+$2D topological order should be a unitary braided fusion $n$-category (an $E_2$-algebra), denoted by $\EC_{n+2}$, if we assume an $(n+$$2)$-stability condition discussed in the next paragraph. Moreover, its braidings should be able to detect all excitations. This condition is equivalent to the condition that $\EC_{n+2}$ has a trivial centralizer. Therefore, we obtain that a simple closed $n+$2D topological order should be given by a non-degenerate unitary braided fusion $n$-category. The boundary-bulk relation (studied in later sections) suggests that the boundary theory of the bulk phase $\EC_{n+2}$ should be given by an $E_1$-algebra such that its $E_1$-center (recall Remark\,\ref{rema:E_12-alg}) is $\EC_{n+2}$. Then Conjecture\,\ref{conj:f-Z-bf} suggests that a simple (potentially-anomalous) $n+$1D topological order, as a boundary of a simple closed $n+$2D topological order, can be described by an indecomposable unitary multi-fusion $n$-category (an $E_1$-algebra).


An $(n+$$1)$D topological state $\EC_{n+1}$ is $(n+$$1)$-stable if the vacuum degeneracy on $S^n$ (as the space manifold) is trivial \cite{kong-wen}. Mathematically, it amounts to the stability condition that the tensor unit of $\EC_{n+1}$ is simple, or equivalently, $\EC_{n+1}$ is a unitary fusion $n$-category. If the vacuum degeneracy on $S^n$ is non-trivial, the topological state is $(n+$$1)$-unstable, and it can flow to the $(n+$$1)$-stable ones under the local perturbations of the Hamiltonians. The reason for $(n+$$1)$-unstable topological states to have natural categorical descriptions is that unitary multi-fusion $n$-categories naturally appear as stable higher dimensional phases. They naturally occur in the processes of dimensional reduction. 
We give a few examples of 3-unstable topological states in Example\,\ref{expl:n-category}, \ref{expl:toric-code}, \ref{expl:lw-mod}.

\medskip
We propose the following classification of topological orders (up to invertible topological orders \cite{kong-wen,freed}). For convenience, we formulate this classification proposal as the macroscopic definition of a topological order. 
\begin{defn} \label{def:cat-TO} {\rm  
For $n\geq 0$, 
a (potentially anomalous) $n+$1D topological order is defined by a unitary multi-fusion $n$-category $\EC_{n+1}$. 
\bnu
\item The trivial $n+$1D topological order is given by $\one_{n+1}$. 
\item It is {\it simple} if $\EC_{n+1}$ is indecomposable; it is {\it composite} if otherwise. 
\item It is {\it $(n+$$1)$-stable} if $\EC_{n+1}$ is a unitary fusion $n$-category; it is {\it $(n+$$1)$-unstable} if otherwise. 
\item It is {\it closed} (or {\it anomaly-free}) if $\EC_{n+1}$ is a closed, i.e. $Z_{n+1}(\EC_{n+1})=\one_{n+2}$; it is {\it anomalous} if otherwise. 
\item For $n\geq 1$, it is $(n+$$1)$-stable and closed if and only if $\EC_{n+1}$ is a non-degenerate unitary braided fusion $(n-$$1)$-category. 
\enu
}
\end{defn}

\begin{rema} \label{rema:micro-macro} {\rm 
It is not known if the microscopic and macroscopic definitions of a topological order are equivalent (up to invertible topological orders \cite{kong-wen,freed}). To find equivalent microscopic and macroscopic definitions is a fundamental open problem in the study of this subject. Superficially speaking, Def.\,\ref{def:closed-TO-1} is compatible with Def.\,\ref{def:cat-TO} in the sense that the physical notion of the \bulk is equivalent to the mathematical notion of the center as we will show later. 

} \end{rema}

\begin{expl} \label{expl:n-category} {\rm
We give some examples of topological orders in low dimensions. 
\bnu



\item A 0+1D topological order $\EC_1$ is given by a unitary multi-fusion $0$-category, which is nothing but $\hom(x,x)$ for an object in a unitary $1$-category. Note that $\hom(x,x)$ is a finite dimensional $C^\ast$-aglebra, which is a direct sum of matrix algebras. If $\EC_1$ is simple, it is just a matrix algebra. It is $1$-stable if and only if $\EC_1$ is fusion, i.e. $\EC_1\simeq \Cb$. We also denote the unique $1$-stable 1D topological order by $\one_1$, i.e. $\one_1=\Cb$. As a $1$-category, $\one_1=\hilb$, where $\hilb$ is the category of finite dimensional Hilbert spaces. 

Note that if $\EC_1$ is simple, we $\EC_1$ is a matrix algebra, i.e. $\EC_1=M_{k\times k}(\Cb)$. This matrix algebra can be viewed as the local operator algebra of a quantum system $\Cb^k$ defined on a point ($D^0$ as the space manifold) with Hamiltonian given by the identity matrix. If $k>1$, this system is degenerate. But this degeneracy is not $1$-stable. By perturbing the Hamiltonian, one can easily lift this degeneracy.

\item Although we have not included the categorical definition of a $0$D topological order in Def.\,\ref{def:cat-TO}, we can add it by hands. Since a 1-stable closed 1D topological order is just 1-dimensional algebra $\Cb$ (an $E_1$-algebra), a 0D topological order, viewed as a wall between two 1D topological orders, is just a finite dimensional Hilbert space $V$ together with a distinguished element $v\in V$, i.e. an $E_0$-algebra $(V,v)$.

\item A simple $2$D topological order can be physically realized on an open 1-disk in the boundary of a local Hamiltonian model on the 2-disk (\cite{bravyi-kitaev,kitaev-kong}). It has particle-like topological excitations, which can be fused. Therefore, it can be described by a unitary multi-fusion 1-category $\EC_2$. When the tensor unit $1_\EC$ is simple, i.e. $\hom(1_\EC, 1_\EC)=\Cb$, $\EC_2$ is a unitary fusion 1-category. In this case, the space of ground states of the local Hamiltonian model on $S^1$ is given by $\hom(1_\EC, 1_\EC)$ (\cite{kong-wen,ai}), and is not degenerate. So it is 2-stable. All unitary fusion 1-categories can be realized as the boundary theories of Levin-Wen models \cite{kitaev-kong}. We give some examples of 2-unstable phases, which are stable as higher dimensional phases, in Example\,\ref{expl:toric-code} and Example\,\ref{expl:lw-mod}.

The trivial 2D topological order $\one_2$, which is also closed. If $\EC_2$ contains a non-trivial particle-like excitation, it cannot be detected by other excitations because there is no braiding. It needs a 2+1D bulk to detect it. Therefore, $\EC_2$ describes an anomalous topological order. Therefore, $\one_2$ is the only closed 2-stable 2D topological order. 


\item A simple $3$-stable $3$D topological order can be described by a unitary fusion 2-category, or equivalently, a unitary 3-category $\EC_3$ with a unique simple object $\ast$. 1-morphisms are string-like excitations, 2-morphisms are particle-like excitations and 3-morphisms are instantons. When $\EC_3$ is closed, there is no simple 1-morphisms (string-like excitations) other than the trivial one $\id_\ast$ because they can not be braided with other excitations\footnote{In this case, string-like excitations are still possible as constructed in \cite{kitaev-kong}, but they can all be obtained from condensations of point-like excitations \cite{kong-anyon}. Therefore, such string-like excitations must be excluded from a minimal categorical description of a topological order \cite{kong-wen}. If there is a string-like excitation that is not condensed, it is not detectable via braiding because it can not be braided with any other excitations. Therefore, such a topological order must be anomalous.}. As a result, $\EC_3$ is a non-degenerate unitary braided fusion 1-category. In this case, the $Z_3$-centers of $\EC_3$ and the centralizer $\EC_3'$ should coincide and are both trivial (recall Conjecture\,\ref{conj:closed-f=bf}). If the topological order $\EC_3$ is anomalous, we can not reduce $\EC_3$ to $\hom(\id_\ast, \id_\ast)$ even if $\id_\ast$ is the unique simple 1-morphism because the notion of $Z_3$-center and that of centralizer are different (see Remark\,\ref{rema:lurie}). Its $Z_3$-center contains string-like excitations in general. For example, in the $\Zb_n$ gauge theory in 3+1D (see for example \cite{wang-wen-2}), all string-like excitations are mutually symmetric but have non-trivial braidings with particle-like excitations. When all string-like excitations are condensed, it creates an anomalous 3D topological order $\EC_3$ on the boundary with only particle-like excitations.

\enu
}
\end{expl}

\begin{rema} \label{rema:codim} {\rm
Note that we find Def.\,\ref{def:cat-TO} by first fixing the categorical description of a simple closed $n+$2D topological order $\EB_{n+2}$, which is an $E_2$-algebra, then determining that of an $n+$1D topological order (viewed as the 1-codimensional boundary of $\EB_{n+2}$) as an $E_1$-algebra in Def.\,\ref{def:cat-TO}. So this can be called a 1-codimensional definition (or classification). Higher codimensional definition is also possible. An $n$D topological order, realized by a 2-codimensional excitation in $\EB_{n+2}$, can be naturally described by the unitary $n$-category $\EE_n=\hom_{\EB_{n+2}}(\id_\ast, \id_\ast)$ with a distinguished object $\iota$. Such a pair $(\EE_n, \iota)$ is an $E_0$ algebra. One can go one-step further. An $n-$1D topological order, realized by a 3-codimensional excitation in $\EB_{n+2}$, can be described by the unitary $n$-category $\EE_n$ together with a distinguished object $\iota^{(0)}$ and a distinguished 1-morphism $\iota^{(1)}$ in $\hom_\EE(\iota^{(0)}, \iota^{(0)})$. Such a triple $(\EE_n, \iota^{(0)}, \iota^{(1)})$ can be viewed as an $E_{-1}$-algebra, which is a not-yet-defined notion. What justifies this terminology is that its bulk (or center) is an $E_0$-algebra. Continuing along this path of logic, we obtain that an $n-$$k$D topological order (for $k\leq n$), realized by a $(2+$$k)$-codimensional excitation in $\EB_{n+2}$, can be described by the unitary $n$-category $\EE_n$ together with a choice of object $\iota^{[0]}$ and a choice of 1-morphism $\iota^{[1]}$ in $\hom(\iota^{[0]},\iota^{[0]})$ and a choice of 2-morphism $\iota^{[2]}$ in $\hom(\iota^{[1]}, \iota^{[1]})$, so on and so forth. Then $(\EE_n, \iota^{[0]}, \iota^{[1]}, \iota^{[2]}, \cdots, \iota^{[k]})$ can be viewed as $E_{-k}$-algebra. We leave a more careful study of these algebras elsewhere. 
} 
\end{rema}

\begin{expl} \label{expl:codim=2} {\rm
We give a couple of example higher codimensional definitions. 
\bnu

\item A 2-stable closed 2D topological order is given by 1-dimensional algebra $\Cb$ (an $E_2$-algebra), 
a simple 1D topological order is given by finite dimensional simple algebra $A$ (an $E_1$-algebra). So a 2-codimensional definition of a 0D simple topological order is a simple $A$-module $V$ together with a distinguished element $v$, i.e. an $E_0$-algebra $(V, v)$. 

\item Although it does not make physical sense to talk about $-1$D topological order, it might be viewed intuitively as a ``wall" between two simple 0D topological orders, which are given by two finite dimensional $A$-modules $(U,u)$ and $(V,v)$. Therefore, we can give a 3-codimensional definition of a $-1$D topological order as an $A$-module map $f: U \to V$ such that $f(u)=v$. The ``wall" is invertible if $f$ is an invertible map. 
\enu
}
\end{expl}

\begin{rema} \label{rema:global-TO}
{\rm Topological orders defined on space manifolds other than open disks do occur in this work but only briefly discussed (see Example\,\ref{expl:toric-code}, \ref{expl:lw-mod} and Remark\,\ref{rema:unique-f0}). For a systematic study of topological orders on other manifolds from a macroscopic point of view, we need the theory of factorization homology on disk-stratified manifolds \cite{lurie2,aft1,aft,afr} (see also Remark\,\ref{rema:spins}). Given a unitary modular tensor 1-category as a $2+1$D local topological order, we will show in \cite{ai} that the global topological invariants (or the global order parameters) on closed 2-dimensional surfaces of all genera (as the space manifolds) obtained from factorization homology are all given by the category $\one_1=\hilb$. Moreover, if we trap finitely many topological particles on a given closed surface, factorization homology produces an object in $\hilb$, i.e. a finite dimensional Hilbert space, which is nothing but the space of degenerate ground states on that closed surface. Furthermore, this result remains to be true for the surfaces decorated by arbitrary gapped 1- and 2-codimensional defects. 
}
\end{rema}

\begin{rema} \label{rema:spins} {\rm
In general, topological excitations are defined on disk-stratified submanifolds equipped with certain tangential structures \cite{aft1,afr}, such as framing. In this work, we completely ignore these tangential structures. For all algebraic constructions appeared in this work, such as the construction of center, gapped domain walls and Morita/Witt equivalence, they do not seem to play any explicit role. 
}
\end{rema}

\begin{rema} \label{rema:mtc} {\rm
Our notion of closed/anomalous topological order is different from but not contradicting to that of a TQFT (defined on closed space manifolds with boundaries (and corners) and non-trivial spacetime cobordisms). For example, a quantum Hall system or a modular tensor category gives a closed topological order according to our definition, but the associated 2+1D TQFT is ``anomalous" due to the framing anomaly \cite{rt,turaev}. It has a non-trivial bulk 3+1D TQFT, which is invertible but with no non-trivial excitations \cite{walker-wang,freed-teleman,freed}. This is compatible with our result that a quantum Hall system on an open 2-disk is closed because its \bulk (defined on a spatial open 3-disk) contains no non-trivial topological excitation. 
} 
\end{rema}

\begin{rema}[Minimal Assumption] \label{rema:minimal} {\rm
In physics, a topological excitation is an equivalence class of excitations. It is invariant under smooth deformation and the action of any local operators \cite{kong-wen}. So on the categorical side, it is reasonable to identify all morphisms that are isomorphic. More precisely, given an ordinary unitary $n$-category, we take one representative from each equivalence class of $k$-morphisms for $0\leq k <n$. These representatives form a much smaller unitary $n$-category, in which distinct morphisms are not isomorphic. As a result, all $k$-morphisms form a set with only countable many elements. We call such a category {\it minimal} \cite{lurie1}. In many parts of this work, we would like to assume that the unitary $n$-categories under consideration are minimal so that it is legitimate to say that two unitary $n$-categories (or topological orders) are identical, i.e. $\EC_n=\ED_n$. This assumption is not absolutely necessary for our theory. But it simplifies the discussion. One of the major simplifications is that two minimal unitary $n$-categories are equivalent only if they are isomorphic. In this case, one can identify these two $n$-categories. Since this assumption is not absolutely necessary, we use it only when it is needed. For example, it is used significantly in Sec.\,\ref{sec:universal}. 
For mathematical results associated to $n$-categories, we still use the usual notion of an equivalence $\simeq$ between categories. But for physics minded readers, we do recommend to take this assumption and regard $\simeq$ as $=$ throughout this work. 
}
\end{rema}

\section{The category $\TO_n^{closed-wall}$ of topological orders} \label{sec:TO-wall}

In this section, we discuss some structures of the category $\TO_n^{wall}$ that are used in this work.

\subsection{Time-reverse of a topological order}  \label{sec:time-reverse}
An $n$D topological order $\EC_n$ can be realized by a $n$D local Hamiltonian lattice model $\Gamma$. If we mirror reflect the system along the time direction (or any odd number of space-time directions),
we obtain a local Hamiltonian lattice model $\Gamma^\op$ (or $\overline{\Gamma}$). We denote associated topological order by $\EC_n^\op$ (or $\overline{\EC}_n$), which is called the time-reverse (or the reverse) of $\EC_n$. Mathematically, $\EC_n^\op$ is the $n$-category obtained from $\EC_n$ by flipping all $n$-morphisms; $\overline{\EC}_n$ is obtained from $\EC_n$ by flipping all $l_1$-, $\cdots$, $l_k$-morphisms for $k$ being odd. 

\begin{expl} \label{expl:time-reverse} {\rm
We discuss a few low dimensional cases. 

\bnu


\item A $0+$1D topological order is given by a semisimple algebra $A=\hom_{\EC_1}(x,x)$ (with multiplication $m(a,b)=ab$ for $a,b\in A$) for an object in a unitary $1$-category $\EC_1$. The physical meaning of an arrow in $\EC_1$ is an instanton. Reversing the time amounts to flipping all the arrows in $\EC_1$. Therefore, the time-reverse of $\EC_1$ is the opposite category $\EC_1^\op$, which is the same category $\EC_1$ but with arrows reversed. So the time reverse of $A$ is just the opposite algebra $A^\op$ with the opposite multiplication $m^\op(a,b)=ba$. For example, in Example\,\ref{expl:lw-mod}, when we fold two boundaries in Fig.\,\ref{fig:lw-mod} (a), we need flipped the orientation of the $\EN$-boundary. This explains the ``$\op$" in the first line of Eq. (\ref{eq:lw-mod-dim-red-2}).

\item In $2$D, the time-reverse of a unitary multi-fusion 1-category $\EC_2$ is the opposite category $\EC_2^\op$ equipped with the tensor product $\otimes$. The coherence isomorphisms, i.e. the associator and unit isomorphisms, in $\EC_2^\op$ are taken to be the inverse of those in $\EC_2$.

Actually, the unitary fusion 1-category $\EC_2^\op$ is monodically equivalent to the fusion 1-category $\EC_2^{\rev}$ (obtained by flipping 1-morphisms), which is the category as $\EC_2$ but equipped with the opposite tensor product $\otimes^\rev$, i.e. $x \otimes^\rev y := y\otimes x$, and the coherence isomorphisms are the inverse of those in $\EC_2$. This monoidal equivalence $\EC_2^{\rev}\simeq \EC_2^\op$ is given by taking duals $x \mapsto \overline{x}$. Note that $\EC_2^\rev$, viewed as a unitary 2-category with a simple object, is just the space-reverse of $\EC_2$. So the monoidal equivalence $\EC_2^{\rev}\simeq \EC_2^\op$ simply says that the time-reverse is equivalent to the space-inverse. 

\item In $n$D for $n>2$, the time-reverse of a unitary $n$-category is again given by the opposite category $\EC_n^\op$, which is the same category as $\EC_n$ but with reversed $n$-morphisms and all coherence isomorphisms are taken to be the inverse of those in $\EC_n$. For example, for a closed $\EC_3$ topological order, i.e. a unitary braided fusion 1-category, the braiding in $\EC_3^\op$ is taken to be the inverse of that in $\EC_3$. 

\enu
}
\end{expl}

\begin{rema} \label{rema:time-reverse} {\rm
For a unitary $n$-category and $k\neq l$, flipping the arrows for all $k$-morphisms and all $l$-morphisms leaves the category unchanged (up to equivalences) due to its fully dualizability. Therefore, a mirror reflection of a topological order in any odd number directions is equivalent to the time reverse, i.e. $\overline{\EC}_n\simeq \EC_n^\op$. 
}
\end{rema}

\subsection{Dimensional reductions and tensor products}  \label{sec:dim-red}

For $n\geq 0$, a 1-codimensional excitation $x$ in a potentially anomalous $n+$1D topological order $\EC_{n+1}$ determines an anomalous $n$D topological order defined in a neighborhood of $x$, depicted in Fig.\,\ref{fig:PnC}.
\begin{figure}
$$
\raisebox{-3.5em}{\begin{picture}(100, 85)
   \put(20,0){\scalebox{2}{\includegraphics{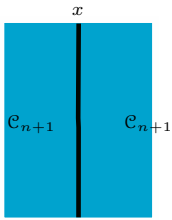}}}
   \put(20,0){
     \setlength{\unitlength}{.75pt}\put(-18,-40){
     \put(51, 148)     { \scriptsize $x $}
     \put(18, 91)     { \scriptsize $\EC_{n+1}$}
     \put(77, 91)     { \scriptsize $\EC_{n+1}$}
     }\setlength{\unitlength}{1pt}}
  \end{picture}}
$$
\caption{A 1-codimensional excitation $x$ in an $n+$1D simple topological order $\EC_{n+1}$ can be viewed as an $n$D topological order, denoted by $(P_n(\EC_{n+1}), x)$.}
\label{fig:PnC}
\end{figure}
This neighborhood automatically includes the action of 1-codimensional excitations on $x$ via fusion. Such obtained anomalous $n$D topological order can be described by an $E_0$-algebra $(P_n(\EC_{n+1}), x)$, where $P_n(\EC_{n+1})$ is the unitary $n$-category obtained from the unitary multi-fusion $n$-category $\EC_n$ by forgetting the monoidal structures. If $x$ is the trivial 1-codimensional excitation, we simply denote the pair by $P_n(\EC_{n+1})$. 

For a closed $n+$2D topological order $\EB_{n+2}$, i.e. a unitary braided fusion $n$-category with a unique simple 1-morphism $\id_\ast$, $P_{n+1}(\EB_{n+2})$ is just the unitary fusion $n$-category $\hom(\ast, \ast)$ but forgetting its braiding structures. In this case, we can also define $P_n(\EB_{n+2})$ to be the pair $(P_n(\EB_{n+2}),\id_\ast^2)$, where $P_n(\EB_{n+2})$ is the unitary $n$-category $\hom(\id_\ast, \id_\ast)$ but forgetting its braiding and monoidal structures. Note that $\EB_{n+2}$ is an $E_2$-algebra; $P_{n+1}(\EB_{n+2})$ is an $E_1$-algebra; $P_n(\EB_{n+2})$ is an $E_0$-algebra.


\begin{rema} {\rm
One can continue this process to define $P_k(\EB_{n+2})$ for $k<n$ as we have done in Remark\,\ref{rema:codim}. Namely, $P_k(\EB_{n+2}):=(\hom(\id_\ast, \id_\ast), \id_\ast^2, \cdots, \id_\ast^{n+2-k})$, which can be viewed as an $E_{k-n}$-algebra. We don't need them in this work. 
}
\end{rema}

\begin{expl} \label{expl:Pn} {\rm
A couple of examples from 2+1D.
\bnu
\item In the toric code model, the bulk excitations are given by the Drinfeld center $Z_2(\rep_{\Zb_2})$ of the unitary fusion 1-category $\rep_{\Zb_2}$ with simple objects $1,e,m,\epsilon$. The trivial domain wall gives an anomalous 2D topological order $P_2(Z_2(\rep_{\Zb_2}))$, which is just the same unitary fusion 1-category as $Z(\rep_{\Zb_2})$ but forgetting the braiding structure. The trivial particle-like excitation $1$ can be viewed as an anomalous 1D topological order by including a neighborhood of $1$. Therefore, one should include all the image of the action of $Z_2(\rep_{\Zb_2})$ on $1$. This image is nothing but $P_1(Z_2(\rep_{\Zb_2}))$, which is the same unitary 1-category as $Z_2(\rep_{\Zb_2})$ (together with the tensor unit) but forgetting all braiding and fusion structures. An $e$-particle can be viewed as a 1D topological order $(Z_2(\rep_{\Zb_2}), e)$, where $Z_2(\rep_{\Zb_2})$ is viewed as a unitary 1-category.

\item Consider a Levin-Wen model with two boundaries such that the bulk lattice defined by a unitary fusion 1-category $\EC_2$ and the lattice near two boundaries defined by unitary indecomposable $\EC$-modules $\EM$ and $\EN$, respectively \cite{kitaev-kong} (see Fig.\,\ref{fig:lw-mod}). Excitations on the $\EM$-boundary (or the $\EN$-boundary) is given by $\EC_\EM^\vee$ (or $\EC_\EN^\vee$), where $\EC_\EM^\vee:=\fun_\EC(\EM, \EM)^\rev$ is the unitary fusion 1-category of unitary $\EC$-module functors equipped with the opposite tensor product. A defect of codimension 2 between the $\EM$-boundary and the $\EN$-boundary is given by a $\EC$-module functor $f$ in $\fun_\EC(\EM, \EN)$ (see Fig.\,\ref{fig:lw-mod}) \cite{kitaev-kong}. When this defect $f$ is viewed as an anomalous $1$D topological order, we must include a neighborhood of this defect, more precisely, include all excitations generated by the fusion of the boundaries/bulk excitations with this defect. This fusion action covers all objects in $\fun_\EC(\EM, \EN)$ because $\fun_\EC(\EM, \EN)$ is an indecomposable $\EC_\EN^\vee$-$\EC_\EM^\vee$-bimodule \cite{eo}. It is convenient to denote the bimdoule by a pair $(\fun_\EC(\EM, \EN), f)$ (see the purple dot in Fig.\,\ref{fig:lw-mod} (a)). When $\EM=\EN$ and $f=\id_\EM$, it is nothing but $P_1(\fun_\EC(\EM, \EM))$.

\enu
}
\end{expl}

One can stack an $n$D topological order $\EA_n$ on the top of the another one $\EB_n$. This operation is denoted by $\boxtimes$, also called a tensor product. More general tensor product can be obtained by gluing $\EA_n$ with $\EB_n$ by an $(n+$$1)$D bulk $\EC_{n+1}$, denoted by $\EA_n\boxtimes_{\EC_{n+1}} \EB_n$. It is summarized by the following graphic equations.
\be \label{eq:AxB}
\EA_n \boxtimes \EB_n :=
\raisebox{-3.5em}{\begin{picture}(100, 80)
   \put(20,0){\scalebox{2}{\includegraphics{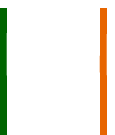}}}
   \put(20,0){
     \setlength{\unitlength}{.75pt}\put(-18,-40){
     \put(40, 91)     { $ \one_{n+1} $}
     \put(-6, 91)     { $\EA_n$}
     \put(99, 91)     { $\EB_n$}
     }\setlength{\unitlength}{1pt}}
  \end{picture}}
\quad,\quad\quad\quad
\EA_n \boxtimes_{\EC_{n+1}} \EB_n :=
\raisebox{-3.5em}{\begin{picture}(100, 80)
   \put(20,0){\scalebox{2}{\includegraphics{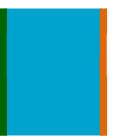}}}
   \put(20,0){
     \setlength{\unitlength}{.75pt}\put(-18,-40){
     \put(40, 91)     { $ \EC_{n+1} $}
     \put(-6, 91)     { $\EA_n$}
     \put(99, 91)     { $\EB_n$}
     }\setlength{\unitlength}{1pt}}
  \end{picture}}
\ee
It is clear that $\boxtimes=\boxtimes_{\one_{n+1}}$. Notice that $\EA_n\boxtimes_{\EC_{n+1}} \EB_n$ can be viewed either as an $n$D topological state (very often $n$-unstable, see Example\,\ref{expl:toric-code} and \ref{expl:lw-mod}) or as an $(n+$$1)$D phase with two boundaries. The former one is a dimensional reduction of the later one.

\begin{rema} \label{rema:tensor-product} {\rm
By Def.\,\ref{def:cat-TO}, when $\EC_{n+1}$ is simple and closed, i.e. a unitary braided fusion $(n-1)$-category with the usual unit, and $\EA_n$ and $\EB_n$ are simple, i.e. a unitary fusion $1$-category with the usual unit, then the tensor product $\EA_n\boxtimes_{\EC_{n+1}} \EB_n$ is just the usual tensor product of a right $\EC_{n+1}$-module and a left $\EC_{n+1}$-module in the category of unitary fusion $(n-1)$-categories with unitary monoidal $(n-1)$-functor as morphisms. In particular, the right action $\EA_n \boxtimes \EC_{n+1} \to \EA_n$ and the left action $\EC_{n+1} \boxtimes \EB_n \to \EB_n$ are both unitary monoidal.
}
\end{rema}

\begin{expl} \label{expl:toric-code} {\rm
We would like to illustrate the construction $\EA_n\boxtimes_{\EC_{n+1}} \EB_n$ by concrete lattice models when $n=2$. Consider two narrow bands of toric code model depicted in Figure\,\ref{fig:toric-code}. These two narrow bands can be viewed as two 1+1D topological orders. This is an example of dimensional reduction.
\bnu
\item In (a), the $e$- and $m$-particle in the bulk or on one of the boundary condense to vacuum on the other boundary. Therefore, there is no non-trivial particle-like excitations in the band at all, and the $2$D topological order thus obtained is trivial, i.e.
\begin{align}
\EA_2\boxtimes_{\EC_3} \EB_2 &= \hilb. \nonumber
\end{align}
This fact is also supported by abstract nonsense. When the toric code model is viewed as an example of Levin-Wen models, the bulk lattice is determined by the data of the unitary fusion category $\rep_{\Zb_2}$ (the category of representations of $\Zb_2$ group). The smooth boundary is again determined by that of $\rep_{\Zb_2}$ but viewed as a right module over the fusion category $\rep_{\Zb_2}$ (or a right $\rep_{\Zb_2}$-module); the rough boundary is determined by the category $\hilb$ which is a left $\rep_{\Zb_2}$-module \cite{bravyi-kitaev}\cite{kitaev-kong}. The topological excitations on the smooth boundary is given by the unitary fusion category $\rep_{\Zb_2}\simeq \fun_{\rep_{\Zb_2}}(\rep_{\Zb_2}, \rep_{\Zb_2})$, those on the rough boundary by $\rep_{\Zb_2}\simeq \fun_{\rep_{\Zb_2}}(\hilb, \hilb)$. The bulk excitations are given by the unitary modular tensor category $Z(\rep_{\Zb_2}):=\fun_{\rep_{\Zb_2}|\rep_{\Zb_2}}(\rep_{\Zb_2}, \rep_{\Zb_2})$, which is the Drinfeld center of $\rep_{\Zb_2}$. The tensor product\footnote{We have secretly used the fact that $\rep_{\Zb_2}$ is a symmetric fusion 1-category such that $\rep_{\Zb_2}\simeq \rep_{\Zb_2}^\rev$. See Eq.\,(\ref{eq:lw-mod-dim-red}) for a more general and precise expression of this tensor product.}:
$$
\EA_2\boxtimes_{\EC_3} \EB_2=\fun_{\rep_{\Zb_2}}(\rep_{\Zb_2}, \rep_{\Zb_2}) \boxtimes_{Z(\rep_{\Zb_2})}
\fun_{\rep_{\Zb_2}}(\hilb, \hilb)
$$
is defined by Tambara's tensor product \cite{tambara} between a left $Z(\rep_{\Zb_2})$-module and a right $Z(\rep_{\Zb_2})$-module (see also \cite{eno2009}). Moreover, this tensor product has a natural monoidal structure (see \cite[Thm.\,5.8]{kong-zheng}) and we have
\begin{align}
\EA_2\boxtimes_{\EC_3} \EB_2 &= \fun_{\rep_{\Zb_2}}(\rep_{\Zb_2}, \rep_{\Zb_2}) \boxtimes_{Z(\rep_{\Zb_2})}
\fun_{\rep_{\Zb_2}}(\hilb, \hilb) \nn
&\simeq \fun_{\hilb}(\rep_{\Zb_2} \boxtimes_{\rep_{\Zb_2}} \hilb,\,\, \rep_{\Zb_2} \boxtimes_{\rep_{\Zb_2}} \hilb) \nn
&\simeq \fun_{\hilb}(\hilb,\,\, \hilb) \simeq \hilb,  \label{eq:toric-code-1}
\end{align}
where the first $\simeq$ was proved in \cite[Thm.\,5.8]{kong-zheng}.

\begin{figure}[bt]
$$
\raisebox{-50pt}{
  \begin{picture}(95,190)
   \put(0,0){\scalebox{6}{\includegraphics{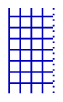}}}
   \put(0,0){
     \setlength{\unitlength}{.75pt}\put(-18,-19){
     \put(118, 230)       {\scriptsize $ \mbox{rough boundary} $}
     \put(-68, 230)     {\scriptsize $ \mbox{smooth boundary} $}
          }\setlength{\unitlength}{1pt}}
  \end{picture}}
\quad\quad\quad\quad\quad\quad
\raisebox{-50pt}{
  \begin{picture}(105,120)
   \put(0,0){\scalebox{6}{\includegraphics{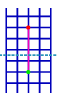}}}
   \put(0,0){
     \setlength{\unitlength}{.75pt}\put(-18,-19){
     \put(72, 183)       {\scriptsize $ e $}
     \put(72, 65)       {\scriptsize $ e $}
     \put(5,110)     {\scriptsize $m$}
     \put(155,110)    {\scriptsize $m$}
     }\setlength{\unitlength}{1pt}}
  \end{picture}}
$$
$$
(a) \quad\quad\quad\quad \quad \quad\quad\quad\quad \quad
\quad\quad\quad\quad  (b)
$$
\caption{Above figures illustrate two physical configurations of toric code model bounded by two gapped boundaries discussed in Example\,\ref{expl:toric-code}. In (a), two boundaries are different. In (b), both boundaries are the smooth boundary.  }
\label{fig:toric-code}
\end{figure}

\item In (b), both boundaries are the smooth boundaries. In this case, $m$-particle is clearly condensed in the $2$D phase. It seems that there is only $e$-particles living in the $2$D topological order. However, the string (the purple line in Fig.\,\ref{fig:toric-code} (b)) that creates a pair of $e$-particles at its ends can be detected by a string (the dotted line) that create a pair of condensed $m$-particles. It means that the $e$-string is a $2$D ``vacuum" that is different from the trivial string between two null-particles $1$. Moreover, the two $e$-particles (the red dot and the green dot in Fig.\,\ref{fig:toric-code} (b)) should be viewed as domain walls between two different vacuums, should be treated as different types of particles. They form a particle and anti-particle pair. As a consequence, the $2$D phase is described by a unitary multi-fusion 1-category
\be  \label{eq:multi-fusion-toric-cord}
\hilb \times M_{2\times 2} = \left( \begin{array}{cc}  \hilb & \hilb  \\
 \hilb & \hilb  \end{array} \right).
\ee
where two diagonal subcategories are the usual vacuum state (any vertical blue line in Fig.\,\ref{fig:toric-code} (b) and the $2$D vacuum state given by the purple string in Fig.\,\ref{fig:toric-code} (b), the two diagonal subcategories are the domain walls between two vacuums.

\smallskip
Above result is also guaranteed by abstract nonsense. Indeed, by \cite[Thm.\,5.8]{kong-zheng}, the $2$D phase obtained via dimensional reduction is given by
\begin{align}
\EA_2\boxtimes_{\EC_3} \EB_2 &= \fun_{\rep_{\Zb_2}}(\rep_{\Zb_2}, \rep_{\Zb_2}) \boxtimes_{Z(\rep_{\Zb_2})}
\fun_{\rep_{\Zb_2}}(\rep_{\Zb_2}, \rep_{\Zb_2}) \nn
&\simeq \fun_{\hilb}(\rep_{\Zb_2} \boxtimes_{\rep_{\Zb_2}} \rep_{\Zb_2},\,\, \rep_{\Zb_2} \boxtimes_{\rep_{\Zb_2}} \rep_{\Zb_2}) \nn
&\simeq  \fun_{\hilb}(\rep_{\Zb_2},\,\, \rep_{\Zb_2}), \label{eq:toric-code-ACB}
\end{align}
which is exactly the unitary multi-fusion 1-category given in (\ref{eq:multi-fusion-toric-cord}).

When the distance between two boundaries is small, the tunneling of an $m$-particle from one boundary to the other is a local operator. In this case, the ground state degeneracy can be easily lifted by introducing this tunneling effect into the Hamiltonian.
Therefore, the phase described by unitary multi-fusion 1-category (\ref{eq:multi-fusion-toric-cord}) (or (\ref{eq:toric-code-ACB})) is unstable and will flow to the only closed and stable 2D phase $\one_2$. On the other hand, if we keep the distance between two boundaries large even in the thermodynamic limit, then the physical configuration depicted in Fig.\,\ref{fig:toric-code} (b) with 2-fold ground state degeneracy is stable (under local perturbations) as a $3$D topological order with two boundaries \cite{wang-wen}.

\enu
Before we end this example, we would like to comment that the unitary multi-fusion 1-category defined in (\ref{eq:multi-fusion-toric-cord}) is also closed because it is realizable by a lattice model in $2$D. This is consistent with the mathematical result that the Drinfeld center of the unitary multi-fusion category defined in (\ref{eq:multi-fusion-toric-cord}) is trivial \cite{eno2009}. Since they describe stable 3D topological orders (with possible gapped boundaries), they also play important role in our study of topological orders. 

}
\end{expl}

\begin{figure}[bt]
$$
\raisebox{-50pt}{
  \begin{picture}(95,190)
   \put(0,0){\scalebox{1.5}{\includegraphics{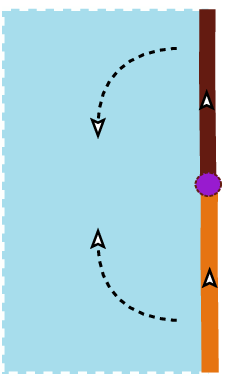}}}
   \put(0,0){
     \setlength{\unitlength}{.75pt}\put(-18,-19){
     \put(128, 240)       {\scriptsize $ \EC_\EM^\vee:=\fun_\EC(\EM, \EM)^\rev $}
     \put(132, 10)     {\scriptsize $ \EC_\EN^\vee:=\fun_\EC(\EN, \EN)^\rev $}
     \put(150, 128)   {\scriptsize $(\fun_\EC(\EM, \EN), f)$}
     \put(40, 128)     {\scriptsize $ Z(\EC) $}
          }\setlength{\unitlength}{1pt}}
  \end{picture}}
\quad\quad\quad\quad\quad\quad\quad
\raisebox{-50pt}{
  \begin{picture}(105,160)
   \put(0,35){\scalebox{1.5}{\includegraphics{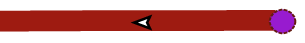}}}
   \put(0,35){
     \setlength{\unitlength}{.75pt}\put(-18,-19){
     \put(170, 65)       {\scriptsize $ (\fun_\EC(\EM, \EN),f) $}
     \put(100, 95) {\scriptsize  $\EE_2$}
     }\setlength{\unitlength}{1pt}}
  \end{picture}}
$$
$$
(a) \quad\quad\quad\quad \quad \quad\quad\quad\quad \quad
\quad\quad\quad\quad\quad\quad  (b)
$$
\caption{Above figures illustrate the dimensional reduction process in a Levin-Wen model discussed in Example\,\ref{expl:lw-mod}. $\EE_2$ is given by (\ref{eq:lw-mod-dim-red}).
}
\label{fig:lw-mod}
\end{figure}

\begin{expl}  \label{expl:lw-mod} {\rm
We give more examples of $\EA_n\boxtimes_{\EC_{n+1}} \EB_n$ for $n=1$ in Levin-Wen models \cite{kitaev-kong}. Consider a lattice model depicted in Fig.\,\ref{fig:lw-mod} (a), the bulk lattice defined by a unitary fusion 1-category $\EC_2$, the upper/lower boundary lattice is defined by a unitary indecomposable $\EC$-module $\EM$/$\EN$. The excitations in the bulk are given by the $Z_2$-center (the Drinfeld center) $Z_2(\EC_2)$ of $\EC_2$, the excitations on the $\EM$-boundary by $\EC_\EM^\vee:=\fun_\EC(\EM,\EM)^\rev$, those on $\EN$-boundary by $\EC_\EN^\vee$ and the defect junction (the purple dot) by a unitary $\EC$-module functor $f\in \fun_\EC(\EM, \EN)$. When the defect junction is viewed as a 1D topological order (by including the action of nearby excitations on $f$), it is given by $(\fun_\EC(\EM, \EN), f)$\footnote{The action of nearby excitations of $f$ actually form a subcategory $\fun_\EC(\EN,\EN)f\fun_\EC(\EM,\EM)$ of $\fun_\EC(\EM, \EN)$. This subcategory is equivalent to the data $(\fun_\EC(\EM, \EN), f)$.}, which is an $E_0$-algebra (see Remark\,\ref{rema:codim}). By a dimensional reduction process, i.e. folding the two boundaries along two dotted arrows, we obtain (b) in Fig.\,\ref{fig:lw-mod}, where
\begin{align} \label{eq:lw-mod-dim-red}
\EE_2 &= \fun_\EC(\EN, \EN)^\rev \boxtimes_{Z(\EC)} \fun_\EC(\EM, \EM)^\rev.
\end{align}
On the other hand, when we fold $\EN$-boundary upwards and flip its orientation, the left $\EC$-module $\EN$ becomes the right $\EC$-module $\EN^\op$. According to \cite{kitaev-kong}, we should also have
\begin{align}  \label{eq:lw-mod-dim-red-2}
\EE_2 &= \fun_{\hilb}(\EN^\op\boxtimes_\EC \EM, \EN^\op\boxtimes_\EC \EM)^\rev \nn
&\simeq \fun_{\hilb}(\fun_\EC(\EN, \EM), \fun_\EC(\EN, \EM))^\op \nn
&\simeq \fun_{\hilb}(\fun_\EC(\EN, \EM)^\op, \fun_\EC(\EN, \EM)^\op), \nn
&\simeq \fun_{\hilb}(\fun_\EC(\EM, \EN), \fun_\EC(\EM, \EN))
\end{align}
where we have used the identity $\EN^\op \boxtimes_\EC \EM \simeq
\fun_\EC(\EN, \EM)$ \cite{eno2009}. Indeed, the compatibility of (\ref{eq:lw-mod-dim-red}) and (\ref{eq:lw-mod-dim-red-2}) was proved in \cite[Thm.\,5.8]{kong-zheng}. Note that $\EE_2$ is a unitary multi-fusion 1-category. As a 2D topological order, it should be closed because it is the \bulk of the 1D topological order $(\fun_\EC(\EM, \EN),f)$. This is consistent with the mathematical fact that the $Z_2$-center of $\EE_2$ is trivial \cite{eno2009}. Moreover, as we will show in Sec.\,\ref{sec:center}, $\EE_2$ is also the $Z_1$-center of $(\fun_\EC(\EM, \EN),f)$. Note also that the dimension of ground state degeneracy on $S^1$ is the dimension of $\hom_\EE(1_\EE, 1_\EE)$, which is the number of simple objects in $\fun_\EC(\EM, \EN)$.
In particular, when $\EN=\EM=\EC$, $\EE=\fun(\EC, \EC)$, and if $\EC=\rep_{\Zb_2}$, $\EE_2$ is nothing but the topological phase constructed in Fig.\,\ref{fig:toric-code} (b) (recall Eq. (\ref{eq:toric-code-ACB})). It describes a 2-unstable $2$D phase and can flow to the only 2-stable closed $2$D phase $\one_2$. This is consistent with recent works \cite{wang-wen,hy,lww}. One can cook up more general examples from Levin-Wen models enriched by defects as depicted in Fig.\,\ref{fig:lw-defect-duality} (see Example\,\ref{expl:lw-ext-domain-wall}).
}
\end{expl}

\void{
\begin{rema} \label{rema:unit-map} {\rm
In Example\,\ref{expl:lw-mod}, the defect junction is given by $\fun_\EC(\EN, \EM)$. But in \cite{kitaev-kong}, a single defect junction is given by an object $\phi$ in the category $\fun_\EC(\EN, \EM)$ instead of the whole category. The reason for choosing $\fun_\EC(\EN, \EM)$ is because we want to view the defect junction not as a single excitation but as an anomalous $0+1$D topological phase, which is defined by a neighborhood of this junction (see \cite[Lem.\,1]{kong-wen}). Therefore, one must include the entire image under the actions of $\fun_\EC(\EM, \EM)$ and $\fun_\EC(\EN, \EN)$ on $\phi$ from two sides. It is known that this action cover the entire category $\fun_\EC(\EN, \EM)$ because $\fun_\EC(\EN, \EM)$ is an indecomposable $\fun_\EC(\EN, \EN)$-$\fun_\EC(\EM, \EM)$-bimodule \cite{eo}. Note also that $\phi \in \fun_\EC(\EN, \EM)$ defines the unit map $\one_1\to \fun_\EC(\EN, \EM)$ that is needed to define the $1$D topological order. 
}
\end{rema}
}

It is also clear that the following identities (recall the Minimal Assumption Remark\,\ref{rema:minimal}): 
\begin{align}
\one_n \boxtimes \EA_n &= \EA_n = \EA_n \boxtimes \one_n, \nn
(\EA_n \boxtimes_{\EB_{n+1}} \EC_n) \boxtimes_{\ED_{n+1}} \EE_n
&= \EA_n \boxtimes_{\EB_{n+1}} (\EC_n \boxtimes_{\ED_{n+1}} \EE_n) \label{eq:boxtimes-associative}
\end{align}
hold. We simply denote the two sides of the equation (\ref{eq:boxtimes-associative}) by $\EA \boxtimes_\EB \EC \boxtimes_\ED \EE$ in the rest of this paper. Moreover, if $\EB_{n+1}$ and $\ED_{n+1}$ are closed, we have the following identity:
\be  \label{eq:BCD}
P_n(\EB_{n+1}) \boxtimes_{\EB_{n+1}} \EC_n = \EC_n = \EC_n \boxtimes_{\ED_{n+1}} P_n(\ED_{n+1}).
\ee

\begin{rema} {\rm
When $n=1$, $\EA_1, \EC_1, \EE_1$ are unitary 1-categories. In this case, the associativity (\ref{eq:boxtimes-associative}) is just the associativity of Tambara's tensor product \cite{tambara,eno2009}. When $n=2$, $\EA_2, \EC_2, \EE_2$ are unitary multi-fusion 1-categories and $\EB_3$ a unitary braided fusion $1$-category, the tensor products give again unitary multi-fusion 1-categories (see \cite{kong-zheng}). In this case, the associativity is guaranteed. Higher $n$ cases are not completely known (see \cite{lurie2}).}
\end{rema}

\begin{rema} \label{rema:hCat} {\rm
There is another kind of interesting dimensional reductions. Given a unitary $n$-category $\EC_n$, i.e. an $n$D topological order, we obtain the so-called the homotopy $k$-category $\mathrm{h}_k\EC$ (for $0\leq k< n$), which is defined as the $k$-category obtained from $\EC_n$ by throwing away all $k+1$ and higher morphisms in $\EC_n$ and replacing old $k$-morphisms by their equivalence classes (\cite{lurie2}). The physical meaning of $\mathrm{h}_k\EC$ (together with a distinguished object) is a $k$D topological order, which is obtained in the eye of an observer living a $k$D plane of original $n$D topological order such that the plane intersects transversally with all $l$-codimensional excitations in $\EC_n$ for $0\leq l\leq k$ but not contains any higher codimensional excitations. This kind of dimensional reductions can be very useful. For example, it suggests that a simple $\EC_n$ is anomalous if $\mathrm{h}_k\EC$ is anomalous for some $1< k< n$. For example, if $\mathrm{h}_2\EC$ is anomalous, then $\EC_n$ contains at least one non-trivial domain walls, which is not detectable by mutual braidings. Hence, $\EC_n$ is anomalous too. This is the physical intuition behind Conjecture\,\ref{conj:closed-f=bf}. 
}
\end{rema}

\subsection{Closed domain walls and the category $\TO_n^{closed-wall}$}  \label{sec:wall}

\begin{figure}[tb]
 \begin{picture}(150, 90)
   \put(150,10){\scalebox{2}{\includegraphics{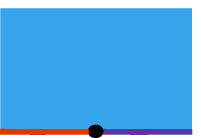}}}
   \put(110,-55){
     \setlength{\unitlength}{.75pt}\put(-18,-19){
     \put(85, 98)       { $\EC_n$}
     \put(180, 98)       { $\ED_n$}
     \put(130,96)      { $\EM_{n-1}$}
     \put(130, 160)    { $ \EA_{n+1} $}
   }\setlength{\unitlength}{1pt}}
  \end{picture}
\caption{A closed gapped domain wall (or simply a domain wall) $\EM_{n-1}$ between the $\EC_n$-phase and the $\ED_n$-phase sharing the same bulk $\EA_{n+1}$. 
}
\label{fig:closed-wall}
\end{figure}

By a {\it closed} or {\it anomaly-free} gapped domain wall $\EM_{n-1}$ between $\EC_n$ and $\ED_n$, or a $\EC_n$-$\ED_n$-wall, we mean an $(n-$$1)$D topological order connecting the $\EC_n$-phase and the $\ED_n$-phase and sharing the same bulk as depicted in Fig.\,\ref{fig:closed-wall}. We postpone the precise definition of the anomalous gapped domain wall to Sec.\,\ref{sec:Z-functor}. All domain walls occur before Sec.\ref{sec:Z-functor} are closed. Therefore, we abbreviate a closed domain wall to a domain wall until Sec.\,\ref{sec:Z-functor} unless we want to emphasize. 

\medskip
If $\EC_n$ and $\ED_n$ are simple, the domain wall can be simple or composite. If $\EC_n$ or $\ED_n$ is composite, the wall must be composite. We denote the $n-$1D topological order determined by the trivial wall $\id_\ast$ in a simple $n$D topological order $\EC_n$ by $\id_{\EC_n}$, i.e. $\id_{\EC_n}=P_{n-1}(\EC_n)$. 

\begin{defn} {\rm
The category $\TO_n^{closed-wall}$ is defined as the $n$-category with objects given by simple $n$D topological orders, 1-morphisms by all 1-codimensional (disk-like) gapped closed walls between topological orders, \ldots, $l$-morphisms by $l$-codimensional (disk-like) gapped closed walls between $(l-$$1)$-codimensional closed walls, \ldots, $n$-morphisms are instantons. The composition of $l$-morphisms are given by the tensor products given in Eq.\,(\ref{eq:AxB})
} 
\end{defn}

\begin{rema} {\rm
Note that the category $\TO_n^{closed-wall}$ contains only simple topological orders, gapped closed walls and closed walls between walls. We exclude all composite topological orders because they break the functoriality of $\Z_n$ (see \cite[Remark\, 5.21]{kong-zheng} for $n=2$ cases). The category $\TO_n^{closed-wall}$ is not a unitary $n$-category because it does not satisfy Axiom 5 in Def.\,\ref{def:unitary-n-cat}. 
}
\end{rema}

The time-reverse of a $\EC_n$-$\ED_n$-wall, denoted by $\EM^\op$ (or $\overline{\EM}$ see Remark\,\ref{rema:time-reverse}), is automatically a $\ED_n$-$\EC_n$-wall. It gives the category $\TO_n^{closed-wall}$ a duality structure. The tensor product $\boxtimes$ (recall the first equation in Eq.\,(\ref{eq:AxB})) provides $\TO_n^{closed-wall}$ the structure of a symmetric monoidal $n$-category. 

\medskip
An invertible gapped domain wall between two simple topological orders $\EC_n$ and $\ED_n$ is just an invertible morphism in $\TO_n^{closed-wall}$, and can be defined inductively.
\begin{defn} {\rm
A gapped domain wall $\EM_{n-1}$ between $\EC_n$ and $\ED_n$ is called {\it invertible} if there is an $(n-2)$D invertible domain wall between $\EM\boxtimes_{\ED_n} \EM^\op$ and  $\id_{\EC_n}$, denoted by $\EM\boxtimes_{\ED_n} \EM^\op \simeq \id_{\EC_n}$, and another $(n-2)$D invertible domain wall between $\EM^\op \boxtimes_{\EC_n} \EM$ and $\id_{\EC_n}$, denoted by $\EM^\op \boxtimes_{\EC_n} \EM \simeq \id_{\EC_n}$ (see Def.\,\ref{def:isomorphism}).  
}
\end{defn}

\begin{rema} {\rm
Note that the Minimal Assumption Remark\,\ref{rema:minimal} does not trivialize the higher dimensional domain walls. Their physical meaning is still preserved.  
}
\end{rema}

An invertible domain wall between two phases $\EC_n$ and $\ED_n$ is also transparent. It means that excitations in $\EC_n$-phase can tunnel through the wall  and become excitations in $\ED_n$-phase. This tunneling is an invertible process and preserves the fusion and braiding properties. Indeed, for an invertible $\EC_n$-$\ED_n$-wall $\EM_{n-1}$, the topological excitations on one side can tunnel through the invertible wall to the other side. For excitations of codimension higher than 1 and $n>1$, this tunneling process is invertible and is depicted in Fig.\,\ref{fig:tunneling}:
\begin{figure}
$$
 \raisebox{-50pt}{
  \begin{picture}(100,110)
   \put(50,8){\scalebox{1.5}{\includegraphics{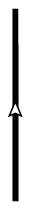}}}
   \put(50,8){
     \setlength{\unitlength}{.75pt}\put(-151,-235){
     \put(110,330)  {\scriptsize $ \EC_n $}
     \put(185,330)  {\scriptsize $ \ED_n $}
     \put(148,355)  {\scriptsize $ \EM_{n-1} $}
     \put(148,225)  {\scriptsize $\EM_{n-1} $}
     \put(110, 290) {\scriptsize $\times$}
     }\setlength{\unitlength}{1pt}}
  \end{picture}}
  \rightsquigarrow
   \raisebox{-50pt}{
   \begin{picture}(130,110)
   \put(50,8){\scalebox{1.5}{\includegraphics{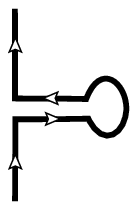}}}
   \put(50,8){
     \setlength{\unitlength}{.75pt}\put(-151,-235){
     \put(120,290)  {\scriptsize $ \EC_n $}
     \put(203, 330)  {\scriptsize $\ED_n$}
     \put(170,307)  {\scriptsize $ \EM $}
     \put(170,270)  {\scriptsize $ \EM $}
     \put(148,355)  {\scriptsize $ \EM_{n-1} $}
     \put(148,225)  {\scriptsize $\EM_{n-1} $}
     \put(203, 290) {\scriptsize $\times$}
     }\setlength{\unitlength}{1pt}}
  \end{picture}}
  \rightsquigarrow
  \raisebox{-50pt}{
   \begin{picture}(130,110)
   \put(50,8){\scalebox{1.5}{\includegraphics{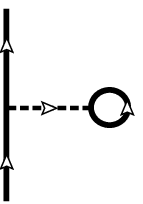}}}
   \put(50,8){
     \setlength{\unitlength}{.75pt}\put(-151,-235){
     \put(120,290)  {\scriptsize $ \EC_n $}
     \put(170,330)  {\scriptsize $ \ED_n $}
     \put(170,250)  {\scriptsize $ \ED_n $}
     \put(148,355)  {\scriptsize $ \EM_{n-1} $}
     \put(148,225)  {\scriptsize $\EM_{n-1} $}
     \put(230,288) {\scriptsize $\EM_{n-1} $}
     \put(210, 288) {\scriptsize $\times$}
     }\setlength{\unitlength}{1pt}}
  \end{picture}}
$$
\caption{The ``$\times$" in the graphs represents a topological excitation in $\EC_n$-phase. The two horizontal lines in the second picture have the opposite orientation. They represent the domain wall $\EM$ and its inverse $\EM^\op$. Their tensor product gives the trivial wall $\id_{\ED_n}$ depicted as the dotted line in the third graph.}
\label{fig:tunneling}
\end{figure}
where $\times$ represents a topological excitation in $\EC_n$-phase and the dotted line in the third graph is the trivial domain wall $\id_{\ED_n}$ in $\ED_n$-phase. For 1-codimensional excitations and $n\leq 1$, it is similar. This tunneling process is completely invertible. Moreover, it is clear that it also respects the fusion and braiding \cite{kitaev-kong}. In this case, we also denote $\EC_n \simeq \ED_n$. 


This physical definition of invertible domain wall becomes a mathematical one if we identify an $n+$1D topological order with a unitary multi-fusion $n$-category, $\boxtimes_{\EC_{n+1}, \ED_{n+1}}$ with mathematical tensor products, and define the invertible domain wall mathematically in the lowest dimension. We expect that an invertible domain wall is equivalent to an invertible unitary $n$-functor. 

\begin{expl} \label{expl:inv-wall} {\rm
Using Def.\,\ref{def:cat-TO} and Example\,\ref{expl:codim=2}, we prove that invertible domain walls are indeed invertible unitary $n$-fucntors in low dimensional cases. 
\bnu
\item A 0D domain wall between two simple 1D phases $A$ and $B$ (two simple algebras over $\Cb$) is given by a $A$-$B$-bimodule $V$ (together with a distinguished element $v$), i.e. an $E_0$-algebra $(V,v)$ (recall Example\,\ref{expl:codim=2}). A -1D ``wall" between two such $0$D topological orders $(U,u)$ and $(V,v)$ is an $A$-$B$-bimodule map $f: U\to V$ such that $f(u)=v$ (recall Example\,\ref{expl:codim=2}). This -1D ``wall" is invertible if $f$ is an invertible bimodule map. 

If $(M, m)$ is an invertible 0D wall between $A$ and $B$, it implies that $M$ is an invertible $A$-$B$-bimodule and more. 
It turns out that it implies that $A\simeq B$ as algebras as we expected in the discussion of Fig.\,\ref{fig:tunneling}. More explicitly, let $A=\text{End}(U)$ and $B=\text{End}(V)$ for two finite dimensional vector spaces $U$ and $V$, the only invertible $A$-$B$-bimodule is given by $\hom_\Cb(V, U)$ with its inverse given by the $B$-$A$-bimodule $\hom_\Cb(U,V)$. If $(\hom_\Cb(V,U), f: V\to U)$ is an invertible wall, then there is $g:U \to V$ such that 
\begin{align}
&(\hom_\Cb(V,U), f)\otimes_B (\hom_\Cb(U,V), g) \simeq (A,\id_U)  \nn
&(\hom_\Cb(U,V), g)\otimes_A (\hom_\Cb(V,U), f) \simeq (B, \id_V), \nonumber
\end{align}
which means that the canonical bimodule isomorphism $\hom_\Cb(V, U)\otimes_B\hom_\Cb(U,V) \simeq A$ must map $f\otimes_B g$ to $f\circ g=\id_V$ and the canonical bimodule isomorphism $\hom_\Cb(U, V)\otimes_A\hom_\Cb(V,U) \simeq B$ must map $g\otimes_A f$ to $g\circ f=\id_U$. As a consequence, $g$ and $f$ are invertible, and, therefore, $A\simeq B$ as algebras. It is easy to check that an invertible 0D wall between $A$ and $B$ is equivalent to an algebra isomorphism between $A$ and $B$. 

\item A 1D domain wall between two simple 2-stable 2D topological orders, which are given by two unitary fusion 1-categories $\EA_2$ and $\EB_2$, must be a unitary $\EA$-$\EB$-bimodule $\EE_1$ equipped with a distinguished object $x$. $\boxtimes_{\EA,\EB}$ in this case is just the Tambara's tensor product also denoted by $\boxtimes_{\EA,\EB}$. If $(\EE,x)$ is invertible, it implies that $\EA$ and $\EB$ are Morita equivalent and more. In particular, we can set $\EA_2=\fun_\EC(\EM, \EM)$, $\EB_2=\fun_\EC(\EN,\EN)$ and $\EE_1=\fun_\EC(\EM, \EN)$ for some unitary fusion 1-category $\EC$ and indecomposable unitary $\EC$-modules $\EM$ and $\EN$ (for example, one can simply take $\EC=\EA$ and $\EM=\EA$) \cite{eno2008}. The inverse of the wall $\EE$ is $\EE^\op\simeq \fun_\EC(\EN, \EM)$. Let $f:\EM \to \EN$ be a unitary $\EC$-module functor. Then the domain wall $(\fun_\EC(\EM, \EN), f)$ is invertible if and only if there is $g\in \fun_\EC(\EN, \EM)$ such that $g\circ f \simeq \id_\EM$ and $f\circ g \simeq \id_\EN$. In other words, $f$ must be an invertible functor, and, therefore, $\EA \simeq \EB$ as unitary fusion 1-categories. It is easy to check that an invertible 1D domain wall is equivalent to a unitary monoidal equivalence between unitary fusion 1-categories. 

\item In 3D, it was shown explicitly in 2+1D Levin-Wen models that invertible gapped walls 1-to-1 correspond to  unitary braided monoidal equivalences between two exact unitary braided fusion 1-categories \cite{kitaev-kong,eno2009}. 

\enu
We expect that high dimensional cases are similar.
}
\end{expl}

If the domain wall $\EM_{n-1}$ is not invertible, then tunneling process depicted in Fig.\,\ref{fig:tunneling} still makes sense but with the dotted line replaced by $\EM^\op \boxtimes_\EC \EM \simeq \oplus_i \EN_i$, where $\EM^\op \boxtimes_\EC \EM$ is a composite topological order and $\EN_i$ is an indecomposable component of $\EM^\op \boxtimes_\EC \EM$. In this case, a topological excitation ``tunnel" through the wall, always pulls a string (labeled by $\oplus_i \EN_i$) behind it \cite{kitaev-kong}.

\begin{expl}  {\rm
We give some examples of closed walls in Levin-Wen models.
\bnu
\item In the toric code model, any line in the bulk lattice can be viewed as a trivial (obviously invertible/transparent) domain wall. A non-trivial invertible domain wall was constructed in Figure\,1 in \cite{kitaev-kong} (see also the dotted line in Fig.\,\ref{fig:external-wall}). This domain wall also gives the EM duality of the bulk phase.

\item In the Levin-Wen models constructed in \cite{kitaev-kong}, a gapped domain wall can be constructed from a bimodule over fusion categories. More precisely, if the bulk lattices on the two side of a domain wall are defined by unitary fusion 1-categories $\EA_2$ and $\EB_2$, then the lattice on the domain wall can be defined by a $\EA$-$\EB$-bimodule $\EE$. In this case, the bulk excitations on the two sides of the wall are given by the Drinfeld centers $Z_2(\EA_2)$ and $Z_2(\EB_2)$, respectively. The excitations on the $\EM$-wall is given by $\fun_{\EA|\EB}(\EM, \EM)$. When the bulk excitations move close to the wall, they become wall excitations, this process defines two bulk-to-wall maps:
$$
Z_2(\EA) \xrightarrow{L} \fun_{\EA|\EB}(\EM, \EM) \xleftarrow{R}  Z_2(\EB).
$$
More precisely, $L$ and $R$ are two monoidal functors defined by
\begin{align}
L: \quad (\EA \xrightarrow{f} \EA)\quad &\mapsto \quad (\EM \simeq \EA \boxtimes_\EA \EM \xrightarrow{f\boxtimes_\EA \id_\EA} \EA \boxtimes_\EA \EM \simeq \EM), \nn
R: \quad (\EB \xrightarrow{g} \EB) \quad &\mapsto \quad (\EM \simeq \EM \boxtimes_\EB \EB \xrightarrow{\id_\EB \boxtimes_\EB g}  \EM \boxtimes_\EB \simeq \EM). \nonumber
\end{align}

If the bimodule $\EM$ is invertible in the sense that $\EM\boxtimes_\EB \EM^\op \simeq \EA$ and $\EM^\op \boxtimes_\EB \EM \simeq \EB$. Then both $L$ and $R$ are monoidal equivalences. Then the excitations tunneling from the left side to the right side is given by the monoidal equivalence $R^{-1}\circ L: Z(\EA) \to Z(\EB)$. Moreover, $R^{-1}\circ L$ respects the braiding, i.e. a braided monoidal equivalence. Therefore, an invertible $\EA$-$\EB$-bimodule gives an isomorphism between $Z(\EA)$ and $Z(\EB)$. It was also known that there is an one-to-one correspondence between invertible $\EA$-$\EB$-bimodules and braided monoidal equivalences between $Z(\EA)$ and $Z(\EB)$ \cite{kitaev-kong,eno2009}. If $\EM$ is not invertible, we have $\EM^\op \boxtimes_\EA \EM \simeq \oplus_i \EN_i$ as $\EB$-$\EB$-bimodules. A topological particle pass the wall always pull strings labeled by $\EN_i$. The real tunneling process amounts to cutting all the strings with non-trivial labels $\EN_i$ (except $\EB$) and setting the particle free, and is given by the functor $R^\vee\circ L$.

\enu
}
\end{expl}

It is known that 2D domain walls between two closed phases $\EC_3$ and $\ED_3$, i.e. two unitary non-degenerate braided fusion 1-categories, are classified by Lagrangian algebras in $\EC_3\boxtimes \overline{\ED}_3$, where $\overline{\ED}_3$ is the same braided fusion 1-category as $\ED_3$ but with the braiding defined by the anti-braiding of $\ED_3$. If the bulk-to-wall maps are $\EC_3 \xrightarrow{L} \EE_2 \xleftarrow{R} \ED_3$, or equivalently, $\EC_3\boxtimes \overline{\ED}_3 \xrightarrow{L\boxtimes \overline{R}} \EE_2$, then the Lagrangian algebra is the determined by $(L\boxtimes \overline{R})^\vee(1_\EE)$, where $(L\boxtimes \overline{R})^\vee$ is the right adjoint functor of $L\boxtimes R$ and $1_\EE$ is the tensor unit of $\EE_2$.
We expect that the same result holds for higher dimensional cases but with unitary (braided) fusion 1-categories replaced by unitary (braided) fusion $n$-categories. 

\begin{rema} {\rm
Gapped domain walls in two closed 3D topological orders $\EA_3$ and $\EB_3$ can all be obtained from condensations of lower dimensional excitations. When $\EA_3=\EB_3$, such walls are not elementary and should be excluded in the minimal categorical description of the topological order \cite{kong-wen}. 
}
\end{rema}

\begin{rema} {\rm
Gapped boundaries, walls or defects in TQFTs have been extensively studied from various perspectives in literature (see for example \cite{ffrs-defect,mw,kapustin-saulina-2,dkr2,fsv,fsv2,fv,fs} and references therein). 
}
\end{rema}

\section{The category $\TO_n^{fun}$ of topological orders}  \label{sec:TO-fun}

In this section, we first explain the uniqueness of the bulk of a give gapped boundary theory in Sec.\,\ref{sec:unique-bulk}. This uniqueness defines the \bulk, denoted by $\Z_n(\EC_n)$, of a given boundary $\EC_n$. Then we use it to give a physical definition of the notion of a morphism between two topological orders of the same dimension in Sec.\,\ref{sec:morphisms}. We discuss higher morphisms in Sec.\,\ref{sec:higher-morphisms}.  

\subsection{The unique-bulk hypothesis}  \label{sec:unique-bulk}
In general, an anomalous $n$D topological order $\EC_n$ can always be realized as a defect in a higher dimensional (possibly trivial) topological order. This realization is almost never unique. But we can always reduce it via a process of dimensional reduction to a gapped boundary of a closed $(n+1)$D topological order (see Fig.\,\ref{fig:dim-reduction}). We believe that such obtained closed $(n+1)$D topological order is unique. We denote the unique bulk topological order by $\Z_n(\EC_n)$ and refer to it as the \bulk of $\EC_n$.

\medskip
Actually, it is an immediate consequence of Def.\,\ref{def:TO}. The key of this proof is that when an anomalous $n$D topological order is realized as a boundary of an $(n+1)$D topological order, it should remain the same $n$D phase in an arbitrary neighborhood of the boundary (see \cite[Lem.\,2]{kong-wen}). However, it is unclear if the microscopic definition Def.\,\ref{def:TO} is equivalent to the macroscopic definition by the fusion and braiding properties of its topological excitations. For this reason, we would like to refer to the uniqueness of the \bulk as the {\it unique-\bulk hypothesis}. We would like to provide more evidence of this hypothesis from the macroscopic point of view.

\smallskip
A closed 2+1D topological order is given by a unitary braided fusion 1-category, which is also a unitary modular tensor category (see for example \cite{eno2002,kitaev2,dgno}). An object in this 1-category is also called an anyon. If the 3D topological order is exact, then it is given by the Drinfeld center $Z_2(\ED_2)$ of a unitary fusion 1-category $\ED_2$. The phase $Z_2(\ED_2)$ can be realized as bulk excitations in a Levin-Wen model with a gapped boundary with the boundary excitations given by $\ED_2$ \cite{kitaev-kong}. If a non-degenerate unitary braided fusion 1-category $\EC_3$ is another bulk of the same boundary theory $\ED_2$, then it is clear that the bulk-to-boundary map, a monoidal functor $L: \EC \to \ED$ factors through $Z(\ED)$ \cite{fsv}, i.e. there exists a unique braided monoidal functor $\tilde{L}: \EC \to Z_2(\ED)$ such that following diagram:
$$
\xymatrix{  \EC_3 \ar[rr]^{\exists !\,\, \tilde{L}} \ar[rd]_L & & Z_2(\ED_2) \ar[ld]^{forget}   \\  & \ED_2 & }
$$
is commutative\footnote{It implies that $Z_2(\ED_2)=\EC \boxtimes \EC'$ \cite{dgno}, where $\EC'$ is the centralizer of the full subcategory $\EC$ in $Z_2(\ED_2)$. If the boundary $\ED_2$ is trivial, i.e. $\ED_2=\one_2$, then $\EC$ must be trivial as well. This result can be generalized to higher dimensions. The notion of center is defined by its universal property (see Thm.\,\ref{thm:universal}), by which, there is a canonical map from any closed bulk $\EC_{n+1}$ of a boundary $\ED_n$ to the center $Z_n(\ED_n)$ of $\ED_n$. This map should preserve all the fusion and braidings.
If $\ED_n=\one_n$, then $Z_n(\ED_n)=\one_{n+1}$. Since $\EC_{n+1}$ is closed, all excitations except the vacuum have non-trivial braiding with other excitations. The map $\EC_{n+1} \to \one_{n+1}$, preserving all the braidings, is impossible unless $\EC_{n+1}=\one_{n+1}$.}. To show that $\tilde{L}$ must be a braided monoidal equivalence we assume that a gapped boundary of a topological order can be obtained by an anyon condensation. 
If the bulk phase is described by a non-degenerate braided fusion 1-category $\EC$, an anyon condensation is determined by a connected commutative separable algebra $A$ in $\EC$, and the condensed phase is the given by the non-degenerate braided fusion 1-category $\EC_A^{loc}$ of local $A$-modules in $\EC$, and the quasi-particles confined on the gapped domain wall by the fusion 1-category $\EC_A$ of $A$-modules in $\EC$ \cite{kong-anyon}. Moreover, we have
$$
\EC \boxtimes \overline{\EC_A^{loc}}  \simeq Z_2(\EC_A).
$$
When the condensed phase is trivial, $\EC_A^{loc}=\hilb$, we obtain $\EC=Z_2(\EC_A)$. This gives a proof that, in 2+1D, the \bulk is uniquely determined by the Drinfeld center of the gapped boundary theory.

\medskip
To generalize above results to higher dimensions, we need develop a theory of condensation in higher dimensions. A possible approach via higher category theory is outlined in Sec.\,\ref{sec:conclusion}. In particular, we expect that a lot of above results also hold in higher dimensions with unitary braided fusion 1-categories replaced by unitary braided fusion $n$-categories.

\medskip
From now on, we assume that the \bulk of an $n$D topological order $\EC_n$ is unique and is denoted by $\Z_n(\EC_n)$. As a special case of (\ref{eq:BCD}), we have
$$
P_n(\Z_n(\EC_n)) \boxtimes_{\Z_n(\EC_n)} \EC_n=\EC_n.
$$
An immediate consequence of the unique-\bulk hypothesis is the following identity 
\be \label{eq:Z2=0}
\Z_{n+1}(\Z_n(\EC_n))=\one_{n+1}, 
\ee
which leads to some interesting complexes and cohomology groups \cite{kong-wen}. For example, the Witt group of non-degenerate braided fusion categories \cite{dmno} can be understood as the 3th cohomology group of a special complex \cite{kong-wen}. The result given by Eq.\,(\ref{eq:Z2=0}) is somewhat dual to the well-known fact that the boundary of the boundary of a manifold is empty. We hope that it leads to more mathematical results in the future.


\subsection{A morphism between two topological orders} \label{sec:morphisms}

In mathematics, a morphism of a mathematical structure preserves the structure. It is reasonable that a morphism of between two $n$D topological orders should preserve the universal fusion-braiding properties of topological excitations \cite{kong-wen}. 
So we expect that such a morphism is given by a unitary $n$-functor preserving the units.

There are some drawbacks of above definition. First, its physical meaning is not evident. Secondly, it depends on the categorical formulation of a topological order. Thirdly, it is not convenient to use for our purpose. So in this work, we would like to propose a physical definition which is independent of the mathematical formulation of a topological order. Note that a unitary $n$-functor describes a universal process of mapping excitations in one phase into another. This universal mapping process should be physically realizable. So we would like to define a morphism to be a physical realization of the universal mapping process. One way to achieve it is to let these topological excitations in $\EC_n$ to pass a region of spacetime in which the universal mapping process is realized. This spacetime naturally has dimension $n+1$ or higher. What it suggests is that a morphism can be realized by physics in at least one-dimensional higher.

\medskip
We introduce a physical definition of an isomorphism between two $n$D topological orders $\EC_n$ and $\ED_n$.
\begin{defn} \label{def:isomorphism} {\rm
An isomorphism $a: \EC_n \to \ED_n$ is an invertible gapped domain wall $\EM_{n-1}$, viewed as an $(n-1)$D topological order, between the $\EC_n$-phase and the $\ED_n$-phase. We also denote the trivial domain wall $P_{n-1}(\EC_n)$ in $\EC_n$-phase by $\id_\EC$, i.e. $\id_\EC=P_{n-1}(\EC_n)$. 
}
\end{defn}

\begin{rema} {\rm
Above definition is completely physical and independent of the categorical definition of a topological order. But from the perspective of Def.\,\ref{def:cat-TO}, an isomorphism should be an invertible unitary $n$-functor. Example\,\ref{expl:inv-wall} provides the proofs of a few low dimensional cases.
}
\end{rema}

By our definition, an isomorphism $a$ itself is an $(n-1)$D domain wall. For this reason, we also use the notation $a_{n-1}:=a$ to remind readers of its dimension. The composition of two isomorphisms $\EC_n \xrightarrow{a} \ED_n \xrightarrow{b} \EE_n$ is defined by the fusion of two gapped domain walls, i.e.
\be
b\circ a := b_{n-1} \boxtimes_{\ED_n} a_{n-1}.
\ee
Notice that the identity isomorphism indeed behaves like the usual identity maps. The composition of isomorphisms is associative.

\begin{defn} \label{def:morphism-2} {\rm
A morphism $f: \EC_n \to \ED_n$ is a pair $(f_n^{(0)}, f_{n-1}^{(1)})$ such that
\bnu
\item $f_n^{(0)}$ is a gapped domain wall, viewed as an $n$D topological order, between two $(n+1)$D topological orders $\Z_n(\EC_n)$ and $\Z_n(\ED_n)$,
\item $f_{n-1}^{(1)}: f_n^{(0)} \boxtimes_{\Z_n(\EC_n)}  \EC_n \xrightarrow{\simeq} \ED_n$ is an invertible domain wall, viewed as an $(n-1)$D topological order, between $f_n^{(0)} \boxtimes_{\Z_n(\EC_n)} \EC_n$ and $\ED_n$.
\enu
}
\end{defn}

The physical configuration associated to this morphism can be depicted as follows:
\be \label{pic:morphism-2}
\raisebox{-35pt}{
\begin{picture}(140, 80)
   \put(10,5){\scalebox{2}{\includegraphics{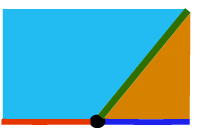}}}
   \put(10,5){
     \setlength{\unitlength}{.75pt}\put(-18,-40){
     \put(30,99)      {\footnotesize $\Z_n(\ED)$}
     \put(130, 72)     {\footnotesize $\Z_n(\EC)$}
     \put(168, 45)     {\footnotesize $\EC_n$.}
     \put(0,45)       {\footnotesize $\ED_n$}
     \put(109, 100)     {\footnotesize $f_n^{(0)}$}
     \put(66,60)        {\footnotesize $f_{n-1}^{(1)}$}
     }\setlength{\unitlength}{1pt}}
  \end{picture}}
\ee
Note that we must have $\Z_n(f_n^{(0)}) = \overline{\Z_n(\EC_n)} \boxtimes \Z_n(\ED_n)$.
For simplicity, in most parts of this work, we make $f_{n-1}^{(1)}$ implicit and use the following simplified picture instead.
\be \label{pic:morphism}
\raisebox{-5pt}{
\begin{picture}(140, 30)
   \put(0,5){\scalebox{2}{\includegraphics{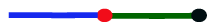}}}
   \put(-25,5){
     \setlength{\unitlength}{.75pt}\put(-18,-70){
     \put(60,89)      {\footnotesize $\Z_n(\ED)$}
     \put(150, 88)     {\footnotesize $\Z_n(\EC)$}
     \put(200, 88)     {\footnotesize $\EC_n$}
     \put(122, 90)     {\footnotesize $f_n^{(0)}$}
     }\setlength{\unitlength}{1pt}}
  \end{picture}}
\ee

\begin{rema} {\rm
In Def.\,\ref{def:morphism-2} and in (\ref{pic:morphism}), we chose to put $\Z_n(\EC_n)$ on the left side of $\EC_n$ according to the orientation convention in Levin-Wen models (see Fig.\,\ref{fig:lw-mod}).
}
\end{rema}

\begin{rema} {\rm
Although our definition of a morphism between two topological orders comes from our physical intuition, it does remind us of the notion of mapping cylinder in mathematics. This coincidence is perhaps not accidental because the low-energy effective theory of topological order is generally believed to be a TQFT. 
}
\end{rema}

\begin{expl}  {\rm
We give a few examples of morphisms that will be used later. Let $\EC_n$ and $\ED_n$ be two simple $n$D topological orders.
\bnu

\item We consider a special isomorphism $\id_{\EC_n}: \EC_n \to \EC_n$. It is nothing but the trivial domain wall $\id_{\EC_n}$ in a $\EC_n$-phase. It can be re-expressed as a morphism
$\id_{\EC_n}=(P_n(\Z_n(\EC_n)), \id_{\EC_n})$,
where $(\id_{\EC_n})_{n-1}^{(1)}$ is defined by the canonical isomorphism
$(\id_{\EC_n})_{n-1}^{(1)}: P_n(\Z_n(\EC_n)) \boxtimes_{\Z_n(\EC_n)} \EC_n = \EC_n
\xrightarrow{\id_{\EC_n}} \EC_n$.

\item Let $a: \EC_n \to \ED_n$ be an isomorphism. It can be re-expressed as a morphism: $a=(P_n(\Z_n(\EC_n)), a)$, where $a_{n-1}^{(1)}: \EC\boxtimes_{\Z_n(\EC_n)} P_n(\Z_n(\EC_n)) = \EC_n \xrightarrow{a} \ED_n$.
We show in Prop.\,\ref{prop:invertible} that isomorphisms are the same as invertible morphisms as expected.

\item There is a {\it unit morphism} $\iota_{\EC_n}: \one_n \to \EC_n$ defined by
$\iota_{\EC_n} = (\EC_n, \id_{\EC_n})$.

\item There is a bulk-to-boundary map $\Z_n(\EC_n) \to \EC_n$. This map can be defined as a morphism $r: P_n(\Z_n(\EC_n)) \to \EC_n$ as follows:
$$
r:=(P_n(\Z_n(\EC_n)) \boxtimes \EC_n , r_{n-1}^{(1)}).
$$
Note that $\Z_n(P_n(\Z_n(\EC_n)))=\Z_n(\EC_n) \boxtimes \overline{\Z_n(\EC_n)}$. The physical configuration associated to $(P_n(\Z_n(\EC_n)) \boxtimes \EC_n) \boxtimes_{\Z_n(P_n(\Z_n(\EC_n)))} P_n(\Z_n(\EC_n))$ is depicted in the following picture:
$$
 \begin{picture}(140, 65)
   \put(20,0){\scalebox{2}{\includegraphics{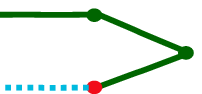}}}
   \put(20,0){
     \setlength{\unitlength}{.75pt}\put(-18,-19){
     \put(165, 50)       { \footnotesize $P_n(\Z_n(\EC))$}
     \put(80, 37)       {\footnotesize $\EC_n$}
     \put(77, 92)     {\footnotesize $ P_n(\Z_n(\EC)) $}
     \put(25, 90)     {\footnotesize $ \Z_n(\EC) $}
     \put(125, 20)     {\footnotesize $ \overline{\Z_n(\EC)} $}
     \put(125, 72)   {\footnotesize $ \Z_n(\EC) $}
     \put(15, 34)     {\footnotesize $\one_{n+1}$ }
     }\setlength{\unitlength}{1pt}}
  \end{picture}
$$

\item A bulk-to-boundary map can be enhanced to an action $\Z_n(\EC_n) \boxtimes \EC_n \to \EC_n$, which can be defined by a morphism $\rho:  P_n(\Z_n(\EC_n)) \boxtimes \EC_n  \to \EC_n$, which is given by the following pair
$$
\rho:=(P_n(\Z_n(\EC_n)) \boxtimes P_n(\Z_n(\EC_n)), \rho_{n-1}^{(1)}),
$$
where $\rho_{n-1}^{(1)}$ is given by
$$(P_n(\Z_n(\EC_n)) \boxtimes P_n(\Z_n(\EC_n))) \boxtimes_{\Z_n(\EC_n) \boxtimes \overline{\Z_n(\EC_n)} \boxtimes \Z_n(\EC_n)}  (P_n(\Z_n(\EC_n)) \boxtimes \EC_n)=\EC_n \xrightarrow{\id_{\EC_n}} \EC_n.
$$
The associated physical configuration is depicted in the following picture:
\be \label{eq:rho}
\raisebox{-35pt}{ \begin{picture}(140, 85)
   \put(0,0){\scalebox{2}{\includegraphics{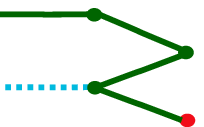}}}
   \put(0,0){
     \setlength{\unitlength}{.75pt}\put(-18,-19){
     \put(170, 75)       {\footnotesize $P_n(\Z_n(\EC))$}
     \put(170, 22)      {\footnotesize $\EC_n$}
     \put(31, 63)       {\footnotesize $P_n(\Z_n(\EC))$}
     \put(125, 47)     {\footnotesize $ \overline{\Z_n(\EC)} $}
     \put(87, 26)     { \footnotesize$ \Z_n(\EC) $}
     \put(15, 118)     {\footnotesize $ \Z_n(\EC) $}
     \put(70, 118)     {\footnotesize $ P_n(\Z_n(\EC)) $}
     \put(128, 98)   {\footnotesize $ \Z_n(\EC) $}
     \put(15, 38)     {\footnotesize $\one_{n+1}$ }
     }\setlength{\unitlength}{1pt}}
  \end{picture}}
\ee

Notice that one can recover the morphism $r$ from $\rho$ by composing $\rho$ with
$$P_n(\Z_n(\EC_n)) = P_n(\Z_n(\EC_n)) \boxtimes \one_n \xrightarrow{\id_{P_n(\Z_n(\EC_n))}\boxtimes \iota_{\EC_n}}  P_n(\Z_n(\EC_n)) \boxtimes \EC_n.$$

\item It is intuitively clearly that there should be a natural morphism $f: \EC_n \to \EB_n \boxtimes_{\EA_{n+1}} \EC_n$, where $\EA_{n+1}$ is not necessarily closed. The definition of $f$ is shown in the following picture.
$$
\begin{picture}(140, 26)
   \put(20,0){\scalebox{2}{\includegraphics{pic-def-morphism-2.eps}}}
   \put(-10,0){
     \setlength{\unitlength}{.75pt}\put(-18,-70){
     \put(60,86)      {\footnotesize $\Z_n(\ED)$}
     \put(150, 86)     {\footnotesize $\Z_n(\EC)$}
     \put(200, 88)     {\footnotesize $\EC_n$}
     \put(125, 90)     {\footnotesize $f_n^{(0)}$}
     }\setlength{\unitlength}{1pt}}
  \end{picture}
$$
where $f_n^{(0)}= \EB_n\boxtimes_{\EA_{n+1}} P_n(\Z_n(\EC_n))$ and $\ED_n:=\EB_n \boxtimes_{\EA_{n+1}} \EC_n$. This morphism can be obtained from a dimensional reduction process depicted in Figure\,\ref{fig:from-pre-to-morphism}.

\enu

}
\end{expl}

The composition of morphisms can also be defined.
\begin{defn} {\rm
Two morphisms $\EC \xrightarrow{f} \ED \xrightarrow{g} \EE$ can be composed to a morphism $g \circ f = \EC \to \EE$ defined by
\be
g \circ f:= \Big( g_n^{(0)}\boxtimes_{\Z_n(\ED)} f_n^{(0)}, \,\,\,  g_n^{(0)} \boxtimes_{\Z_n(\ED)} f_n^{(0)} \boxtimes_{\Z_n(\EC_n)} \EC_n \xrightarrow{g_{n-1}^{(1)} \circ (\id_{g_n^{(0)}} \boxtimes_{\Z_n(\ED)} f_{n-1}^{(1)})} \EE \Big).
\ee
}
\end{defn}

\begin{lemma}
The composition of morphisms is associative and unital, i.e. $h\circ (g \circ f) = (h\circ g) \circ f$ and $\id_{\EC_n} \circ f = f = f\circ \id_{\ED_n}$ for all composable $f,g,h$.
\end{lemma}

In above Lemma, the associativity and the unital properties hold on the nose. It amounts to say that we assume the Minimal Assumption on the categorical side. In this case, we obtain a 1-category $\TO_n^{fun}$ of simple $n$D topological orders with 1-morphisms defined by Def.\,\ref{def:morphism-2}. Notice also that the trivial phase $\one_n$ is an initial object in $\TO_n^{fun}$, and the map $P_{n-1}(-)$ introduced in Sec.\,\ref{sec:dim-red} defines a functor $P_{n-1}: \TO_n^{fun} \to \TO_{n-1}^{fun}$.


\begin{rema}  \label{rema:unique-f0} {\rm
A morphism can be viewed as a physical realization of a particular universal process of mapping excitations in $\EC_n$ to $\ED_n$. We believe that the set of equivalence classes of physical realizations maps surjectively onto the set of universal mapping processes. Logically, there is no problem if the map is not one-to-one.   It is exactly the case if we consider more general physical realizations introduced in Sec.\,\ref{sec:weak-n-morphism}. On the other hand, the notion of a morphism is very special among all physical realizations (see Prop.\,\ref{prop:univ-weak-morphism}). Ideally, we hope that the map is bijective. This amounts to show that the identity $\EE_n\boxtimes_{\Z_n(\EC_n)} \EC_n \simeq \ED_n$ fixes $\EE_n$ up to isomorphisms. In this Remark, we would like to argue physically that this is reasonable. For simplicity, we assume that $\EC_n$ and $\ED_n$ are simple. The idea is to show that the right $n$D boundary $\EC_n$ of the physical configuration $\EE_n\boxtimes_{\Z_n(\EC_n)} \EC_n$ can be replaced by $\Z_n(\EC_n)$ by a physical process. Consider the physical meaning of the expression $\EC \boxtimes_{\EC_n\boxtimes \EC_n^\op} \EC^\op$. Note that the subscript $\EC_n\boxtimes \EC_n^\op$ is a two-layered $n$D systems and depicted as the upper/lower semi-circle in the first figure in Fig.\,\ref{fig:unique-f0}. If the expression $\EC \boxtimes_{\EC\boxtimes \EC^\op} \EC^\op$ indeed has a physical meaning, then $\EC$ must be viewed as two $(n-1)$D phases (recall Remark\,\ref{expl:Pn}) and depicted as the two dark ``points" in the first figure in Fig.\,\ref{fig:unique-f0}.
Then the expression $\EC \boxtimes_{\EC\boxtimes \EC^\op} \EC$ can be viewed as the boundary of the hole in the first figure in Fig.\,\ref{fig:unique-f0}.
\begin{figure}[bt]
$$
\raisebox{-3.5em}{\begin{picture}(100, 85)
   \put(-20,-20){\scalebox{2}{\includegraphics{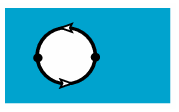}}}
   \put(-20,-20){
     \setlength{\unitlength}{.75pt}\put(-18,-40){
     \put(58, 126)     { \scriptsize $ \EC_n $}
     \put(25, 115)     { \scriptsize $\EC_n$}
     \put(92, 115)     { \scriptsize $\EC_n^\op$}
     \put(100, 87)     { \scriptsize $\Z_n(\EC_n)$}
     }\setlength{\unitlength}{1pt}}
  \end{picture}}  \nonumber
\quad\quad\quad\quad\quad\quad\quad\quad
\raisebox{-3.5em}{\begin{picture}(100, 85)
   \put(10,0){\scalebox{2}{\includegraphics{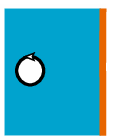}}}
   \put(10,0){
     \setlength{\unitlength}{.75pt}\put(-18,-40){
     \put(45, 125)     { \scriptsize $ \Z_n(\EC_n) $}
     \put(65, 94)     { \scriptsize $\EC_n$}
     \put(-2, 91)     { \scriptsize $\EE_n$}
     }\setlength{\unitlength}{1pt}}
  \end{picture}}
$$
$$
\EC_n \boxtimes_{\EC_n\boxtimes \EC_n^\op} \EC_n^\op
 \quad\quad\quad\quad\quad\quad\quad
\quad\quad\quad  \EE_n \boxtimes_{\Z_n(\EC_n)}  \EC_n \boxtimes_{\EC_n\boxtimes \EC_n^\op} \EC_n^\op 
$$
\caption{Above two figures illustrate the physical interpretations of $\EC \boxtimes_{\EC\boxtimes \EC^\op} \EC^\op$ and $\EE_n \boxtimes_{\Z_n(\EC_n)} \EC_n \boxtimes_{\EC\boxtimes \EC^\op} \EC_n^\op$, respectively. They are explained in Remark\,\ref{rema:unique-f0}.}
\label{fig:unique-f0}
\end{figure}
Moreover, this hole embedded in the $\Z_n(\EC_n)$-phase should be viewed as an excitation of codimension 2 in the $\Z_n(\EC_n)$-phase and can be absorbed by the boundary $\EE_n$ (see the second figure in Fig.\,\ref{fig:unique-f0}). Equivalently, regarding a neighborhood of the hole as the $(n-2)$D topological order $P_{n-2}(\Z_n(\EC_n))=\Z_n(\EC_n)$ (recall Remark\,\ref{expl:Pn}), above arguments suggest the following identities:
\be \label{eq:unique-f0}
\ED_n \boxtimes_{\EC_n\boxtimes \EC_n^\op} \EC_n \simeq \EE_n \boxtimes_{\Z_n(\EC_n)} \EC_n \boxtimes_{\EC_n\boxtimes \EC_n^\op} \EC_n^\op  \simeq \EE_n \boxtimes_{\Z_n(\EC_n)}  \Z_n(\EC_n)  = \EE_n,
\ee
where we have used $\EC_n \simeq \EC_n^\op$ as $\EC_n$-bimodules due to the unitarity. This implies the uniqueness of $\EE_n$. The evidence of above arguments can be found in the case $n=2$. When $n=2$, $\EC_2$ is a unitary fusion 1-category, and $\Z_2(\EC_2)$ is the Drinfeld center of $\EC_2$, and we have $\EC_2^\op \simeq \EC_2^\rev$ as unitary fusion 1-categories. Then the category $\EC_2 \boxtimes_{\EC\boxtimes \EC^\rev} \EC_2$, usually called Hochschild homology, coincides with the Drinfeld center $Z_2(\EC_2)=\EC_2^\op \boxtimes_{\EC\boxtimes \EC_2^\rev} \EC$ of $\EC_2$. So the identity (\ref{eq:unique-f0}) holds precisely in this case. More details will be given elsewhere. We expect that the same result holds for unitary fusion $n$-categories. 
}
\end{rema}


The following result follows from Def.\,\ref{def:morphism-2} immediately.
\begin{thm}  \label{prop:factorize}
Let $f: \EC_n \to \ED_n$ be a morphism. If $\EC_n$ is closed, then $\ED_n = \EC_n \boxtimes \EE_n$ for some $n$D topological order $\EE_n$. If $\ED_n$ is also closed, so is $\EE_n$.
\end{thm}

\begin{rema} \label{rema:matrix-alg} {\rm
For simple 1D topological orders, Thm\,\ref{prop:factorize} reproduces a classical result. More precisely, if there is an algebraic homomorphism $f: M_{m\times m} \to M_{n\times n}$, then we must have $M_{n\times n} \simeq M_{m\times m} \otimes_\Cb M_{k\times k}$, where $n=mk$ and the subscript $\Cb$ of $\otimes$ should be viewed as the usual center of the matrix algebra $M_{m\times m}$.
}
\end{rema}

\begin{rema} {\rm
When $n=3$, a closed 3D topological order can be described by a non-degenerate unitary braided fusion category $\EC_3$. It is reasonable that a morphism between two such topological orders $\EC_3$ and $\ED_3$ is equivalent to a braided monoidal functor $f: \EC_3 \to \ED_3$ preserving the units. It is known that $f$ is fully-faithful and $\ED_3 \simeq \EC_3\boxtimes \EC_3'$, where $\EC_3'$ is the centralizer of $\EC_3$, i.e. the full subcategory of $\ED_3$ consisting of all objects $x$ such that $c_{y,x}\circ c_{x,y} = \id_{x\otimes y}$ for all $y\in \ED_3$ \cite{mueger1,dgno}. 
}\end{rema}

\void{
\begin{rema} {\rm
If a 3D topological orders $\EC_3$ is given by a (not necessarily non-degenerate) braided fusion category, it is known that the center of $\Z_3(\EC_3)$ is given by the centralizer $\EC_3'$ of $\EC$, which is defined as the full subcategory of $\EC$ consisting of object $x$ such that
the composed isomorphism $x\otimes y \xrightarrow{c_{x,y}} y\otimes x \xrightarrow{c_{y,x}} x\otimes y$ is the identity map $\id_{x\otimes y}$. It will be interesting to check if Def.\,\ref{def:morphism-2} is compatible with the usual notion of a braided monoidal functor.
}
\end{rema}
}

\subsection{Higher morphisms and the $(n+1,1)$-category $\TO_n^{fun}$} 
\label{sec:higher-morphisms}
One can also introduce the notion of a 2-isomorphism between two morphisms. 
We need it only in 
Sec.\,\ref{sec:universal-higher-morphism}. 

\begin{defn} \label{def:phi}  {\rm
Let $f,g: \EC_n \to \ED_n$ be two morphisms between two simple $n$D topological orders $\EC_n$ and $\ED_n$. A 2-isomorphism $\phi: f \Rightarrow g$ is a pair $(\phi^{(0)}, \phi^{(1)})$ such that
\bnu
\item $\phi_{n-1}^{(0)}: f_n^{(0)} \to g_n^{(0)}$ is an $(n-1)$D invertible gapped domain wall between $f_n^{(0)}$ and $g_n^{(0)}$, both of which are domain walls between $\Z_n(\EC_n)$ and $\Z_n(\ED_n)$.

\item $\phi_{n-2}^{(1)}: g_{n-1}^{(1)} \circ ( \phi_{n-1}^{(0)} \boxtimes_{\Z_n(\EC_n)} \id_{\EC_n}) \to f_{n-1}^{(1)}$ is an $(n-2)$D invertible gapped domain wall between two associated $(n-1)$D domain walls.
\be  \label{diag:phi-(1)}
\raisebox{1em}{\xymatrix@R=0.2em{
& g_n^{(0)}\boxtimes_{\Z_n(\EC_n)}  \EC_n \ar[ddr]^{g_{n-1}^{(1)}} &  \\
& \Downarrow \phi_{n-2}^{(1)} &  \\
f_n^{(0)} \boxtimes_{\Z_n(\EC_n)} \EC_n \ar[uur]^{\hspace{-2cm} \phi_{n-1}^{(0)}\boxtimes_{\Z_n(\EC_n)} \id_{\EC_n}}  \ar[rr]_{f_{n-1}^{(1)}}  & &   \ED_n
}}
\ee
\enu
In this case, we denote $f \simeq g$.
}
\end{defn}


\begin{expl}
Every morphism $f:\one_n\to\EC_n$ is canonically 2-isomorphic to the unit morphism. Actually, the pair $((f^{(1)}_{n-1})^{-1},\id_{\id_{\EC_n}})$ defines a 2-isomorphism $\iota_f: \iota_{\EC_n} \Rightarrow f$.
\end{expl}


\begin{rema} {\rm
In general, two morphisms $f$ and $g$ are not isomorphic, even if $f_n^{(0)} \simeq g_n^{(0)}$. However, if $f_n^{(0)} \simeq g_n^{(0)}$, then there always exists $f_{n-1}^{(1)}$ and $g_{n-1}^{(1)}$ such that the diagram (\ref{diag:phi-(1)}) is commutative, i.e. $f\simeq g$. For example, one can simply define $f_{n-1}^{(1)}:= g_{n-1}^{(1)} \circ (\phi_{n-1}^{(0)} \boxtimes_{\Z_n(\EC_n)} \id_{\EC_n})$.
}
\end{rema}

\begin{prop}  \label{prop:invertible}
A morphism $f:\EC_n\to\ED_n$ is identical to an isomorphism if and only if there is a morphism $g:\ED_n\to\EC_n$ such that $g\circ f\simeq \id_{\EC_n}$ and $f\circ g\simeq \id_{\ED_n}$.
\end{prop}
\pf
Only the sufficiency needs to be proved. Suppose $g\circ f \simeq \id_{\EC_n}$ and $f\circ g \simeq \id_{\ED_n}$. Then $f^{(0)}$ is an invertible domain wall between $\Z_n(\ED_n)$ and $\Z_n(\EC_n)$. It is just a matter of convention how we identify $\Z_n(\ED_n)$ with $\Z_n(\EC_n)$. Using $f^{(0)}$, we can identify $\Z_n(\ED_n)$ with $\Z_n(\EC_n)$, then $f^{(1)}$ can be viewed as a domain wall between $\ED_n$ and $\EC_n$ (recall Def.\,\ref{def:morita}). Moreover, it is invertible. So $f$ is nothing but the isomorphism $f^{(1)}$.
\epf


By induction on dimensions, we define a $k$-isomorphism $\phi\xrightarrow{(k)}\psi$ between two $(k-1)$-isomorphisms $\phi,\psi:f\xrightarrow{(k-1)}g$ as a pair $(\Phi_{n-k+1}^{(0)}, \Phi_{n-k}^{(1)})$ of $(n-k+1)$D and $(n-k)$D invertible gapped domain walls between associated domain walls. In this way, we obtain an $(n+1,1)$-category of $n$D topological orders, denoted still by $\TO_n^{fun}$. An $(m,1)$-category is an $m$-category with only invertible $k$-morphisms for $k>1$. 


\section{The universal property of the \bulk}

In this section, we prove that the \bulk satisfies the same universal property as that of the center (of an algebra) in mathematics. Namely, \bulk = center. It is the main result of this paper, and is completely independent of the precise mathematical formulation of a topological order. We also explain why the universal property leads to the usual notion of the center, at last, we show that the morphism defined in Def.\,\ref{def:morphism-2} coincides with the usual notions of morphisms in mathematics by assuming \bulk = center.

\subsection{The universal property of the \bulk} \label{sec:universal}


\medskip
The action $\rho: P_n(\Z_n(\EC_n)) \boxtimes \EC_n \to \EC_n$ is unital, i.e. $\rho \circ (\iota_{\Z_n(\EC_n)} \boxtimes \id_{\EC_n}) = \id_{\EC_n}$, which is equivalent to the commutativity of the following diagram:
\be  \label{diag:univ-property-1}
\raisebox{20pt}{ \xymatrix{
& P_n(\Z_n(\EC_n)) \boxtimes \EC_n \ar[rd]^\rho &  \\
\EC_n \ar[rr]^{\id_{\EC_n}} \ar[ur]^{\iota_{P_n(\Z_n(\EC_n))} \boxtimes \id_{\EC_n}} & & \EC_n\, .
}}
\ee
We would like to show that the pair $(P_n(\Z_n(\EC_n)), \rho)$ satisfies the universal property of center that determines the pair up to unique isomorphism. As a consequence, $\Z_n(\EC_n)=(P_n(\Z_n(\EC_n)), \rho)$.

\begin{thm} \label{thm:universal}
The pair $(P_n(\Z_n(\EC_n)), \rho)$ satisfies the following universal property of center. Let $(\EX_n, f)$ be another such a pair. In other words, $\EX_n$ is an $n$D topological order and $f: \EX_n \boxtimes \EC_n \to \EC_n$ a morphism such that the following diagram
\be  \label{diag:X-n-0}
\xymatrix{
& \EX_n \boxtimes \EC_n \ar[rd]^f &  \\
\EC_n \ar[rr]^{\id_{\EC_n}} \ar[ur]^{\iota_{\EX_n} \boxtimes \id_{\EC_n}} & & \EC_n
}
\ee
is commutative. Then there is a unique morphism $\underline{f}: \EX_n \to P_n(\Z_n(\EC_n))$ such that the following two diagrams
\be \label{diag:univ-prop-0}
\xymatrix{
& \EX_n \ar[rd]^{\underline{f}} &  \\
\one_n \ar[rr]^-{\iota_{P_n(\Z_n(\EC_n))}} \ar[ur]^{\iota_{\EX_n}} & & P_n(\Z_n(\EC_n))
}
\quad\quad
\xymatrix{
& P_n(\Z_n(\EC_n)) \boxtimes \EC_n \ar[rd]^\rho &  \\
\EX_n \boxtimes \EC_n \ar[rr]^f \ar[ur]^{\exists !\, \underline{f} \boxtimes \id_{\EC_n}} & & \EC_n
}
\ee
are commutative. The notation ``$\exists !$" means ``exists a unique". 
\end{thm}
\pf
The physical configuration associated to the composed morphism $f\circ (\iota_{\EX_n} \boxtimes \id_{\EC_n})$ is depicted in the following picture.
\be \label{eq:univ-proof-2}
\raisebox{-20pt}{
 \begin{picture}(140, 65)
   \put(0,15){\scalebox{2}{\includegraphics{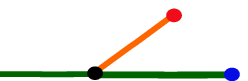}}}
   \put(0,15){
     \setlength{\unitlength}{.75pt}\put(-18,-19){
     \put(205, 32)       {\footnotesize $\EC_n$}
     \put(85, 60)       {\footnotesize $\Z_n(\EX_n)$}
     \put(80, 6)      {\footnotesize $f_n^{(0)}$ }
     \put(0, 10)     {\footnotesize $ \Z_n(\EC_n) $}
     \put(160, 78)     {\footnotesize $ \EX_n $}

     \put(135, 35)     {\footnotesize $\Z_n(\EC_n)$ }
     }\setlength{\unitlength}{1pt}}
  \end{picture}}
\ee
So, the commutativity of \eqref{diag:X-n-0} amounts to the identity $\EX_n \boxtimes_{\Z_n(\EX_n)} f^{(0)} = \id_{\EC_n}^{(0)}=P_n(\Z_n(\EC_n))$.
Notice that the pair $(f^{(0)}, \id_{P_n(\Z_n(\EC_n))})$ defines a morphism $\underline{f}: \EX_n \to P_n(\Z_n(\EC_n))$. Clearly, we have $\iota_{P_n(\Z_n(\EC_n))} = \underline{f} \circ \iota_{\EX_n}$.

Now we draw the physical configuration associated to the composed morphism $\rho \circ (\underline{f} \boxtimes \id_{\EC_n})$ as follows:
\be \label{eq:univ-proof-2}
\raisebox{-45pt}{ \begin{picture}(140, 95)
   \put(0,15){\scalebox{2}{\includegraphics{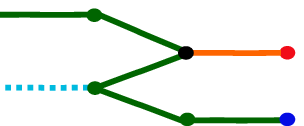}}}
   \put(0,15){
     \setlength{\unitlength}{.75pt}\put(-18,-19){
     \put(250, 78)       {\footnotesize $\EX_n$}
     \put(250, 25)      {\footnotesize $\EC_n$}
     \put(160,60)  {\footnotesize $f_n^{(0)}$ }
     \put(91, 25)       {\footnotesize $\Z_n(\EC_n)$}
     \put(80, 68)     {\footnotesize $ \overline{\Z_n(\EC_n)} $}
     \put(180, 35)     {\footnotesize $ \Z_n(\EC_n) $}
     \put(115, 103)     {\footnotesize $ \Z_n(\EC_n) $}
     \put(170, 85)     {\footnotesize $ \Z_n(\EX_n) $}
     \put(10, 115)   {\footnotesize $ \Z_n(\EC_n) $}
     \put(10, 60)     {\footnotesize $\one_{n+1}$ }
     }\setlength{\unitlength}{1pt}}
  \end{picture}}
\ee
where three green ``dots" are all labeled by $P_n(\Z_n(\EC_n))$. By deforming the pictures topologically, we see that $\rho \circ (\underline{f} \boxtimes \id_{\EC_n}) = f$. This proves the existence of $\underline{f}$.

\smallskip
It remains to prove the uniqueness of such $\underline{f}$. Assume a morphism $g: \EX_n \to P_n(\Z_n(\EC_n))$ satisfies $\iota_{P_n(\Z_n(\EC_n))} = g \circ \iota_{\EX_n}$ and $\rho \circ (g \boxtimes \id_{\EC_n}) = f$. 
From the former identity, we obtain the identities $\EX_n \boxtimes_{\Z_n(\EX_n)} g^{(0)} = P_n(\Z_n(\EC_n))$ and $g^{(1)} = \id_{P_n(\Z_n(\EC_n))}$; from the latter, we obtain $g^{(0)} = f^{(0)}$. That is, $g=\underline{f}$.
\epf

\begin{rema} {\rm
In Thm\,\ref{thm:universal}, we have used only strictly commutative diagrams by the Minimal Assumption (see Remark\,\ref{rema:minimal}). If we want to relax this condition by considering higher isomorphisms, the universal property of the center is much more complicated \cite{lurie2}. 
We briefly discuss it in Sec.\,\ref{sec:universal-higher-morphism}.
}
\end{rema}



\subsection{The universal property and the mathematical notion of center} \label{sec:center}

The universal property stated in Thm.\,\ref{thm:universal} is exactly the universal property of the notion of center in mathematics. This universal property defines the center up to isomorphisms. In this subsection, we try to  explain this well-known fact to physical-minded readers. 

\medskip 
We first explain the notion of center for an ordinary algebra over $\Cb$. 
Let $A$ be an algebra over $\Cb$. The center of $A$ is usually defined by
\be \label{eq:Z-A}
Z(A) := \{ z\in A | za = az, \forall a\in A \},
\ee
which is a (commutative) subalgebra of $A$. 
But this notion can be redefined by its universal property. Consider another algebra $B$.  The tensor product $B\otimes A$ has a natural algebra structure. If there is an algebra homomorphism $f: B\otimes A \to A$, then $B$ must satisfy some special property. For example, the obvious action $m: A \otimes A \to A$, defined by $m(a\otimes b) = ab$, is not an algebra map. To see this, take $1\otimes a, b\otimes 1 \in A\otimes A$. On the one hand, we have
\begin{equation} \label{eq:ba}
m((1\otimes a) \cdot (b\otimes 1))  = m(b\otimes a) = ba.
\end{equation}
On the other hand,
\begin{equation} \label{eq:ab}
m(1\otimes a) \cdot m(b\otimes 1) = ab.
\end{equation}
Therefore, the multiplication map $m: A\otimes A\to A$ is not an algebra map if $A$ is not commutative.
However, $Z(A) \otimes A \xrightarrow{m} A$ is an algebra homomorphism. Moreover, the following diagram:
\be  \label{diag:ZAA}
\xymatrix{
& Z(A) \otimes A \ar[rd]^m  & \\
A \ar[rr]^{\id_A} \ar[ru]^{\iota_{Z(A)} \otimes \id_A}  & & A
}
\ee
where $\iota_{Z(A)}: \Cb \to Z(A)$ is the unit map, is commutative.
The pair $(Z(A), Z(A)\otimes A \xrightarrow{m} A)$ satisfies the following so-called {\it universal property}: if $B$ is an algebra and $f: B\otimes A \to A$ is a unital action and an algebra map, there is a unique algebra map $\underline{f}: B \to Z(A)$ such that the following diagram:
\be \label{diag:alg-center}
\xymatrix{ & Z(A) \otimes A \ar[rd]^m & \\
B\otimes A \ar[rr]^f \ar[ur]^{\exists !\, \, \underline{f} \otimes 1}  & & A
}
\ee
is commutative. 
This is true because the restriction $\underline{f}:=f(-,1): B \to A$ is an algebra map, $f(1\otimes a) =a$, and we have
\begin{align}
f(b \otimes a) &= f((b\otimes 1) \cdot (1\otimes a)) = \underline{f}(b) \cdot f(1\otimes a) = \underline{f}(b) a \nn
&= f((1\otimes a) \cdot (b\otimes 1)) = a\underline{f}(b),
\end{align}
i.e. $\underline{f}(b)\in Z(A)$. Therefore, $\underline{f}: B \to Z(A)$ is an algebra map and the diagram (\ref{diag:alg-center}) is commutative. The uniqueness of $\underline{f}$ is obvious.

Conversely, the universal property determines the center up to canonical isomorphisms. 
Namely, any algebra $Z$, satisfying the universal property, is canonically isomorphic to the algebra $Z(A)$ defined by Eq.\,(\ref{eq:Z-A}). 

\smallskip
The universal property also implies that there is a canonical isomorphism
$$
\hom(B\otimes_\Cb A, A) \xrightarrow{\simeq} \hom(B, Z(A)), 
$$
where both ``homs'' are sets of algebra homomorphisms. It provides a useful characterization of the center $Z(A)$ as an internal hom. Another useful characterization of $Z(A)$ is the set of $A$-$A$-bimodule maps from $A$ to $A$, i.e. $Z(A) = \hom_{A|A}(A, A)$.

\medskip
The notion of center can be defined in very general context. Jacob Lurie introduced the notion of center for an $E_k$-algebra object in a symmetric monoidal $\infty$-category \cite{lurie2}. What we need in the study of topological orders are some special cases of Lurie's general notion. We discuss some examples below. 

\begin{expl} {\rm
Recall Example\,\ref{expl:n-category}, a $0$D topological order $\EC_0$ can be defined as a finite dimensional Hilbert space $U$ with a distinguished element $u$, $\EC_0=(U,u)$. $\EC_0$ can be viewed as an $E_0$-algebra \cite{lurie2}. 
In this case, the $\boxtimes$ is just usual tensor product of Hilbert spaces. We have
$$
\hom(\EX_0 \otimes_\Cb \EC_0, \EC_0) \simeq \hom(\EX_0, \hom_\Cb(\EC_0, \EC_0)),
$$
where both sides are sets of $E_0$-algebra homomorphisms, i.e. linear maps that preserves the distinguished elements. Therefore, we have the center $Z_0(\EC_0)=\hom_\Cb(\EC_0, \EC_0)$, which a simple $1$D topological order. This center is an $E_0$-center. 
}
\end{expl}

\begin{expl} {\rm
A $1$D topological order $\EC_1$ is a semisimple algebra over $\Cb$. Its bulk $\Z_1(\EC_1)$ is the $Z_1$-center $Z_1(\EC_1)=\hom_{\EC_1|\EC_1}(\EC_1,\EC_1)$, i.e. the usual center of an algebra over $\Cb$. 
}
\end{expl}

\begin{expl} {\rm
If we use the 2-codimensional definition of a $1$D topological order, it is given by an $E_0$-algebra $\EC_1=(\EM, m)$, where $\EM$ is a unitary 1-category and $m$ is a distinguished object in $\EM$. In this case, all arrows in diagrams (\ref{diag:X-n-0}) and (\ref{diag:univ-prop-0}) are unitary 1-functors that preserve the distinguished object, and $\boxtimes$ is the Tambara's tensor product over the unitary fusion 1-category $\hilb$. Then we have
\be \label{eq:internal-hom}
\fun(\EX_1 \boxtimes \EC_1, \EC_1) \simeq \fun(\EX_1, \fun(\EC_1, \EC_1)),
\ee
where both sides are categories of unitary 1-functors preserving the distinguished objects. Therefore, we must have $Z_1(\EC_1)=\fun(\EC_1, \EC_1)$ (an $E_0$-center). As a special case, in Fig.\,\ref{fig:lw-mod} (b), we have $\EE_2=Z_1(\fun_\EC(\EM, \EN))$. So the boundary-bulk relation also holds in Fig.\,\ref{fig:lw-mod} (b). 
}\end{expl}

\begin{expl} {\rm
A 2D topological order is given by a unitary multi-fusion 1-category $\EC_2$. Let $Z_2(\EC_2)$ be the Drinfeld center of $\EC_2$ and $m: Z_2(\EC_2) \boxtimes \EC_2 \to \EC_2$ the tensor product functor. Then the pair $(Z_2(\EC_2),m)$ satisfies the universal property in Thm.\,\ref{thm:universal}, i.e. 
\be \label{diag:monoidal-center}
\xymatrix{ & Z_2(\EC_2) \boxtimes \EC_2 \ar[rd]^{m} & \\
\ED_2 \boxtimes \EC_2 \ar[rr]^f \ar[ur]^{\exists !\, \, \underline{f} \boxtimes \id_{\EC_2}}  & & \EC_2}
\ee
where $\ED_2$ is a monoidal 1-category and $f$ is a unitary monoidal 1-functor. Note that the functor $\underline{f}=f(-\boxtimes 1_{\EC_2}): \ED_2 \to \EC_2$ defines a unitary monoidal 1-functor. $f$ being monoidal requires that each object $g(d)\in \EC_2$ acquires a half braiding, i.e. an isomorphism $x \otimes \underline{f}(d) \xrightarrow{c_{x,\underline{f}(d)}} \underline{f}(d) \otimes x$ for all $x\in \EC_2$, satisfying some natural properties. As a consequence, $\underline{f}$ defines a functor $\underline{f}: \ED_2 \to Z_2(\EC_2)$ such that Diagram\,(\ref{diag:monoidal-center}) is commutative. The uniqueness of $\underline{f}$ is obvious. The Drinfeld center $Z_2(\EC_2)$ of $\EC_2$ can be equivalently defined by $Z_2(\EC_2) := \fun_{\EC|\EC}(\EC, \EC)$. In this case, the functor $m$ is given the evaluation functor $(F, x)\mapsto F(x)$. 
}
\end{expl}

\begin{expl} \label{expl:E2-center} {\rm
Let $\EC_3$ be a (unitary) braided fusion 1-category, which is an $E_2$-algebra. We replace all categories in (\ref{diag:monoidal-center}) by (unitary) braided monoidal 1-categories and all arrows by (unitary) braided monoidal 1-functors, then it is easy to show that the image of the functor $\underline{f}=f(-\boxtimes 1_\EC): \ED \to \EC$ consists of object $x\in \EC$ that are symmetric with all objects in $\EC$, i.e. $c_{y,x} \circ c_{x,y} =\id_{x\otimes y}$ for all $y\in \EC$. This universal property defines the centralizer (or the $E_2$-center) of $\EC$, denoted by $\EC'$. If $\EC$ is non-degenerate, we should have $Z_3(\EC_3)=\EC'=\one_4$ by Conjecture\,\ref{conj:closed-f=bf}. When $\EC$ is not non-degenerate, $Z_3(\EC_3)$ does not coincide with $\EC'$ (see Remark\,\ref{rema:lurie}). The centralizer (or $E_2$-center) of a (unitary) braided fusion $n$-category can be defined by the same universal property. 
}
\end{expl}


\begin{rema} {\rm
All above examples are special cases of a more general principle: an $E_n$-center of an $E_n$-algebra is automatically an $E_{n+1}$-algebra (see \cite{lurie2}). But an $E_n$-algebra only describes particle-like excitations. In the study of topological orders, by adding excitations of all codimensions, only $E_k$-algebras for $k\leq 2$ matter. 

}
\end{rema}

Given an indecomposable unitary multi-fusion $n$-category for $n\geq 0$, the $Z_{n+1}$-center $Z_{n+1}(\EC_{n+1})$ is defined by
\be \label{eq:center-n-tensor-cat}
Z_{n+1}(\EC_{n+1})= \fun_{\EC|\EC}(\EC, \EC),
\ee
where $\fun_{\EC|\EC}(\EC, \EC)$ is the $n$-category of unitary $\EC$-$\EC$-bimodule functors and is believed to be a unitary braided fusion $n$-category (recall Conjecture\,\ref{conj:f-Z-bf}). There is a natural evaluation functor $\ev: Z_{n+1}(\EC_{n+1}) \boxtimes \EC_{n+1} \to \EC_{n+1}$ defined by $F\boxtimes x \to F(x)$, which is unitary and monoidal. Moreover, the pair $(Z_{n+1}(\EC_{n+1}), \ev)$ satisfies the universal property of center:
\be \label{eq:center-n-tensor-cat-univ-prop}
\xymatrix{
& Z_{n+1}(\EC_{n+1}) \boxtimes \EC_{n+1} \ar[rd]^{\ev} & \\
\EX_{n+1} \boxtimes \EC_{n+1} \ar[rr]^f  \ar[ru]^{\exists ! \underline{f} \boxtimes \id_\EC}
& & \EC_{n+1},
}
\ee
in which the diagram is commutative (up to higher isomorphisms), for any unitary fusion $n$-category $\EX_{n+1}$ and a unitary monoidal $n$-functor $f$ (or a unitary $(n+$$1)$-functor). This $Z_n$-center is an $E_1$-center. See \cite{baez-neuchl} for the center of a monoidal 2-category and \cite{lurie2} for general situations.


\begin{rema} \label{rema:lurie} {\rm
This remark is due to Jacob Lurie. There is a notion of center for an $E_k$-algebra object in a symmetric monoidal $\infty$-category (\cite[Sec.\,5.3.1]{lurie2}). When we take the $\infty$-category to be the $\infty$-category of $(\infty,n)$-categories with some additional structures, we get a notion of center $Z(-)$ for an $E_k$-monoidal $(\infty,n)$-category. 
This center is automatically an $E_{k+1}$-monoidal $(\infty,n)$-category and given by Eq.\,(\ref{eq:center-n-tensor-cat}) for $k=1$, and by the same formula for $k>1$ but with the bimodule replaced by $E_k$-module. If $\EC$ is an $E_k$-monoidal $(\infty,n)$-category, then $\hom_\EC(1_\EC,1_\EC)$, where $1_\EC$ is the tensor unit of $\EC$, is an $E_{k+1}$-monoidal $(\infty, n-1)$-category. 
Using the results proven in Sec.\,4.8 in \cite{lurie2}, one can show that
$$
Z(\hom_\EC(1_\EC,1_\EC)) \simeq \hom_{Z(\EC)}(1_{Z(\EC)}, 1_{Z(\EC)})
$$ 
as $E_{k+2}$-monoidal $(\infty, n-1)$-categories. In general, $Z(\hom_\EC(1_\EC,1_\EC))$ is not enough to recover the $E_{k+1}$-monoidal $(\infty, n)$-category $Z(\EC)$. 
}
\end{rema}

Thm\,\ref{thm:universal} simply says that the \bulk is indeed the center in the mathematical sense. Moreover, if we assume that a morphism between two $n$D topological orders is equivalent to a unitary $n$-functor, then we must have $\Z_n(-)=Z_n(-)$. 
Then Eq.\,(\ref{eq:Z2=0}) leads to the following mathematical conjecture. 

\begin{conj} \label{conj:Z2=0}
For an indecomposable unitary multi-fusion $n$-category $\EC_{n+1}$ ($n\geq 0$), 
we must have $Z_{n+2}(Z_{n+1}(\EC_{n+1}))=\one_{n+3}$.
\end{conj}

This conjecture is true for $n=2$ \cite{mueger2,eno2009}.

\subsection{Physical morphisms coincide with mathematical ones}
\label{sec:morphism=morphism}

In this subsection, using the categorical definition of a $n$D topological order $\EC_n$ and assuming $\Z_n(-)=Z_n(-)$, we show in a few low dimensional cases that Def.\,\ref{def:morphism-2} coincide with the notion of a unitary $n$-functor.  

\medskip
In 0D, let $(U, u)$ and $(V,v)$ be two 0D topological orders. Namely, $U$ and $V$ are finite dimensional Hilbert spaces and $u\in U, v\in V$. $\boxtimes$ is just the usual Hilbert space tensor product $\otimes_\Cb$ and the center of $U$ is the matrix algebra $\text{End}(U)$. A linear map $f: U \to V$ from $U$ to $V$ such that $f(u)=v$ can be rewritten as a pair $(f^{(0)}, f^{(1)})$, where $f^{(0)}=(\hom_\Cb(U,V), f)$ is an 0D topological order and $f^{(1)}$ is the canonical isomorphism
\be \label{eq:0D-wall}
\hom_\Cb(U,V)\otimes_{\text{End}(U)} U \xrightarrow{\simeq} V
\ee
defined by the evaluation map $g\otimes_{\text{End}(U)} u_1 \mapsto g(u_1)$. Conversely, a domain wall between $\text{End}(V)$ and $\text{End}(U)$ is an
$\text{End}(V)$-$\text{End}(U)$-bimodule, which has to be a direct sum of $\hom_\Cb(U,V)$. Therefore, $\hom_\Cb(U,V)$ is the unique domain wall such that the isomorphism (\ref{eq:0D-wall}) is possible. Therefore, our physical notion of a morphism is equivalent to that of a linear map preserving the distinguished elements.

\medskip
In 1D, for simple $1$D topological orders, an algebra homomorphism $f: A \to B$ between two matrix algebras $A$ and $B$ is only possible if $B=C \otimes_\Cb A$ for another matrix algebra $C$ (see Remark\,\ref{rema:matrix-alg}), where $\Cb$ should be viewed as the center of $A$. It coincides with our physical notion of a morphism between two $1$D topological orders $A$ and $B$. 

\medskip
Consider the 2-codimensional definition of a $1$D topological order, i.e. an $E_0$-algebra $(\EA_1,a)$, where $\EA_1$ is a unitary 1-category and $a\in \EA_1$. The $Z_1$-center of $\EA_1$ is just $\fun(\EA,\EA)$. A unitary 1-functor $f: (\EA,a) \to (\EB,b)$ such that $f(a)=b$ can be rewritten as a pair $(f^{(0)}, f^{(1)})$, where $f^{(0)}=(\fun(\EA, \EB), f)$ is a $1$D topological order and $f^{(1)}$ is the canonical unitary equivalence:
$$
\fun(\EA, \EB) \boxtimes_{\fun(\EA, \EA)} \EA \xrightarrow{\simeq} \EB
$$
defined again by the evaluation map $g\boxtimes_{\fun(\EA, \EA)} c \to g(c)$ for $c\in \EA$. Therefore, our physical notion of a morphism is equivalent to that of a unitary 1-functor preserving the distinguished objects.

\medskip
In 2D, a 2-stable simple 2D topological order can be described by a unitary fusion 1-category $\EC_2$. In this case, it was proved via anyon condensation that the \bulk of $\EC_2$ is indeed given by the Drinfeld center $Z_2(\EC_2)$ of $\EC_2$, i.e. $\Z_2(\EC_2)=Z_2(\EC_2)$ \cite{kong-anyon}.
Using this fact, we obtain a notion of a morphism between multi-fusion 1-categories from Def.\,\ref{def:morphism-2}. In \cite[Thm.\,5.13]{kong-zheng}, we will prove that this notion is equivalent to usual notion of a monoidal 1-functor. More precisely, a monoidal 1-functor $f: \EC \to \ED$ between two multi-fusion 1-categories $\EC$ and $\ED$ can be rewritten as a pair $(\fun_{\EC|\EC}(\EC, {}_f\ED_f), f^{(1)})$, where ${}_f\ED_f$ is the $\EC$-$\EC$-bimodule $\ED$ induced by $f$ and $f^{(1)}$ is the canonical monoidal equivalence $\fun_{\EC|\EC}(\EC, {}_f\ED_f) \boxtimes_{Z_2(\EC)} \EC \xrightarrow{\simeq} \ED$ \cite{kong-zheng}.

\begin{rema} {\rm
In general, we expect that the $(n+1,1)$-category $\TO_n^{fun}$ of pairs $(\ED_n, \iota)$ with 1-morphisms given by unitary $n$-functors, 2-isomorphisms by natural isomorphisms and 3-isomorphisms by invertible modifications, so on and so forth. 
}
\end{rema}

\section{The boundary-bulk relation and the functoriality of $\Z_n$} 
\label{sec:Z-functor}

The unique-bulk hypothesis and \bulk = center are only parts of the boundary-bulk relation. For a rather complete boundary-bulk relation, we would like to propose the {\it strong unique-bulk hypothesis}, which extends the uniqueness of the \bulk of an $n+$1D topological order to that of the ``\bulk" of a $k$-codimensional wall for $k=1,\cdots,n$ (see Fig.\,\ref{fig:unique-bulk-hypothesis}). This hypothesis can be true if we generalize the microscopic definition of a topological order (Def.\,\ref{def:TO}) to a (potentially anomalous) wall on the bottom boundary in Fig.\,\ref{fig:unique-bulk-hypothesis} (such as $a_{n-1}, b_{n-1}, c_{n-1}$). We assume this strong unique-bulk hypothesis in this section.

\begin{figure}[tb]
 \begin{picture}(150, 90)
   \put(100,10){\scalebox{2}{\includegraphics{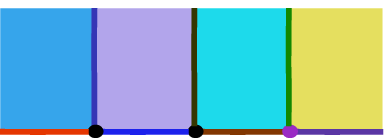}}}
   \put(60,-55){
     \setlength{\unitlength}{.75pt}\put(-18,-19){
     \put(95, 92)       { $\EA_n$}
     \put(170, 92)       { $\EB_n$}
     \put(245, 92)       { $\EC_n$}
     \put(320, 92)       { $\ED_n$}
     \put(130,100)      { $a_{n-1}$}
     \put(208,98)      { $b_{n-1}$}
     \put(280,100)      { $c_{n-1}$}
     \put(85, 150)    {\scriptsize $\Z_n(\EA)$}
     \put(165, 150)    {\scriptsize $\Z_n(\EB)$}
     \put(240, 150)    {\scriptsize $\Z_n(\EC)$}
     \put(310, 150)    {\scriptsize $\Z_n(\ED)$}
     \put(130,215)     { $\EX_n$}
     \put(210,215)     { $\EY_n$}
     \put(285,215)     { $\EZ_n$}
     }\setlength{\unitlength}{1pt}}
  \end{picture}
\caption{The {\it strong unique-bulk hypothesis} says that, in above physical configuration,  not only the $(n+$$1)$D bulks $\Z_n(\EA)$, $\Z_n(\EB)$, $\Z_n(\EC)$ and $\Z_n(\ED)$ are uniquely determined by $\EA,\EB,\EC,\ED$, respectively, but also the $n$D domain walls $\EX_n$ is uniquely determined by $(\EA,a,\EB)$, and $\EY_n$ by $(\EB,b,\EC)$, and $\EZ_n$ by $(\EC,c,\ED)$. We also denote $\EX_n$ by $\Z_n^{(1)}(a_{n-1})$. 
}
\label{fig:unique-bulk-hypothesis}
\end{figure}

\subsection{Closed and anomalous domain walls} \label{sec:anomalous-wall}

In Fig.\,\ref{fig:unique-bulk-hypothesis}, note that $\EX_n$ is not the \bulk of $a_{n-1}$ but uniquely determined by $a_{n-1}$. We define $\Z_{n-1}^{(1)}(a_{n-1}):=\EX_n$.  Note that $a_{n-1}$ can be viewed as a closed wall between $\EA_n$ and $\EX_n\boxtimes_{\Z_n(\EB_n)} \EB_n$ or a closed wall between $\EA_n^\op\boxtimes_{\Z_n(\EA_n)} \EX_n$ and $\EB_n$. It is not a closed wall between $\EA_n$ and $\EB_n$ unless $\EX_n$ is trivial. 

\begin{defn} \label{def:wall} {\rm 
A domain wall $a_{n-1}$ between the $\EA_n$-phase and the $\EB_n$-phase is called 
{\it Morita-closed} if $\EX_n$ is invertible; it is called {\it closed} if $\Z_n(\EA_n)=\Z_n(\EB_n)$\footnote{When we say $\Z_n(\EA_n)=\Z_n(\EB_n)$, we mean that we have made a choice of how we identify them.} and $\EX_n=\id_{\Z_n(\EA_n)}$; it is called {\it anomalous} if $\EX_n$ is not invertible. 
}
\end{defn}


\begin{defn} \label{def:morita} {\rm
If two simple topological orders $\EA_n$ and $\EB_n$ are connected by a Morita-closed gapped domain wall $a_{n-1}$, then we say that they are {\it Morita equivalent}, denoted by $\EA_n \sim \EB_n$, and the Morita equivalence is given by $a_{n-1}$. When $\EA_n$ and $\EB_n$ are both closed, the Morita equivalence is also called the {\it Witt equivalence} \cite{dmno,fsv,kong-anyon,kong-wen}. 
}
\end{defn}

\medskip
By the definition of Morita equivalence, we have 
\be \label{eq:morita}
\EA_n \sim \EB_n \quad\quad \Rightarrow \quad\quad \Z_n(\EA_n) \simeq \Z_n(\EB_n). 
\ee
This physical result is also natural mathematically. In mathematics, Morita equivalent algebras in a certain nice monoidal category always share the same center. Actually, it is also interesting to consider another direction. In general, $Z(A)\simeq Z(B)$ does not imply $A\sim B$. But if we assume certain duality structures on $A$ and $B$, then it is possible. We give a couple of examples below. 
\bnu
\item In 1D, if $A$ and $B$ are two finite dimensional $C^\ast$-algebras, i.e. the direct sums of matrix algebras, describing composite (unstable) topological orders, then $A$ is Morita equivalent to $B$ if and only if $Z(A)\simeq Z(B)$ as algebras.  This result is quite trivial. But it can be generalized to a non-trivial one in the framework of 2D rational conformal field theory, in which $A$ and $B$ are two special symmetric Frobenius algebras in a modular tensor categories, then $A\sim B$ if and only if $Z(A)\simeq Z(B)$ as algebras, where $Z(A)$ is the so-called {\it full center} of $A$ \cite{kr-morita}.

\item In 2D, if two unitary fusion 1-categories $\EA_2$ and $\EB_2$ are Morita equivalent if and only if their Drinfeld centers are equivalent as braided monoidal 1-categories \cite{mueger0, kitaev, eno2008}.  
\enu
It seems natural to ask if the Morita equivalence is equivalent to the equivalence of $Z_n$-centers for a unitary fusion $n$-category. It turns out that it is not true. In 3D, all closed 3D topological orders share the same trivial \bulk, but they are not Morita equivalent (or Witt equivalent) in general. The discrepancy is measured by the Witt equivalence classes, which form an infinite group called Witt group \cite{dmno}. What happens is that, assuming $\Z_n(\EA_n)=\Z_n(\EB_n)$, one can certainly glue the topological phase $\Z_n(\EA_n)$ (with a gapped boundary $\EA_n$) with the phase $\Z_n(\EB_n)$ (with a gapped boundary $\EB_n$) smoothly in the $n+$1D \bulk and create an $n-$1D wall between $\EA_n$ and $\EB_n$ on the boundary. But this $n-$1D wall can be gapless in general. For example, a 3D quantum hall system shares the same $4$D \bulk with $\one_3$, but the domain wall between them is gapless. In general, the Witt classes of closed $n$D topological orders also form a group, still called Witt group \cite{kong-wen}, which measures the discrepancy between the Witt equivalence and the equivalence of their \bulk's.



\begin{expl} {\rm
We give an example of a Morita-closed wall in the toric code model. Consider the toric code model with a smooth boundary, a 1-codimensional wall and a 2-codimensional defect on the boundary depicted in Fig.\,\ref{fig:external-wall}. This model is completely free of frustration. The complete list of mutually commutative stabilizers are given as follows: 
$$
A_{\mathbf{v}}=\sigma_1^x\sigma_2^x\sigma^3\sigma^4, \quad B_{\mathbf{p}}=\sigma_8^z \sigma_{10}^z \sigma_{11}^z\sigma_{12}^z, \quad A_{13,14,15}=\sigma_{13}^x\sigma_{14}^x \sigma_{15}^x
$$
$$
C_{2,5,3|7}=\sigma_2^z\sigma_5^z\sigma_3^z\sigma_7^x,\quad D_{3|7,8,9}=\sigma_3^x\sigma_7^z\sigma_8^z\sigma_9^z, \quad 
Q_{6,17,18,19,20}=\sigma_6^x\sigma_{17}^y\sigma_{18}^z\sigma_{19}^z\sigma_{20}^z.
$$ 
The excitations on the smooth boundary are given by the unitary fusion 1-category $\rep_{\Zb_2}$ \cite{bravyi-kitaev,kitaev-kong}. 
Since the dotted line is an invertible wall that gives the EM-duality \cite{kitaev-kong}, the neighborhood of the plaquette $(17,18,19,20)$ can be viewed as a 0+1D Morita-closed wall between two smooth boundaries. Note that moving an $e/m$-particle around the corner of the edges labeled by $6$ and $17$ (the blue dot) turn it into an $m/e$-particle. Since an $m$-particle condenses on the smooth boundary, all particles condense in the neighborhood of the plaquette $(17,18,19,20)$. Therefore, the Morita-closed wall on the boundary contains no non-trivial excitations, thus can be described by $\one_1=\hilb$. Indeed, $\hilb$ is at the same time an invertible $\rep_{\Zb_2}$-bimodule and gives a non-trivial Morita equivalence between two smooth boundaries $\rep_{\Zb_2}$. This Morita equivalence $\hilb$ between $\rep_{\Zb_2}$ and $\rep_{\Zb_2}$ determines exactly an invertible domain wall (the dotted line) with wall excitations given by $\fun_{\rep_{\Zb_2}|\rep_{\Zb_2}}(\hilb,\hilb)$, which is also equivalent to the EM-duality of the bulk, i.e. a braided auto-equivalence of $Z_2(\rep_{\Zb_2})$.
}
\end{expl}

\begin{figure}[t] 
  \centerline{\includegraphics[scale=1.5]{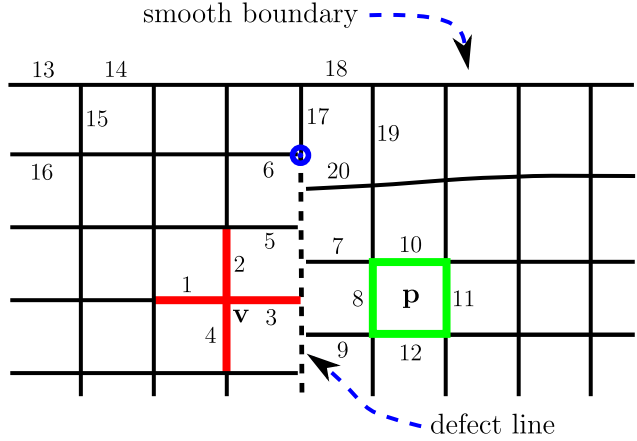}}
  \caption{A Morita-closed wall between two smooth boundaries in the toric code model}
  \label{fig:external-wall}
\end{figure}

\begin{expl} \label{expl:lw-ext-domain-wall}  {\rm
We give some examples of Morita-closed/anomalous walls in Levin-Wen type of models enriched by boundaries and defects depicted in Fig.\,\ref{fig:lw-defect-duality}. Let $\EA_2=\EB_2$ be unitary fusion 1-categories and $\Z_2^{(0)}(\EA_2):=Z_2(\EA_2)$. Fig.\,\ref{fig:lw-defect-duality} depicts a lattice model with a gapped $\EM$-wall between $Z_2(\EA_2)$ and $Z_2(\EB_2)$. The wall excitations are given by unitary fusion 1-category $\Z_2^{(1)}((\EM, a))$, which is defined as follows: 
\be \label{eq:Z-1}
\Z_2^{(1)}((\EM, a)) = Z_2^{(1)}(\EM) := \fun_{\EA|\EB}(\EM, \EM)^\rev.
\ee
An $\EA$-module functor $a\in \fun_\EA(\EA, \EM\boxtimes_\EB \EB)\simeq \EM$ gives a defect junction between $\EA_2$, $\EB_2$ and $\Z_2^{(1)}(\EM)$. If $\EM$ is invertible, then $a$ gives a Morita-closed domain wall between $\EA_2$ and $\EB_2$; if otherwise, $a$ is an anomalous domain wall. When $a$ is viewed as an 1D topological order, it is nothing but $(\EM, a)$. This example shows that, all semisimple indecomposable $\EA$-$\EB$-bimodules are physically realizable as (potentially anomalous) 1D domain walls in Levin-Wen models. Also note that if we fold two boundaries in Fig.\,\ref{fig:lw-defect-duality} upward, we obtain a vertical line with the bottom end given by the 1D topological order $(\EM, a)$. The vertical line is a 2D topological order given by a unitary multi-fusion 1-category $\fun(\EM, \EM)$, which is nothing but the $Z_1^{(0)}$-center of $(\EM_1,a)$.}
\end{expl}

\begin{rema} \label{rema:higher-lw-mod} {\rm
We believe that higher dimensional generalization of Levin-Wen models can be constructed by replacing unitary fusion 1-categories by unitary fusion $n$-categories (see a special case in \cite{walker-wang} and a sketch of general construction in \cite{walker}). We believe that a large part of Example\,\ref{expl:lw-ext-domain-wall} (Fig.\,\ref{fig:lw-defect-duality}) remains to be true for $n>2$. 
}
\end{rema}

\begin{figure}[tb]
 \begin{picture}(150, 90)
   \put(155,10){\scalebox{2}{\includegraphics{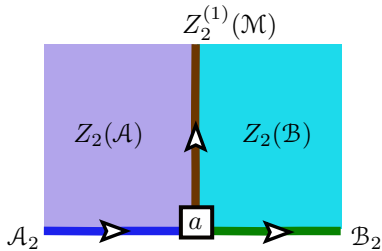}}}
   \put(60,-55){
     \setlength{\unitlength}{.75pt}\put(-18,-19){
     \put(122, 108)       { $\EA_2$}
     \put(295, 108)       { $\EB_2$}
     \put(213,116)      { $ a $}
     \put(155, 160)    { $Z_2(\EA)$}
     \put(240, 160)    { $Z_2(\EB)$}     
     \put(210,215)     { $Z_2^{(1)}(\EM)$}
     }\setlength{\unitlength}{1pt}}
  \end{picture}
\caption{Consider a Levin-Wen model enriched by gapped boundaries and defects \cite{kitaev-kong}. The two bulk lattices are defined by unitary fusion 1-categories $\EA_2$ and $\EB_2$, respectively. The excitations in the left/right bulk are given by the Drinfeld center $Z_2(\EA_2)/Z_2(\EB_2)$ of $\EA_2/\EB_2$. The lattice near the wall between the left bulk and the right bulk is defined by a semisimple indecomposable $\EA$-$\EB$-bimodule $\EM_1$, and is called the $\EM$-wall. The excitations on the $\EM$-wall are given by the unitary fusion 1-category $Z_2^{(1)}(\EM_1):=\fun_{\EA|\EB}(\EM, \EM)^\rev$ \cite{kitaev-kong}. Two boundary lattices are defined by $\EA$ and $\EB$, viewed as a left $\EA$-module and a left $\EB$-module, respectively, and are called the $\EA$-boundary and the $\EB$-boundary. The excitations on the $\EA$-boundary are given by $\EA\simeq \fun_\EA(\EA,\EA)$; those on the $\EB$-boundary are given by $\EB$. The defect junction connecting the $\EA$-boundary, $\EM$-wall and $\EB$-boundary is given by a unitary $\EC$-module functor $a \in \fun_\EA(\EA, \EM\boxtimes_\EB \EB)\simeq \EM$. 
}
\label{fig:lw-defect-duality}
\end{figure}


Let $\EA_n$ and $\EB_n$ be two simple $n$D topological orders. Note that $\Z_n^{(1)}(-)$ defines a surjective map from potentially anomalous (or Morita-closed) domain walls between $\EA_n$ and $\EB_n$ to closed (or invertible) domain walls between $\Z_n(\EA_n)$ and $\Z_n(\EB_n)$. Sometimes, a stronger result can be obtained. When $n=2$, $\EA_2$ and $\EB_2$ are unitary fusion 1-categories. In this case, $\Z_2^{(1)}$ maps bijectively from (potentially anomalous) walls between $\EA_2$ and $\EB_2$ to closed walls between $\Z_2(\EA_2)$ and $\Z_2(\EB_2)$ \cite{dmno,kong-zheng} (see also Thm.\,\ref{prop:Z-functor}). Moreover, Morita-closed domain walls between $\EA_2$ and $\EB_2$ are nothing but invertible $\EA_2$-$\EB_2$-bimodules. When $\EA_2=\EB_2$, they form a group denoted by $\mathrm{Pic}(\EA_2)$. We have the following group isomorphism \cite{eno2009,kitaev-kong}: 
\be \label{eq:Z-(1)}
\Z_2^{(1)}: \mathrm{Pic}(\EA_2) \xrightarrow{\simeq}  \mathrm{Aut}(\Z_2(\EA_2)),
\ee
where $\mathrm{Aut}(\Z_2(\EA_2))$ is the group of invertible walls between $\Z_2(\EA_2)$ and itself, or equivalently, the group of braided auto-equivalences of $Z_2(\EA_2)$. 

\begin{rema} {\rm
The isomorphism in Eq.\,\eqref{eq:Z-(1)} is not an isolated phenomenon. 
In 1D, if $\EA_1$ is a simple algebra over $\Cb$, this result is trivial. In 2D rational conformal field theories, what replaces $\EA_2$ is a simple special symmetry Frobenius algebra $A$ in a modular tensor category. In this case, we also have a group isomorphism $\text{Pic}(A) \simeq \text{Aut}(Z(A))$ \cite{dkr1}, where $Z(A)$ is the full center of $A$, with a similar physical meaning \cite{ffrs-defect}. In the framework of 3D TQFT's, the similar phenomenon was studied recently in \cite{fs,fpsv}. 
}
\end{rema}

 \subsection{The functoriality of $\Z_n$} \label{sec:functor-Z}

In Fig.\,\ref{fig:unique-bulk-hypothesis}, an anomalous domain wall $a_{n-1}$ between $\EA_n$ and $\EB_n$ determines the ``bulk" closed domain wall $\EX_n$, also denoted by $\Z_n^{(1)}(a_{n-1})$. Similarly, we can define $\Z_n^{(2)}(\cdot)$ for (possibly anomalous) walls between (possibly anomalous) walls. For example, if we fatten Fig.\,\ref{fig:unique-bulk-hypothesis} in the third directions, let us imagine an anomalous wall $\alpha_{n-2}$ between the anomalous wall $a_{n-1}$. It is also the gapped boundary of a domain wall $\chi_{n-1}$ between two $\EX_n$-phases. Then this $\chi_{n-1}$ is uniquely determined by $\alpha_{n-2}$ and denoted by $\Z_n^{(2)}(\alpha)$. Similarly, we can define $\Z_n^{(i)}$ for $2<i<n$ (not for instantons $i=n$). In this context, we denote the \bulk $\Z_n$ by $\Z_n^{(0)}$, i.e. $\Z_n=\Z_n^{(0)}$. 

\begin{defn} {\rm
We define the category $\TO_n^{wall}$ to be the $n$-category of simple $n$D topological orders with (potentially anomalous) gapped 1-codimensional walls as 1-morphisms, \ldots, and (potentially anomalous) gapped $l$-codimensional walls between $(l-$$1)$-codimensional walls as $l$-morphisms, \ldots instantons as $n$-morphisms. 
}
\end{defn}

The strong unique-bulk hypothesis suggests that the complete {\bf boundary-bulk relation} is given by the collection of maps $\Z_n=\{ \Z_n^{(0)}, \Z_n^{(1)}, \Z_n^{(2)}, \cdots, \Z_n^{(n-1)}\}$, which defines a functor:
\be  \label{eq:Z_n-functor}
\Z_n: \TO_n^{wall} \to \TO_{n+1}^{closed-wall}.
\ee
Note that the result (\ref{eq:morita}) and the group isomorphism in (\ref{eq:Z-(1)}) are just parts of the first two layers of the hierarchical structures of the functor $\Z_n$. 

\begin{defn} {\rm
A $n+$1D simple topological order (or an $l$-codimensional wall) in $\TO_{n+1}^{closed-wall}$ is called {\it exact} if it lies in the image of the functor $\Z_n$. 
}
\end{defn}

\begin{expl} \label{expl:exact} {\rm
In 2+1D, the $\EM$-walls in Fig.\,\ref{fig:lw-defect-duality} is exact. 
}
\end{expl}

Physically, the funtoriality of $\Z_n$ is tautological if we assume the strong unique-bulk hypothesis. Mathematically, it suggests that the notion of the $Z_n$-center of a unitary fusion $n$-category is functorial if we define the domain/target categories properly. This mathematical funtoriality is highly non-trivial and still conjectural. 

\smallskip
When $n=2$, $\Z_2$ indeed gives a mathematical functor \cite{kong-zheng}. More precisely, let $\EM\EF\mathrm{us}$ be the 1-category of indecomposable multi-fusion 1-categories with 1-morphisms given by the equivalence classes of semisimple bimodules and $\EB\mathrm{rd}^{\mathrm{cl}}$ the 1-category of non-degenerate braided fusion 1-categories ($\EC, \ED$, etc.) with 1-morphisms given by the equivalence classes of closed multi-fusion bimodules (${}_\EC\EE_\ED$), where a closed multi-fusion bimodule is defined by a multi-fusion 1-categories equipped with a $\EC$-$\ED$-bimodule structure induced by a braided monoidal equivalence $Z_2(\EE_2)\xrightarrow{\simeq} \EC_3\boxtimes \overline{\ED}_3$. It is non-trivial to prove that $\EB\mathrm{rd}^{\mathrm{cl}}$ is a well-defined 1-category \cite[Thm.\,5.19]{kong-zheng}. 

\begin{thm} [\cite{kong-zheng}] \label{prop:Z-functor}
For indecomposable multi-fusion 1-categories $\EA_2, \EB_2$ and semisimple $\EA$-$\EB$-bimodule 1-category $\EM_1$, the following assignment
$$
\EA_2 \mapsto Z_2(\EA_2), \quad\quad {}_\EA\EM_{\EB} \mapsto Z_2^{(1)}(\EM_1):=\fun_{\EA|\EB}(\EM, \EM)^\rev 
$$ 
defines a well-defined functor $Z_2:\EM\EF\mathrm{us} \to \EB\mathrm{rd}^{\mathrm{cl}}$.  Moreover, $Z_2$ is fully faithful. 
\end{thm}

\begin{rema} {\rm
In the Language of Jacob Lurie \cite{lurie2}, the monoidal $Z_2(\EA)$-$Z_2(\EB)$-bimodule $Z_2^{(1)}(\EM)$ is an $E_1$-algebra over the $E_2$-algebra $Z_2(\EA)\boxtimes \overline{Z_2(\EB)}$. Thm.\,\ref{prop:Z-functor} is a special case of the lax functoriality of the center of an $E_n$-algebra over an $E_{n+1}$-algebra in a symmetric monoidal $\infty$-category. 
}
\end{rema}

For $n>2$, we expect a similar (but not fully faithful) functoriality of the center $Z_n$ to be true for unitary fusion $n$-categories. Note that $\Z_3$ is not faithful because there are many non-trivial 1-codimensional gapped exact walls in $\TO_3^{wall}$ that are mapped to trivial walls under $\Z_3^{(1)}$ (see Remark\,\ref{expl:exact}).

\begin{figure}[tb]
 \begin{picture}(150, 90)
   \put(155,10){\scalebox{2}{\includegraphics{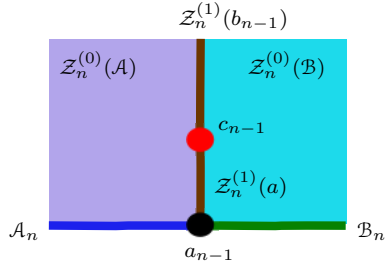}}}
   \put(60,-55){
     \setlength{\unitlength}{.75pt}\put(-18,-19){
     \put(125, 108)       {\footnotesize $\EA_n$}
     \put(300, 108)       {\footnotesize $\EB_n$}
     \put(213,98)      {\footnotesize $ a_{n-1} $}
     \put(228,130)    {\footnotesize $ \Z_n^{(1)}(a) $}
     \put(230,162)    {\footnotesize $ c_{n-1} $}
     \put(150, 190)    {\footnotesize $\Z_n^{(0)}(\EA)$}
     \put(245, 190)    {\footnotesize $\Z_n^{(0)}(\EB)$}     
     \put(210,215)     {\footnotesize $\Z_n^{(1)}(b_{n-1})$}
     }\setlength{\unitlength}{1pt}}
  \end{picture}
\caption{Above picture show intuitively how to define a morphism $a_{n-1} \to b_{n-1}$ for anomalous domain walls $a_{n-1}$ and $b_{n-1}$ between $\EA_n$-phase and $\EB_n$-phase. 
}
\label{fig:wall-morphism}
\end{figure}

\smallskip
We have shown that $\Z_n^{(0)}$ satisfies the same universal property as that of the center. Actually, in Fig.\,\ref{fig:unique-bulk-hypothesis}, the $\Z_n^{(1)}$-\bulk of the domain wall $a_{n-1}$ is also some kind of center. To see this, one can first introduce the notion of a morphism between two domain walls between the $\EA_n$-phase and the $\EB_n$-phase similar to the notion of a morphism between two topological orders. Such a morphism $a_{n-1} \to b_{n-1}$ between domain walls $a_{n-1}$ and $b_{n-1}$ is defined by a pair 
$$
(c_{n-1},\,\,\,\, c_{n-1}\boxtimes_{\Z_n^{(1)}(a)}a_{n-1} \xrightarrow[\simeq]{\phi_{n-2}} b_{n-1}),
$$ 
depicted in Fig.\,\ref{fig:wall-morphism}. Then one can prove that the universal property of $\Z_n^{(1)}$ is the same as that of the center considered in the category of domain walls. We omit details here. Note that, mathematically, the right hand side of Eq.\,(\ref{eq:Z-1}) is indeed a center (or an internal hom) in the category of bimodules. 
Similarly,  we can show that $\Z_n^{(2)}, \Z_n^{(3)}, \cdots$ defined for higher codimensional domain walls are also some kind of centers by proving their universal properties. 
Therefore, the functor $\Z_n$ is indeed the center functor.

\medskip
Moreover, the strong unique-bulk hypothesis implies an ``exact sequence" \cite{kong-wen}: 
$$
\cdots \to \TO_n^{wall} \xrightarrow{\Z_n} \TO_{n+1}^{wall}
\xrightarrow{\Z_{n+1}} \TO_{n+2}^{wall} \to \cdots .
$$
In particular, we expect that the identities $\Z_{n+1}^{(i)}(\Z_n^{(i)}(\cdot))=\one_{n+2-i}$ hold not only for $i=0$ but also for $i=1,2,\ldots,n-1$. The mathematical meaning of these later cases (for $i>0$) have not been discussed before. We briefly discuss the case $i=1$ below.

\medskip
Mathematically, let $\EA_{n+1}$ and $\EB_{n+1}$ be unitary fusion $n$-categories and $(\EM_n, a)$ a (potentially anomalous) wall between $\EA_{n+1}$ and $\EB_{n+1}$, where $\EM_n$ is a unitary $n$-category and an $\EA$-$\EB$-bimodule. In this case, we have 
\be \label{eq:Z-1-1}
Z_{n+1}^{(1)} ((\EM_n, a)) = \fun_{\EA|\EB}(\EM_n, \EM_n)^\rev.
\ee
If $\EA_{n+1}$ and $\EB_{n+1}$ are closed, they are both unitary braided $(n-$$1)$-categories. Let $\EM_n$ be a unitary multi-fusion $\EA$-$\EB$-bimodule (\cite{kong-zheng}). Both of the bulk-to-wall maps $\EA \to \EM$ and $\EB \to \EM$ factor through $Z_n^{(0)}(\EM_n)=\fun_{\EM|\EM}(\EM, \EM)$. 
In this case, we expect that $Z_{n+1}^{(1)} (\EM_n)$ can be equivalently defined as follows: 
\be \label{eq:Z-1-2}
Z_{n+1}^{(1)} (\EM_n) = \fun_{\EM|\EM}^{\EA|\EB}(\EM, \EM)^\rev,
\ee
where the right hand side is the full subcategory of $\fun_{\EM|\EM}(\EM, \EM)$ containing 
those objects that are symmetric to the images of $\EA_{n+1}$ and $\EB_{n+1}$ in $\fun_{\EM|\EM}(\EM, \EM)$. 
The physical meaning of the definition Eq.\,(\ref{eq:Z-1-2}) is quite obvious. Since, in this case, we have
$$
Z_n^{(0)}(\EM_n) = \EA_{n+1} \boxtimes Z_{n+1}^{(1)} (\EM_n) \boxtimes \overline{\EB}_{n+1},
$$
the excitations in $\EA_{n+1}$, $\Z_{n+1}^{(1)} (\EM_n)$ and $\overline{\EB}_{n+1}$ can all be viewed as excitations in $\Z_n^{(0)}(\EM_n)$. It is clear that all three sets of excitations are mutually symmetric. 
\void{
\begin{expl} {\rm 
In Fig.\,\ref{fig:unique-bulk-hypothesis}, if $\EA_n$ and $\EB_n$ are closed, then we have 
$$
\Z_n^{(0)}(\EM) = \EA_n^\op \boxtimes \EX_n \boxtimes \EB_n.
$$ 
In this case, the ``bulk" domain wall $\EX_n=\Z_n^{(1)}((\EM,a_{n-1}))$ can indeed be characterized as the subcategory of $\Z_n^{(0)}((\EM,a_{n-1}))$ consisting of those objects that are symmetric to objects in $\EA_n$ and $\EB_n$. 
}
\end{expl}
}

\begin{rema} {\rm
In \cite{lurie2}, Lurie defined a notion of the center for an $E_k$-algebra over an $E_{k+1}$-algebra. It will be interesting to know if it reduces to Eq.\,(\ref{eq:Z-1-2}) in our cases.  
}
\end{rema}

\begin{expl} {\rm
The functor $\Z_2$ is given by the mathematical functor in Thm.\,\ref{prop:Z-functor} with $\Z_2^{(1)}(\EM_1)=Z_2^{(1)}(\EM_1)$. When we apply the functor $\Z_3$, we obtain
$$
Z_3^{(1)}(Z_2^{(1)}(\EM_1)) = \fun_{Z_2^{(1)}(\EM_1)|Z_2^{(1)}(\EM_1)}^{Z_2(\EA_2)|Z_2(\EB_2)}(Z_2^{(1)}(\EM_1), Z_2^{(1)}(\EM_1))^\rev=\one_3, 
$$
which follows from the fact that $Z_2(\EA_2)\boxtimes \overline{Z_2(\EB_2)} = \fun_{Z_2^{(1)}(\EM)|Z_2^{(1)}(\EM)}(Z_2^{(1)}(\EM), Z_2^{(1)}(\EM))$ and Eq.\,(\ref{eq:Z-1-2}) for $n=3$. 
}
\end{expl}

In general, we expect that the following mathematical result to be true. 
\begin{conj}
Let $\EA_{n+1}$ and $\EB_{n+1}$ be two unitary $n$-fusion categories and $\EM_n$ a unitary indecomposable $\EA$-$\EB$-bimodule (a unitary $n$-category). Let $Z_{n+1}(\EA)$, $Z_{n+1}(\EB)$ and $Z_{n+1}^{(1)}(\EM)$ be $\fun_{\EA|\EA}(\EA,\EA)$, $\fun_{\EB|\EB}(\EB,\EB)$ and $\fun_{\EA|\EB}(\EM, \EM)^\rev$, respectively. Then we have
$$
Z_{n+2}^{(1)}(Z_{n+1}^{(1)}(\EM_n)) = \fun_{Z_{n+1}^{(1)}(\EM)|Z_{n+1}^{(1)}(\EM)}^{Z_{n+1}(\EA)|Z_{n+1}(\EB)}(Z_{n+1}^{(1)}(\EM_n), Z_{n+1}^{(1)}(\EM_n))^\rev=\one_{n+2}.
$$
\end{conj}

\section{Conclusions and outlooks}  \label{sec:conclusion}

Although the main result in this work is the boundary-bulk relation for topological orders in any dimension, another secrete goal of this work is to use this relation as a tool to determine what a proper categorical description of a topological order should be. As far as we can see, the categorical description given in this work seems works pretty well with the boundary-bulk relation. If this categorical description of topological orders indeed works, it suggests something surprising. 
Naively, one expect the higher dimensional theories to be much more complicated than the 3D theory due to the complexity of the fusion-braiding properties of higher dimensional excitations (see for example \cite{wl}). This work, however, suggests that the higher dimensional theories might be similar to the 3D theory. We give two possible generalizations of the 3D theory below. 
\bnu
\item Our theory suggests that the condensation theory in $n+$2D should be similar to that in 3D \cite{kong-anyon}. A closed $n+$2D topological order is given by a unitary braided fusion $n$-category $\EC_{n+2}$ with a trivial $Z_{n+2}$-center, a condensation should be completely determined by a commutative separable algebra $A$ in $\EC_{n+2}$ (possibly satisfying additional conditions). The condensed phase should be described by the unitary braided fusion $n$-category of local $A$-modules, and the gapped domain wall created between the original and condensed phases consists of confined excitations given by the unitary fusion $n$-category $\EC_A$ of $A$-modules in $\EC$. When the condensed phase is trivial, we believe that the bulk-to-boundary functor $L: \EC \to \EC_A$ determines a Lagrangian algebra $L^\vee(A)$ in $\EC$, where $L^\vee$ is the right adjoint functor of $L$, such that the category of local $L^\vee(A)$-modules in $\EC$ is trivial, i.e. $\EC_{L^\vee(A)}^{loc}=\one_{n+2}$. 
 
\item The mathematical description of a closed $n+$2D topological order enriched by invertible 1-codimensional walls might be similar to that in 3D \cite{bbcw,lkw1}. It is very easy to enrich a closed $n+$2D topological order $\EC_{n+2}$, a unitary braided fusion $n$-category, by invertible 1-codimensional walls. We assume that the equivalence classes of these invertible walls are labeled by the elements of a finite group $G$. In other words, $G$ acts on $\EC_{n+2}$ as automorphisms. Adding these invertible 1-codimensional walls to $\EC_{n+2}$, we obtain an extension $\EC_{n+2}^G$ of $\EC_{n+2}$. The $(n+$$2)$-category $\EC_{n+2}^G$ has finitely many equivalence classes of simple 1-morphisms labeled by the elements of $G$ and the identity element $1\in G$ is the identity 1-morphism $\id_\ast$. 
The $(n+$$2)$-category $\EC_{n+2}^G$ is not unitary because the 5th axiom in Def.\,\ref{def:unitary-n-cat} does not hold. In particular, each $\hom(g,h)$ is a non-zero unitary $n$-category for $g,h\in G$ in general. The structures of $\EC_{n+2}^G$ are completely encoded in the substructure $\oplus_{g,h\in G} \hom(g,h)$, which turns out to be a $G$-crossed unitary braided fusion $n$-category. It is the direct higher categorical analogue of a $G$-crossed unitary braided fusion 1-category in 3D (\cite{dgno,bbcw}). In particular, the expression $\EC_{n+2}^G=\oplus_{f\in G} \left( \oplus_{g^{-1}h=f} \hom(g,h) \right)$ gives a $G$-grading on $\EC_{n+2}^G$, and the fusion product $\hom(g,h) \times \hom(g',h')\to \hom(gh,g'h')$ is induced from the tensor product $g\otimes h\simeq gh$ and $g'\otimes h'\simeq g'h'$. As a consequence, we propose that symmetry enriched (by invertible 1-codimensional walls) closed $n+$2D topological orders are classified by $G$-crossed unitary braided fusion $n$-categories. We will give more details in \cite{lkw2}. 

\enu


\appendix

\section{Appendix}

\subsection{The definition of a unitary $n$-category} \label{sec:def-unitary-n-cat}


In this subsection, we give the definition of a unitary $n$-category. Before we start, we first review some elements of an $n$-category (see for example \cite{lurie3}).

\medskip
We assume that a good definition of $n$-category is chosen. We always use the subscript in $\EC_n$ to indict that $\EC$ is an $n$-category. Sometimes the subscript is omitted if it is clear from the context. An object in an $n$-category  $\EC_n$ is also called a $0$-morphism. We also refer to $\EC_n$ itself as the unique $-1$-morphism. In particular, a 0-category $\EC_0$ is just the set of 0-morphisms. A $\Cb$-linear 0-category $\EC_0$ is a vector space over $\Cb$. The opposite category $\EC_n^\op$ is defined by flipping all $n$-morphisms. For a $\Cb$-linear $0$-category $\EC_0$, $\EC_0^\op$ is the dual vector space $\hom_\Cb(\EC_0, \Cb)$. 
For an $n$-category $\EC_n$, we define the homotopy 1-category $\mathrm{h}_1\EC$ by the 1-category with the same set of objects as $\EC$ and 1-morphisms given by the equivalence classes of $1$-morphisms in $\EC$; we define the homotopy 2-category $\mathrm{h}_2\EC$ by the 2-category with the same set of objects and 1-morphisms as $\EC$ and 2-morphisms given by the equivalence classes of $2$-morphisms in $\EC$. 

\begin{defn}  {\rm
Let $\EC_2$ be a 2-category and $f: x \to y$ a 1-morphism. 
\bnu 
\item $f$ is said to have a left adjoint if there exist a 1-morphism $g: y \to x$, the unit 2-morphism $\eta: \id_x \Rightarrow g\circ f$ and the co-unit 2-morphism $\epsilon: f\circ g \Rightarrow \id_y$ such that 
$$
\id_f=(f \xrightarrow{\simeq} f \circ \id_x \xrightarrow{1\eta} f \circ g \circ f \xrightarrow{\epsilon 1} f), \quad\quad
\id_g=(g \xrightarrow{\simeq} \id_x \circ g \xrightarrow{\eta1} g\circ f \circ g \xrightarrow{1\epsilon} g). 
$$

\item $f$ is said to have a right adjoint if there exist a 1-morphism $h: y\to x$, the unit 2-morphism $\tilde{\eta}: \id_y \Rightarrow f\circ h$ and the co-unit 2-morphism 
$\tilde{\epsilon}: h\circ f \Rightarrow \id_x$ such that 
$$
\id_f=(f \xrightarrow{\simeq} \id_y \circ f  \xrightarrow{\tilde{\eta}1} f \circ h \circ f \xrightarrow{1\tilde{\epsilon}} f), \quad\quad
\id_h=(h \xrightarrow{\simeq} h\circ \id_y \xrightarrow{1\tilde{\eta}} h\circ f \circ h \xrightarrow{\tilde{\epsilon}1} h). 
$$
\enu
$\EC_2$ is said to {\it have adjoints for 1-morphisms} if every 1-morphism in $\EC_2$ has both a left and a right adjoint. 
}
\end{defn}

\begin{defn} {\rm
Let $\EC_n$ be an $n$-category for $n\geq 2$. $\EC$ has adjoints for $1$-morphisms if $\mathrm{h}_2\EC$ has adjoints for $1$-morphisms. For $1<k<n$, $\EC$ has adjoints for $k$-morphisms if, for any pair of objects $x, y\in \EC$, the $(n-1)$-category $\hom_\EC(x,y)$ has adjoints for all $(k-1)$-morphisms. $\EC$ is said to have adjoints if $\EC$ has adjoints for $k$-morphisms for all $0<k<n$. 
}
\end{defn}

A monoidal $n$-category $\EC_n$ can be defined by an $(n+1)$-category with a single object $\ast$. Namely, the $n$-category $\hom(\ast, \ast)$ is an $n$-category with a monoidal structure. A monoidal $n$-category $\EC_n$ is said to have duals for the objects if the homotopy 1-category $\mathrm{h}_1\EC$ is a rigid monoidal category. Then $\EC_n$ is said to have duals if $\EC_n$ has duals for objects and adjoints for $k$-morphisms for all $0<k<n$. It is equivalent to $\EC$ having adjoints when it is viewed as an $(n+1)$-category with a single object.

\begin{defn} \label{def:finite-cat} {\rm
A $1$-category $\EC_1$ is called {\it finite} if $\EC$ is abelian, $\Cb$-linear and has finitely many simple objects and every object is the product of finitely many simple objects. Here, an object $x$ is simple if $\hom(x,x)=\Cb$. 
}
\end{defn}

\begin{defn}   \label{def:unitary-n-cat}  {\rm
For $n\geq 0$, an $n$-category $\EC_n$ is {\it unitary} if the following conditions are satisfied:
\bnu

\item $\EC_n$ is {\it $\Cb$-linear}. That is, for every pair of $(n-1)$-morphisms $f,g:X\to Y$ in $\EC$, $\hom(f,g)$ is a finite dimensional vector space over $\Cb$. The compositions of the morphisms in $\EC$ respect this linear structure.

\item $\EC_n$ is {\it finite}. That is, $\EC$ is closed under finite products, has finitely many simple objects, every object is the product of finitely many simple objects and $\hom(i,j)$ is a finite $(n-1)$-category for simple objects $i$ and $j$. Here, an $(n-1)$-morphism $x$ is {\it simple} if $\hom(x,x)=\Cb$, and a $k$-morphism $y$, for $0\le k<n-1$, is {\it simple} if the $(k+1)$-identity morphism $\id_y$ is simple.

\item $\EC_n$ {\it has adjoints}. That is, every $k$-morphism, $1\le k<n$, has both a left adjoint and a right adjoint.

\item There is an equivalence $\delta:\EC_n\to\EC_n^\op$ which fixes all $k$-morphisms for $0\le k<n$, and is antilinear, involutive and positive on $n$-morphisms, i.e.
\begin{equation}  \label{eq:delta}
\delta(\lambda f) = \bar{\lambda} \delta(f), \quad\quad \delta\delta(f) = f, \quad\quad f\circ\delta(f)=0 \Rightarrow f=0,
\end{equation}
for $n$-morphism $f: X \to Y$ and $\lambda \in \Cb$.

\item For two non-isomorphic simple $k$-morphisms
$i^{[k]}$ and $j^{[k]}$, $\hom(i^{[k]},j^{[k]})=0_{n-k-1}$, where $0_{n-k-1}$ is the zero $(n-k-1)$-category that has only the zero object and zero morphisms. 
  
\enu
}
\end{defn}

\begin{rema} {\rm
Note that $\hom(x^{[k]},y^{[k]})$ is a unitary $(n-k-2)$-category for two $k$-morphisms $x,y\in\EC_n$. The physical meaning of the 5th axiom in Def.\,\ref{def:unitary-n-cat} was explained in \cite{kong-wen}. 
}
\end{rema}

\begin{rema} {\rm
If $\EC_n$ is a unitary $n$-category, then $\EC_n$ has a terminal object $0$ as the zero product by definition. Since $\hom(x,0)\simeq\hom(0,x)$ by duality, we see that $0$ is a zero object. More generally, finite coproducts in $\EC$ coincide with finite products by duality. In this case, such a (co-)product is also called a {\it direct sum}. In a unitary $n$-category, the direct sum $x\oplus y$ (or the coproduct) of two objects $x$ and $y$ is characterized by the property that
$\hom(z, x\oplus y) \simeq \hom(z, x) \oplus \hom(z, y)$
as $(n-1)$-categories for all $z$. 
}
\end{rema}

\begin{prop}
Let $\EC_n$ be a unitary $n$-category containing a $k$-morphism $f:x\to y$, $1\le k<n$. Then the left adjoint and the right adjoint of $f$ are canonically isomorphic.
\end{prop}

\pf
Let $g:Y\to X$ be the left adjoint of the $k$-morphism $f$ with unit $u:\id_y\to f\circ g$ and counit $v:g\circ f\to\id_x$. For $1\leq k<n$, the left (or right) adjoint of $u$ and $v$, i.e. $u^\vee:f\circ g\to\id_y$ and $v^\vee:\id_x\to g\circ f$ exhibit $g$ as the right adjoint of $f$. 
When $k=n-1$, $\delta(u)$ and $\delta(v)$ exhibit $g$ as the right adjoint of $f$. 
\epf

\begin{expl} {\rm
Def.\,\ref{def:unitary-n-cat} of a unitary $n$-category is heavily loaded. We unravel this definition in a few lower dimensional cases. Recall Sec.\,\ref{sec:time-reverse} for the definition of $\EC^\op$. 

\bnu

\item 
When $n=0$, $\EC_0^\op$ is obtained from $\EC_0$ by flipping the arrows. It makes sense if we interpret the vector space $\EC_0$ over $\Cb$ as $\hom_\Cb(\Cb, \EC_0)$. As a result, $\EC_0^\op=\hom_\Cb(\EC_0, \Cb)$. Then the conditions in (\ref{eq:delta}) guarantee that $\EC_0$ is a finite dimensional Hilbert space. 

\item When $n=1$, we recover the usual definition of a unitary 1-category (see for example \cite{mueger4}), which is defined as an abelian $\Cb$-linear finite $\ast$-category, where a $\ast$-category means a family of maps $\ast:\hom_\EC(x,y) \to \hom_\EC(y,x)$ (given by $\delta$) such that 
$$(g \circ f)^\ast = f^\ast \circ g^\ast, \quad\quad (\lambda f)^\ast = \bar{\lambda} f^\ast,\quad\quad f^{\ast\ast} = f,\quad\quad \forall \lambda \in \Cb^\times,
$$
satisfying the positivity condition $f\circ f^\ast =0 \Rightarrow f=0$. Note that $\ast = \delta(-)$. 

\item When $n=2$, a unitary $2$-category is a $\Cb$-linear finite $2$-category having adjoints for 1-morphisms such that all hom spaces are unitary $1$-categories and all coherence isomorphisms are unitary, i.e. 
\be  \label{eq:unitarity-coherence}
\delta(\alpha_{f,g,h}) = \alpha_{f,g,h}^{-1}, \quad \delta(l_f) = l_f^{-1}, \quad \delta(r_f) = r_f^{-1}, \quad
\delta(\eta)= \tilde{\epsilon}, \quad \delta(\epsilon)=\tilde{\eta},
\ee
for $1$-morphisms $f,g,h$, where $\alpha$, $l/r$ and $\eta, \tilde{\eta}, \epsilon, \tilde{\epsilon}$ are the associator, the left/right unit isomorphism and duality maps, respectively. Note that Eq.\,(\ref{eq:unitarity-coherence}) follows from the definition of coherence isomorphisms in the opposite category (see Sec.\,\ref{sec:time-reverse}). 

\item When $n>2$, the second axiom in Def.\,\ref{def:unitary-n-cat} encodes the unitarity 
of all coherence isomorphisms. For example, a unitary 3-category with a unique simple object $\ast$ and a unique simple $1$-morphism $\id_\ast$ (in $\hom(\ast, \ast)$) is just the usual unitary braided fusion 1-category, where the braiding is unitary, i.e. $\delta(c_{x,y}) = c_{x,y}^{-1}$ for $x\otimes y \xrightarrow{c_{x,y}} y\otimes x$. 
\enu
}
\end{expl}

\begin{defn} \label{def:unitary-n-functor} {\rm
An $n$-functor $\EF: \EC_n\to \ED_n$ between two unitary $n$-categories $\EC$ and $\ED$ is called {\it unitary} if $\EF$ is $\Cb$-linear for $n$-morphisms and $\EF \circ \delta = \delta \circ \EF$. 
}
\end{defn}

\subsection{The universal property of the \bulk~with higher morphisms} \label{sec:universal-higher-morphism}

In this subsection, we describe the universal property of the \bulk with higher isomorphisms. We drop the Minimal Assumption (see Remark\,\ref{rema:minimal}).

\medskip
Now the action $\rho: P_n(\Z_n(\EC_n)) \boxtimes \EC_n \to \EC_n$ is unital in the sense that there is a 2-isomorphism $\gamma: \rho \circ (\iota_{\Z_n(\EC_n)} \boxtimes \id_{\EC_n}) \overset{\simeq}{\Rightarrow} \id_{\EC_n}$, which is equivalent to the commutativity (up to a 2-isomorphism) of the following diagram:
\be  \label{diag:univ-property-1}
\raisebox{2em}{
\xymatrix@R=0.5em@C=2.5em{
& P_n(\Z_n(\EC_n)) \boxtimes \EC_n \ar[rdd]^\rho &  \\
& \Downarrow \,\, \gamma &  \\
\EC_n \ar[rr]^{\id_{\EC_n}} \ar[uur]^{\iota_{P_n(\Z_n(\EC_n))} \boxtimes \id_{\EC_n}} & & \EC_n\, .
}}
\ee
Then the universal property of the \bulk becomes much more complicated. A complete definition was given by Lurie \cite{lurie2}. We illustrate the first two layers of the structures below. 

\bigskip
\noindent {\bf Universal property of the \bulk}: {\it 
The triple $(P_n(\Z_n(\EC_n)), \rho, \gamma)$ is terminal among all such triples. More precisely, if $(\EX_n, f,\phi)$ is such a triple, where $\EX_n$ is an $n$D topological order and $f: \EX_n \boxtimes \EC_n \to \EC_n$ a morphism such that the following diagram:
\be  \label{diag:X-n}
\raisebox{2em}{
\xymatrix@R=0.5em@C=2.5em{
& \EX_n \boxtimes \EC_n \ar[rdd]^f &  \\
& \Downarrow \,\, \phi &  \\
\EC_n \ar[rr]^{\id_{\EC_n}} \ar[uur]^{\iota_{\EX_n} \boxtimes \id_{\EC_n}} & & \EC_n
}}
\ee
is commutative up to a 2-isomorphism $\phi: f\circ (\iota_{\EX_n} \boxtimes \id_{\EC_n}) \Rightarrow \id_{\EC_n}$, then there is a morphism between the triangles (\ref{diag:univ-property-1}) and (\ref{diag:X-n}) given by a triple $(\underline{f}, \alpha_{\underline{f}},\Phi_{\underline{f}})$, where $\underline{f}: \EX_n \to \Z_n(\EC_n)$ is a 1-morphism and  $\alpha_{\underline{f}}: \rho \circ (\underline{f} \boxtimes \id_{\EC_n}) \Rightarrow f$ a 2-isomorphism such that the following diagram:
\be  \label{diag:univ-prop}
\raisebox{4em}{
\xymatrix@R=0.4em@C=1em{
&  & P_n(\Z_n(\EC_n)) \boxtimes \EC_n \ar@/^2.8pc/[rrdddd]^\rho &  & \\
& & &  & \\
& \iota\boxtimes\id \Rightarrow & \EX_n \boxtimes \EC_n \ar[uu]^{\underline{f} \boxtimes \id_{\EC_n}} \ar[rrdd]^f & \alpha_{\underline{f}} \Downarrow & \\
& & \Downarrow \,\, \phi & & \\
\EC_n \ar[uurr]^{\iota_{\EX_n} \boxtimes \id_{\EC_n}} \ar@/^2.8pc/[uuuurr]^{\iota_{P_n(\Z_n(\EC_n))} \boxtimes \id_{\EC_n}}  \ar[rrrr]^{\id_{\EC_n}} & & & & \EC_n\, .
}}
\ee
is commutative up to a 3-isomorphism $\Phi_{\underline{f}}: \phi \circ (\id\boxtimes\alpha_{\underline{f}}) \circ (\iota_{\underline{f}\circ\iota_{\EX_n}}\boxtimes\id) \Rrightarrow \gamma$. Moreover, the triple $(\underline{f}, \alpha_{\underline{f}},\Phi_{\underline{f}})$ is unique in the sense that, if $(g,\alpha_g,\Phi_g)$ is another triple, then there is an 2-isomorphism $\beta^{(1)}: \underline{f}\Rightarrow g$ such that the diagram consisting of $\beta^{(1)}$, $\alpha_{\underline{f}}$ and $\alpha_g$ is commutative up to a 3-isomorphism $\beta^{(2)}$ such that the diagram consisting of
$\beta^{(2)}$, $\Phi_{\underline{f}}$ and $\Phi_g$ is commutative up to a 4-isomorphism $\beta^{(3)}$. And the triple $(\beta^{(1)}, \beta^{(2)}, \beta^{(3)})$ is unique in a similar sense, so on and so forth. 
}

\medskip
We will not prove above universal property of the \bulk in this work as a complete exposition would lead us too far away, so we do not go further on this subject here.

\subsection{Weak morphisms} \label{sec:weak-n-morphism}

There are more general physical realization of a universal process of mapping excitations in the $\EC_n$-phase to the $\ED_n$-phase. For this reason, we would like to introduce the notion of a {\it weak morphism} between two topological orders.
\begin{defn}  \label{def:morphism-1} {\rm
A weak morphism $f$ from an $n$D topological order $\EC_n$ to another one $\ED_n$ is a triple $(f_{n+1}^{(1)}, f_{n}^{(2)}, f_{n-1}^{(3)})$ such that
\bnu
\item $f_{n+1}^{(1)}$ is an $n+$1D (potentially anomalous) topological order,
\item $f_n^{(2)}$ is a gapped domain wall between $f_{n+1}^{(1)}$ and other phases, 
\item $f_{n-1}^{(3)}: f_n^{(2)}\boxtimes_{f_{n+1}^{(1)}} \EC_n \xrightarrow{\simeq} \ED_n$ is an isomorphism.
\enu
}
\end{defn}

\begin{expl}  {\rm
We give some examples.
\begin{enumerate}

\item When $\ED_n=\EC_n$, the identity weak morphism $\id_\EC$ is defined by
\be  \label{eq:id-2}
\id_\EC :=(\one_{n+1}, \one_n, \id_{\EC_n}).
\ee

\item When $a: \EC_n \to \ED_n$ is an isomorphism, it can also expressed as a triple $a=(\one_{n+1}, \one_n, a)$. Namely, the information of an isomorphism is completely encoded in the third component of the triple, i.e. $a=a_{n-1}=a_{n-1}^{(3)}$.

\item For each $n$D topological order $\EC_n$, there is a natural {\it unit weak morphism} $\one_n \xrightarrow{\iota_\EC} \EC_n$ given by
\be  \label{eq:iota}
\iota_\EC = (\one_{n+1}, \EC_n, \one_n\boxtimes \EC_n=\EC_n \xrightarrow{\id_\EC} \EC_n).
\ee


\end{enumerate}
}
\end{expl}

There are morphisms between two morphisms.
\begin{defn} \label{def:higher-morphism} {\rm
Let $f$ and $g$ be two weak morphisms from $\EC_n$ to $\ED_n$.
A 2-morphism $\phi: f \Rightarrow f'$ from $f$ to $f'$,
\void{
sometimes denoted by a diagram,
$$
\raisebox{1em}{
\xymatrix@R=1em@C=0.2em{
\EC_n \ar@/^1.2pc/[rr]^f \ar@/_1.2pc/[rr]_{g} & \Downarrow \phi &  \ED_n
}}
$$
}
is a pair of morphisms $(a,b)$, where
\bnu
\item $a: f_{n+1}^{(1)} \to g_{n+1}^{(1)}$ is a morphism;
\item $b: f_n^{(2)} \to g_n^{(2)}$ is a morphism.
\enu
such that they can form the following triangle-shaped physical configuration
\be  \label{eq:higher-weak-morphism}
\raisebox{-60pt}{ \begin{picture}(140, 124)
   \put(-10,0){\scalebox{2}{\includegraphics{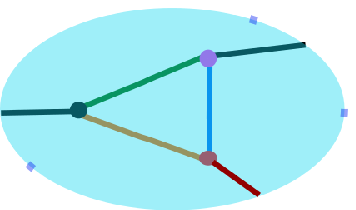}}}
   \put(31,26){
     \setlength{\unitlength}{.75pt}\put(-18,-19){
     \put(0, 73)        { $ f_n^{(2)} $}
     \put(104, 5)       { $b_n^{(2)}$}
     \put(114, 114)      { $\EC_n$}
     \put(38, 26)     { $ b_{n+1}^{(1)} $}
     \put(65, 59)     { $ a_{n+2}^{(1)} $}
     \put(38, 92)     { $ f_{n+1}^{(1)} $}
     \put(126, 60)     { $a_{n+1}^{(2)}$ }
     }\setlength{\unitlength}{1pt}}
  \end{picture}}
\ee
in which all the $(n+2)$-phases are not necessarily closed. $\phi$ is called {\it closed} if $a_{n+2}^{(1)}$ is closed.
}
\end{defn}

\begin{rema} {\rm
The isomorphisms $f_{n-1}^{(3)}$, $g_{n-1}^{(3)}$, $a_n^{(3)}$ and $b_{n-1}^{(3)}$ is not explicit shown in (\ref{eq:higher-weak-morphism}). Their roles can be seen by
squeezing the ``triangle" in (\ref{eq:higher-weak-morphism}) to a single ``point", which is the $n$D $\ED_n$-phase. This collapsing process can be described by the composition of the following isomorphisms:
$$
(b_n^{(2)} \boxtimes_{b_{n+1}^{(1)}} f_n^{(2)}) \boxtimes_{a_{n+1}^{(2)} \boxtimes_{a_{n+2}^{(1)}} f_{n+1}^{(1)}}   \EC
 \xrightarrow{\simeq}
(b_n^{(2)} \boxtimes_{b_{n+1}^{(1)}} f_n^{(2)}) \boxtimes_{g_{n+1}^{(1)}} \EC
\xrightarrow{\simeq} g_n^{(2)} \boxtimes_{g_{n+1}^{(1)}}  \EC \xrightarrow{\simeq} \ED,
$$
in which the first isomorphism is induced by $a_n^{(3)}$ and the second one by $b_{n-1}^{(3)}$, and the last one is defined by $g_{n-1}^{(3)}$. By choosing a different way of collapsing, we obtain the following identity:
$$
b_n^{(2)} \boxtimes_{ b_{n+1}^{(1)} \boxtimes_{a_{n+2}^{(1)}} a_{n+1}^{(2)} } \ED_n\simeq b_n^{(2)} 
\boxtimes_{ b_{n+1}^{(1)} \boxtimes_{a_{n+2}^{(1)}} a_{n+1}^{(2)} } (f_n^{(2)}\boxtimes_{f_{n+1}^{(1)}} \EC_n) \simeq \ED_n,
$$
where the first isomorphism is defined by $(f_{n-1}^{(3)})^{-1}$ and
the second one is due to the independence of how we collapse the ``triangle". Above two mathdisplays summarize what roles the isomorphisms $f_{n-1}^{(3)}$, $g_{n-1}^{(3)}$, $a_n^{(3)}$ and $b_{n-1}^{(3)}$ play in the collapsing processes.
}
\end{rema}

\begin{rema}  {\rm
Weak morphisms are not composable in general. Two morphisms between weak morphisms $\phi: f \Rightarrow g$ and $\psi: g\Rightarrow h$ cannot be composed either.
It is reasonable because physical realizations are not composable in general.}
\end{rema}

If $f: \EC_n \to \ED_n$ be a morphism, then the triple $(\Z_n(\EC_n), f^{(0)}, f^{(1)})$ is automatically a weak morphism, which is universal in a certain sense (see Prop.\,\ref{prop:univ-weak-morphism}). In this case, a 2-isomorphism $\phi$ between two morphisms is automatically a 2-morphism between two weak morphisms.


\medskip
\begin{prop} \label{prop:univ-weak-morphism}
Let $f: \EC_n \to \ED_n$ be a weak morphism between two simple topological orders $\EC_n$ and $\ED_n$. There is a morphism $\tf: \EC_n \to \ED_n$ such that there is a unique closed 2-morphism from $f$ to $\tf$. 
\end{prop}
\begin{proof}
Let $f=(f_{n+1}^{(1)}, f_n^{(2)}, f_{n-1}^{(3)})$ be a weak morphism $\EC_n \to \ED_n$. Using the strong unique-bulk hypothesis, find the unique \bulk of the physical configuration of $f_n^{(2)} \boxtimes_{f_{n+1}^{(1)}} \EC_n$ depicted in Fig.\,\ref{fig:from-pre-to-morphism} (a). Then we fold the vertical box anti-clockwise while keep the position of the ``line segment" $[f_n^{(2)}, f_{n+1}^{(1)}, \EC_n]$ fixed. Then we obtain
an $(n+1)$D physical configuration as shown in Fig.\,\ref{fig:from-pre-to-morphism} (b), in which the $\EC_n$-phase remains the same but now becomes a gapped boundary of its \bulk $\Z_n(\EC_n)$, $\EE_n=P_n(\EY_{n+1}) \boxtimes_{\Z_{n+1}(f^{(1)})} f_n^{(2)}$ and $\Z_n(\ED_n)= \EX_{n+1} \boxtimes_{\Z_{n+1}(f^{(1)})} \EY_{n+1}$.
Moreover, we must have $\EE_n \boxtimes_{\Z_n(\EC_n)} \EC_n \simeq \ED_n$. Indeed, the dimensional reduction is independent of how we squeezing the configuration in detail. So one can choose to squeeze the topless box in (a) horizontally, and obtain a vertical ``line" given by $\Z_n(\ED_n)$ with the bottom end given by $\ED_n$. Therefore, the identity $\EE_n \boxtimes_{\Z_n(\EC_n)} \EC_n \simeq \ED_n$ must hold. Then there is a morphism $\tf=(\EE_n, \tf_{n-1}^{(1)})$ such that the above dimensional reduction process defines a closed 2-morphism from $(f_{n+1}^{(1)}, f_n^{(2)}, f_{n-1}^{(3)})$ to $(\Z_n(\EC_n), \EE_n, \tf_{n-1}^{(1)})$. This closed 2-morphism is unique by the strong unique-bulk hypothesis.
\end{proof}

\begin{figure}[t]
$$
 \begin{picture}(150, 90)
   \put(40,10){\scalebox{2}{\includegraphics{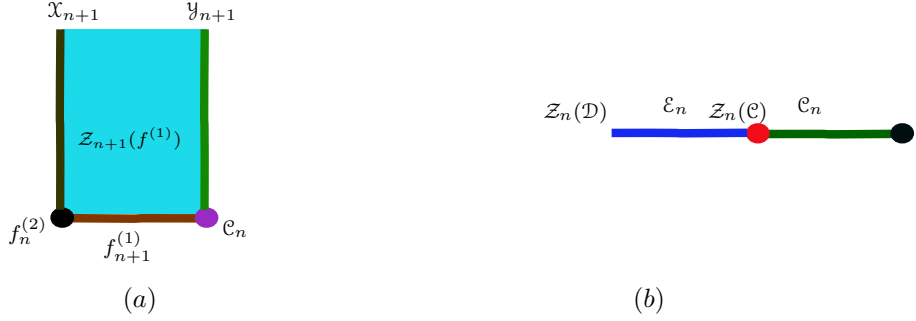}}}
   \put(0,-55){
     \setlength{\unitlength}{.75pt}\put(-18,-19){
     \put(50, 102)       {\footnotesize $f_n^{(2)}$}
     \put(98, 95)     {\footnotesize $ f_{n+1}^{(1)} $}
     \put(158, 105)     {\footnotesize $\EC_n$ }
     \put(85, 150)    {\scriptsize $\Z_{n+1}(f^{(1)})$}
     \put(70, 215)     {\footnotesize $\EX_{n+1}$}
     \put(140,215)    {\footnotesize $\EY_{n+1}$}
     }\setlength{\unitlength}{1pt}}
  \end{picture}
  \quad\quad\quad
  \quad\quad
   \begin{picture}(140, 90)
   \put(20,42){\scalebox{2}{\includegraphics{pic-def-morphism-2.eps}}}
   \put(10,5){
     \setlength{\unitlength}{.75pt}\put(-18,-19){
     \put(123,84)      {\footnotesize $\Z_n(\EC)$}
     \put(40, 84)     {\footnotesize $ \Z_n(\ED) $}
     \put(168, 87)     {\footnotesize $\EC_n$}
     \put(100, 87)     {\footnotesize $\EE_n$}
     }\setlength{\unitlength}{1pt}}
  \end{picture}
$$
$$
(a) \quad\quad\quad\quad\quad\quad\quad\quad\quad\quad\quad\quad\quad\quad
\quad\quad\quad\quad(b)
$$
  \caption{{\small We start from a physical configuration depicted in (a), in which $f_n^{(2)} \boxtimes_{f_{n+1}^{(1)}} \EC_n \simeq \ED_n$. Both $\EX_{n+1}$ and $\EY_{n+1}$ are simple and uniqueness fixed by the unique-bulk hypothesis. By folding the topless box anti-clockwise to the horizontal line, we obtain the physical configuration in (b).}}
  \label{fig:from-pre-to-morphism}
\end{figure}

\end{document}